\documentclass[reqno]{amsart}
\usepackage{amsmath, amsfonts}
\usepackage{amsthm}
\usepackage{amssymb}
\usepackage{dsfont}
\usepackage{amscd}
\usepackage{color}
\usepackage{graphicx, hyperref}
\usepackage{enumitem}

\newcommand{\bigO}{\mathcal{O}}
\newcommand{\intZ}{\mathbb{Z}}
\newcommand{\realR}{\mathbb{R}}
\newcommand{\compC}{\mathbb{C}}
\newcommand{\resp}{resp.}
\newcommand{\ie}{i.e.}
\newcommand{\bbO}{\mathbb{O}}
\newcommand{\bbI}{\mathbb{I}}
\newcommand{\lHopital}{l'H\^{o}pital}
\newcommand{\Pineiro}{Pi\~{n}eiro}
\newcommand{\dequal}{\mathop{=}\limits^{\rm d}}

\DeclareMathOperator{\arsinh}{arsinh}
\DeclareMathOperator{\pv}{p.v.}
\DeclareMathOperator{\local}{local}
\DeclareMathOperator{\diag}{diag}
\DeclareMathOperator{\dist}{dist}
\DeclareMathOperator{\ePDF}{ePDF}
\DeclareMathOperator{\curved}{curved}
\DeclareMathOperator{\vertical}{vertical}
\DeclareMathOperator{\FC}{F-C}
\DeclareMathOperator{\JFC}{J-M-B}

\newtheorem{theorem}{Theorem}[section]
\newtheorem{corollary}[theorem]{Corollary}
\newtheorem{lemma}[theorem]{Lemma}
\newtheorem{prop}[theorem]{Proposition}
\theoremstyle{definition}

\theoremstyle{remark}
\newtheorem{remark}[theorem]{Remark}

\numberwithin{equation}{section}

\def\be{\begin{equation}}
\def\ee{\end{equation}}
\def\ba{\begin{eqnarray*}}
\def\ea{\end{eqnarray*}}
\def\bae{\begin{eqnarray}}
\def\eae{\end{eqnarray}}
\def\bc{\begin{center}}
\def\ec{\end{center}}

\begin{document}

\title[Muttalib--Borodin ensembles]{Muttalib--Borodin ensembles in random matrix theory --- realisations and
correlation functions}

\author{Peter J. Forrester} \address{Department of Mathematics and Statistics, The University of Melbourne, Victoria 3010, Australia; ARC Centre of Excellence for Mathematical \& Statistical Frontiers
}\email{p.forrester@ms.unimelb.edu.au}

\author{Dong Wang} \address{Department of Mathematics, National University of Singapore, Singapore, 119076}\email{matwd@nus.edu.sg}

\date{\today}


\begin{abstract}Muttalib--Borodin ensembles are characterised by the pair interaction term in the eigenvalue probability density
function being of the form $\prod_{1 \le j < k \le N}(\lambda_k - \lambda_j)  (\lambda_k^\theta - \lambda_j^\theta)$.
We study the Laguerre and Jacobi versions of this model --- so named by the form of the one-body interaction terms --- and
show that for $\theta \in \mathbb Z^+$ they can be realised as the eigenvalue PDF of certain random matrices with Gaussian entries. For general $\theta > 0$, realisations in terms of the  eigenvalue PDF of ensembles involving triangular matrices
are given. In the Laguerre case this is a recent result due to Cheliotis, although our derivation is different.
We make use  of a generalisation of a double contour integral formula
for the correlation functions contained in a paper by Adler, van Moerbeke and Wang to analyse the global density (which
we also analyse by studying characteristic polynomials), and the hard edge scaled correlation functions.  
For the global density functional equations for the corresponding resolvents are obtained; solving this gives the moments
in terms of Fuss--Catalan numbers (Laguerre case --- a known result) and particular binomial coefficients (Jacobi case).
For $\theta \in \mathbb Z^+$ the Laguerre and Jacobi cases are closely related to the squared singular values
for products of $\theta$ standard Gaussian random matrices, and truncations of unitary matrices, respectively.
At the hard
edge the double contour integral formulas provide a double contour integral form of the scaled correlation kernel obtained by Borodin in terms
of Wright's Bessel function. 
\end{abstract}

\maketitle
\section{Introduction}
Recent studies in random matrix theory \cite{Claeys-Romano14,Kuijlaars-Stivigny14,Cheliotis14,Borot-Guionnet-Kozlowski13, Forrester-Liu14} have drawn renewed attention to the class of
eigenvalue probability density functions (PDFs) proportional to
\begin{equation}\label{1.1}
\prod_{l=1}^N e^{- V(\lambda_l)} \prod_{1 \le j < k \le N} (\lambda_k - \lambda_j)  (\lambda_k^\theta - \lambda_j^\theta) , \qquad
\lambda_l > 0.
\end{equation}
These PDFs were proposed by Muttalib \cite{Muttalib95} in the context of a simplified model of the joint distribution of the
transmission eigenvalues for disordered conductors in the metallic regime, and with no time reversal symmetry.
The latter is known to have its exact form proportional to \cite{Beenakker-Rejaei93}
\begin{equation}\label{1.2}
\prod_{l=1}^N e^{- V(\lambda_l)} \prod_{1 \le j < k \le N} (\lambda_k - \lambda_j) 
\Big ( \arsinh^2\lambda_k^{1/2} -  \arsinh^2\lambda_j^{1/2}  \Big ), \qquad \lambda_l > 0,
\end{equation}
where for large $\lambda$, $V(\lambda) =  Nc  \arsinh^2\lambda^{1/2} (1 + \bigO(N^{-1}))$, with $c= \ell/L$, $\ell$ denoting the mean
free path length and $L$ the length of the wire. In practice one has $1 \ll c \ll N$.

Recalling that $  {\rm arsinh} \, z = \log (z + \sqrt{z^2 + 1} )$ one sees  (\ref{1.1}) relates to (\ref{1.2}) in the limit
$\theta \to 0^+$ when we have
\begin{equation}\label{1.3}
{1 \over \theta} (\lambda_k^\theta - \lambda_j^\theta) \to  (  \log \lambda_k - \log \lambda_j ),
\end{equation}
although this is still only an approximation to the corresponding factor in (\ref{1.2}). 
Actually in \cite{Muttalib95} attention was restricted to $\theta$ a positive integer; on this
point we remark that the change of variables 
\begin{equation}\label{Cv}
\lambda \to \lambda^{1/\theta}
\end{equation}
 maps $\theta$ to $1/\theta$ in (\ref{1.1})
at the expense of altering $V(\lambda)$.

Our interest is two special cases of (\ref{1.1}).  The first is when
\begin{equation}\label{L}
  e^{-V(\lambda)} = \lambda^c e^{-\lambda},  \qquad  \lambda > 0, \: \: c > -1.
\end{equation}
This is referred to as the Laguerre weight, due to its appearance as the weight function in the orthogonality of the Laguerre
polynomials in the theory of classical orthogonal polynomials. The choice (\ref{L}), together with the choice of
$e^{-V(\lambda)}$ as a Gaussian or Jacobi weight (for the latter see (\ref{J}) below), was considered in some detail by 
Borodin \cite{Borodin99}. Due to the
significant advancement contained in \cite{Borodin99}, we will refer to the general class of PDFs (\ref{1.1}) as Muttalib--Borodin
ensembles, and the particular choice of weight (\ref{L}) in (\ref{1.1}) as the Laguerre Muttalib--Borodin ensemble.

Let us now describe our results. In relation to the Laguerre Muttalib--Borodin ensemble, we first realise the special cases  $c, \theta \in \intZ^+$ as a particular class of complex Wishart matrices isolated in \cite{Adler-van_Moerbeke-Wang11}.
For general parameters we give a new derivation of a recent result of Cheliotis
\cite{Cheliotis14} which gives a realisation in terms of a particular class of random upper-triangular matrices,
and we furthermore develop working in \cite{Adler-van_Moerbeke-Wang11} to generalise this result  (Section \ref{sec:realisations_and_extensions}).  The differential equation satisfied by the characteristic polynomial of the ensemble under the mapping (\ref{Cv}) is studied, and we relate this to the resolvent and global density (Section \ref{sec:characteristic_poly}). We use results contained in \cite{Adler-van_Moerbeke-Wang11},
and take inspiration from the recent work \cite{Liu-Wang-Zhang14}, to obtain a derivation of the global density directly from a double contour formula for the one-point function using the saddle point method (Section \ref{sec:saddle_pt}). Furthermore the double contour integral form of the correlation kernel given in \cite{Adler-van_Moerbeke-Wang11},
suitably generalised from integer to real parameters, is used to rederive the hard edge scaled limit known from \cite{Borodin99} (Section \ref{sec:hed}).

Parallel to the analysis of the Laguerre Muttalib--Borodin ensemble, we also undertake an analogous program of study in relation to the Jacobi weight
\begin{equation}\label{J}
e^{-V(\lambda)} = \lambda^{c_1} (1 -\lambda)^{c_2}, \qquad 0 < \lambda < 1, \quad c_1, c_2 > -1.
\end{equation}
This substituted in (\ref{1.1}) gives the Jacobi Muttalib--Borodin ensemble. 
Our realisation and double contour integral formula for the correlation kernel makes essential use of
results contained in \cite{Adler-van_Moerbeke-Wang11}. In relation to the global density, the resolvent is specified by a nonlinear equation which  we
solve using the Lagrange inversion formula to deduce that the moments of the global density are given in terms of particular binomial coefficients.
A trigonometric parametrisation of the spectral variable is given which allows for the determination of an explicit functional form for
the global density.
The hard edge scaled limit gives the same double contour integral form as found for the Laguerre case, in keeping with the findings of  \cite{Borodin99}.

We now give a precise statement of the main results in our paper for the Laguerre Muttalib--Borodin ensemble. All the results obtained for the Laguerre case have counterparts for the Jacobi case, 
described in the paragraph above and which 
 are presented in the body of the paper subsequent to presentation of the Laguerre case.
 We do not state the results in their most general form below, and readers can find the generalisations in subsequent sections.

\subsection{Main results for Laguerre Muttalib--Borodin ensemble}

First, the Laguerre Muttalib--Borodin ensemble can be realised as the eigenvalues of a random matrix in the \emph{upper-triangular random matrix ensemble}, which is defined in Section \ref{subsec:multiple_Laguerre_ensemble}, and if $\theta , c \in \intZ_{\geq 0}$, it can be realised as the eigenvalues of a random matrix in the \emph{multiple Laguerre ensemble}, which is defined in \cite{Adler-van_Moerbeke-Wang11} and also described in Section \ref{subsec:multiple_Laguerre_ensemble}.
\begin{prop}
  Let $Y$ be a upper-triangular $N \times N$ matrix with all entries independent, the strictly upper-triangular entries distributed as standard complex Gaussians, and the diagonal entries are real positive random variables with $\lvert y_{k, k} \rvert^2 \mathop{=}\limits^{\rm d}  \Gamma[\theta(k - 1) + c + 1, 1]$, or equivalently, $2\lvert y_{k,k} \rvert^2  \mathop{=}\limits^{\rm d} \chi_{2({\theta(k - 1) + c +1)}}^2$, for $k=1,\dots,N$. Then the eigenvalues $\lambda_1, \dotsc, \lambda_n$ of $Y^{\dagger} Y$ have the PDF proportional to \eqref{1.1} with $V$ given by \eqref{L}. Furthermore, if $\theta , c \in \intZ_{\geq 0}$, the PDF of the eigenvalues of $Y$ is the same as the PDF of the eigenvalues of $X^{\dagger} X$, where $X$ is an $N \times N$ random matrix whose entries $x_{j, k}$ with $1 \leq j \leq \theta(k - 1) + c + k$ have independent standard complex normal distribution, while other entries are zero.
\end{prop}
This result is covered by Proposition \ref{P1} and Corollary \ref{cor:X_and_Y_L}. 

The statistical system defined by the Laguerre Muttalib--Borodin ensemble is a determinantal point process,
and so is fully determined by its correlation kernel, for which we give a double contour integral formula.

\begin{prop} \label{prop:kernel_formula_L}
  The correlation kernel for the Laguerre Muttalib--Borodin ensemble can be written
\begin{equation}\label{KK2}
  K^{\rm L}(x,y) = {1 \over (2 \pi i)^2} {h(x) \over h(y)} \oint_{\Sigma} dz  \oint_{\Gamma_\alpha} dw \,
  { x^{-z - 1} y^w \Gamma(z+1) \prod_{k=1}^N ( z - \alpha_k) \over
    (z - w) \Gamma(w+1) \prod_{l=1}^N (w - \alpha_l)},
\end{equation}
where $\alpha_j = \theta(j - 1) + c$ is specified in \eqref{XY1}, and the contours $\Sigma$ and $\Gamma_{\alpha}$ are specified in Proposition \ref{prop:corr_kernel_L}. The function $h(x) = x^{c/2} e^{x/2}$, as specified in \eqref{hh} up to a multiplicative constant.
\end{prop}
This result is covered by Proposition \ref{prop:corr_kernel_L} and the definition \eqref{KK1} of the kernel.

Next, we derive the limiting global density of the ensemble, 
defined in terms of the correlation kernel $K^{\rm L}(x,x)$ according to \eqref{hh1}.

\begin{prop} \label{prop:global_density}
  For all $\theta > 0$, the limiting global density of the Laguerre Muttalib--Borodin ensemble with change of variable \eqref{Cv} is the Fuss--Catalan distribution, that is,
  \begin{equation}\label{1.8}
    \tilde{\rho}^{{\rm L}}_{(1)}(x) = \rho^{\FC}(x).
  \end{equation}
\end{prop}
Here the the Fuss--Catalan distribution is defined in \eqref{eq:Fuss-Catalan_density}. We remark  that this result is not totally new, as we 
explain in Proposition \ref{Pfl}, and after proper interpretation can be proved by a different method. But this method does not generalise simply to the Jacobi case, so we introduce two alternative
methods to prove Proposition \ref{prop:global_density}, one in Section \ref{S3.1} for positive integer $\theta$ and the other in Section \ref{subsec_multiple_Laguerre_ensemble} for general $\theta > 0$. Both of these two methods can be applied to the Jacobi case with little change.

Finally, we consider the local behaviour of $K^{\rm L}(x,y)$ around $0$, first obtained by Borodin \cite{Borodin99} in terms of Wright's
Bessel function.

\begin{prop} \label{prop:Laguerre_hard_edge}
We have
\begin{multline}\label{KK3}
\lim_{N \to \infty} N^{-1/\theta} K^{\rm L}(N^{-1/\theta} x, N^{-1/\theta} y)  \\
= 
\Big ( {y \over x} \Big )^{c/2} {\theta \over (2 \pi i)^2} \oint_{\Sigma^{\delta}_{-1/2}} dz  \oint_{\Gamma_0} dw \,
{ x^{-\theta z - 1} y^{\theta w} \over z - w}
{\Gamma(\theta z + c + 1) \over \Gamma(\theta w + c + 1) }
{\Gamma(z+1) \over \Gamma(w+1)} {\sin \pi z \over \sin \pi w},
\end{multline}
where $\Gamma_0$ is the Hankel loop contour starting at $\infty + i \epsilon$, running parallel to the positive real axis, looping around the origin, and finishing at
$\infty - i \epsilon$ after again running parallel to the negative real axis, while $\Sigma^{\delta}_{-1/2}$ is the contour consisting of two rays, one is from $-1/2$ to $e^{(\pi/2 + \delta)i} \cdot \infty$ and the other from $e^{-(\pi/2 + \delta)i} \cdot \infty$ to $-1/2$, where $\delta \in (0, \pi/2)$, see Figure \ref{fig:hard_edge}. If $\theta \geq 1$, we can also take $\delta = 0$ and let $\Sigma^{\delta}_{-1/2}$ be the upward vertical contour through $-1/2$.
In the case $\theta \in \mathbb Z^+$ this can be rewritten
\begin{multline}\label{KK3a}
x^{1/\theta - 1} \ \lim_{N \to \infty} N^{-1/\theta} K^{\rm L}(\theta (x/N)^{1/\theta} , \theta (y/N)^{1/\theta}) 
= 
\Big ( {y \over x} \Big )^{c/(2 \theta)} \\ \times
{1 \over (2 \pi i)^2} \oint_{\Sigma^{\delta}_{-1/2}} dz  \oint_{\Gamma_0} dw \,
{ x^{- z - 1} y^{w} \over z - w}
{\prod_{k=0}^{\theta - 1} \Gamma(z + (c+1+k)/\theta) \over\prod_{k=0}^{\theta - 1}  \Gamma(w + (c+1+k)/\theta)}
{\Gamma(z+1) \over \Gamma(w+1)} {\sin \pi z \over \sin \pi w}.
\end{multline}
\end{prop}
This result is proved in Section \ref{subsubsec:hard_edge_L}.
We remark that the $\theta \to 0_+$ limit of Propositions \ref{prop:global_density} and \ref{prop:Laguerre_hard_edge} are also obtained in this paper, see Sections \ref{subsubsec:theta=0_L} and \ref{subsubsec:Laguerre_hard_edge_L0}.

\section{Realisations and extensions} \label{sec:realisations_and_extensions}

\subsection{The Laguerre upper-triangular ensemble} \label{subsec:multiple_Laguerre_ensemble}

We use the term complex Wishart matrix to refer to a random matrix of the form $W^\dagger W$ with $W$ containing
complex Gaussian entries with mean and standard deviation to be specified. Furthermore, 
for any $M \times N$ matrix $A$, we denote $A_{m \times n}$ as the $m \times n$ matrix consisting of the upper-left $m \times n$ block of $A$.
We are interested in a particular complex Wishart matrix, due to Adler, van Moerbeke and Wang \cite{Adler-van_Moerbeke-Wang11}, which is
parametrised by non-negative integers $\alpha_1, \dots, \alpha_N$ satisfying
\begin{equation}\label{IN}
  1 + \alpha_1 \le 2 + \alpha_2 \le \cdots \le N + \alpha_N \le M,
\end{equation}
with $M \ge N$. One defines the $M \times N$ random matrix $X = [x_{j,k}]_{j = 1,\dots,M \atop k=1,\dots,N}$ to have
entries
\begin{equation}\label{IN1}
x_{j,k} \mathop{=}\limits^{\rm d} \left \{ \begin{array}{ll}
{\rm N}[0,1/2] + i {\rm N}[0,1/2], & 1 \le j \le k + \alpha_k \nonumber \\
0, & {\rm otherwise}.  \end{array} \right.
\end{equation}
All nonzero entries of $X$ are therefore standard complex Gaussians. Moreover, the condition (\ref{IN}) implies all entries on
and above the diagonal of $X$ are non-zero. Due to its relationship to a certain family of special functions by the
same name, the ensemble of matrices $X_{M \times n}^\dagger X_{M \times n}$, $(n=1,\dots,M)$ was referred to
in \cite{Adler-van_Moerbeke-Wang11} as the multiple Laguerre ensemble.
Our first result is to specify an upper-triangular matrix obtained
from $X$ by a sequence of Householder transformations. Such transformations were introduced into random matrix
theory in \cite{Trotter84}, \cite{Silverstein85}, and furthermore underpin the construction of $\beta$-ensembles as formulated in \cite{Dumitriu-Edelman02}.

We denote by $\Gamma[k,\sigma]$ the gamma distribution,
specified by the density function $(\sigma^{-k}/\Gamma(k)) t^{k-1} e^{-t/\sigma}$, $(t>0)$.
Let $Y = [y_{j,k}]_{j, k = 1, \dotsc, N}$ be an upper-triangular $N \times N$ random matrix with all entries independent. The strictly upper-triangular entries are distributed as standard complex Gaussians, and
the diagonal entries are real positive random variables with distributions depending on parameters $\alpha_k > -1$ specified by
\begin{equation}\label{yk}
  \lvert y_{k,k} \rvert^2  \mathop{=}\limits^{\rm d}  \Gamma[\alpha_k + 1, 1], \quad \text{or equivalently} \quad 2\lvert y_{k,k} \rvert^2  \mathop{=}\limits^{\rm d} \chi_{2({\alpha_k+1)}}^2, \qquad (k=1,\dots,N).
\end{equation}
We say that $Y$ is a random matrix in the upper-triangular ensemble, and that $Y^\dagger Y$ belongs to the Laguerre upper-triangular ensemble.

\begin{prop}\label{P1}
  The random matrices $X^{\dagger}_{M \times 1} X_{M \times 1}$, $X^{\dagger}_{M \times 2} X_{M \times 2}$, $\dotsc$, $X^{\dagger}_{M \times N} X_{M \times N}$  and  $Y^{\dagger}_{N \times 1} Y_{N \times 1}, Y^{\dagger}_{N \times 2} Y_{N \times 2}, \dotsc, Y^{\dagger}_{N \times N} Y_{N \times N}$ have the same joint distribution. In particular, the eigenvalues of $X^{\dagger}_{M \times 1} X_{M \times 1}, X^{\dagger}_{M \times 2} X_{M \times 2}, \dotsc, X^{\dagger}_{M \times N} X_{M \times N}$ have the same joint distribution as the eigenvalues of $Y^{\dagger}_{N \times 1} Y_{N \times 1}, Y^{\dagger}_{N \times 2} Y_{N \times 2}, \dotsc, Y^{\dagger}_{N \times N} Y_{N \times N}$.
\end{prop}

\begin{proof}
  Recall that a complex Householder reflection matrix acting on the left of the $M \times N$ matrix $X$ has the form
  \begin{equation*}
    U = \mathbb I_{M} - 2 \vec{u} \: \vec{u}^{\, \dagger},
  \end{equation*}
  where ${}^\dagger$ denotes the operation of complex conjugate and transpose and $\vec{u}$ is a $M \times 1$ complex column vector with the property that $ \vec{u}^{\, \dagger} \cdot \vec{u} = 1$. This latter requirement implies $U^\dagger U = \mathbb I_M$, so $U$ is unitary. Geometrically $U$ corresponds to a reflection in the complex hyperplane orthogonal to $\vec{u}^{\, \dagger}$.
  
  To prove the proposition, we construct a sequence of $M \times N$ random matrices $X^{(0)} = X, X^{(1)}, \dotsc, X^{(N)}$ and a sequence of $M \times M$ random Householder reflection matrices $U^{(1)}, \dotsc, U^{(N)}$ inductively. The matrix $U^{(l)}$ is determined by $X^{(l - 1)}$ and $X^{(l + 1)} = U^{(l + 1)} X^{(l)}$ ($l = 0, 1, \dotsc, N - 1$), where  $X^{(l)}$ satisfies: (1) the $(j, k)$ entry is zero for $k < j \leq M$ if $1 \leq k \leq l$, or $k + \alpha_k < j \leq M$ if $l < k \leq N$; (2) the diagonal $(k, k)$ entry is real positive and its square is in $\Gamma[\alpha_k + 1, 1]$ distribution if $k \leq l$; (3) all other entries are in complex standard Gaussian distribution. By the method of construction to be detailed below, the block of $X^{(N)}$ consisting of the first $N$ rows is such that
  $X^{(N)}_{N \times N} \dequal Y$, and all entries in the last $(M - N)$ rows of $X^{(N)}$ are zeros. Thus the joint distribution of $(X^{(N)}_{M \times n})^{\dagger} X^{(N)}_{M \times n}$ for $n = 1, \dotsc, N$ is the same as that of $Y^{\dagger}_{N \times n} Y_{N \times n}$. On the other hand, $X^{(N)}$ is a function of $X$ and for any $n$, $(X^{(N)}_{M \times n})^{\dagger} X^{(N)}_{M \times n} = X^{\dagger}_{M \times n} X_{M \times n}$. 
Except for the construction of $\{X^{(l)} \}$, this finishes the proof.

  We now give an algorithm for the construction using induction. For any $l = 0, \dotsc, N - 1$, by the induction assumption
   $X^{(l)} = [x^{(l)}_{j, k}]_{j = 1,\dots,M \atop k=1,\dots,N}$ is well defined and $x^{(l)}_{j, l + 1}$ with $j = l + 1, \dotsc, l + 1 + \alpha_{l + 1}$ are in independent complex standard Gaussian distribution. We denote the $(1 + \alpha_{l + 1})$-dimensional vector
  \begin{equation*}
    \vec{x}^{(l + 1)} = (x^{(l)}_{l + 1, l + 1}, x^{(l)}_{l + 2, l + 1}, \dotsc, x^{(l)}_{l + 1 + \alpha_{l + 1}, l + 1})^T.
  \end{equation*}
  We construct the Householder reflection matrix $U^{(l + 1)}$ using $\vec{u}^{(l + 1)}$, which is a concatenation of an $l$-dimensional zero vector, an $(1 + \alpha_{l + 1})$-dimensional vector $\vec{v}^{(l + 1)}$, and an $(M - l - 1 - \alpha_{l + 1})$-dimensional zero vector. Here the vector $\vec{v}^{(l + 1)}$ is defined as the unit vector
  \begin{equation*}
    \frac{\vec{w}}{\lVert \vec{w} \rVert}, \quad \text{where} \quad \vec{w} = \vec{x}^{(l + 1)} - (\lVert \vec{x}^{(l + 1)} \rVert, 0, \dotsc, 0)^T,
  \end{equation*}
  if $\vec{x}^{(l + 1)} \neq 0$, and is defined simply as $(1, 0, \dotsc, 0)^T$ if $\vec{x}^{(l + 1)} = 0$.

  From the definition of $U^{(l + 1)}$, and recalling $X^{(l + 1)} = U^{(l + 1)} X^{(l)}$, it is clear that $X^{(l+1)}$ and $X^{(l)}$
  are identical in the upper block consisting of the first $l$ rows and the lower block consisting of the last $(M - l - 1 - \alpha_{l + 1})$ rows. In the middle block consisting of the remaining $1 + \alpha_{l + 1}$ rows, the left part consisting of the middle part of the left-most $l$ columns of $X^{(l)}$ are zeros, so they remain zero in $X^{(l + 1)}$. The entries of $X^{(l)}$ in the right part of the middle block consisting of the right-most $(N - l - 1)$ columns are in independent complex standard Gaussian distribution, and they are all independent of $\vec{u}^{(l + 1)}$. So after the left multiplication by $U^{(l + 1)}$, the entries of $X^{(l + 1)}$ in that part of the middle block are also in independent complex standard Gaussian distribution by the rotational invariance of random Gaussian vectors. The $(l + 1)$-th column of the middle block becomes $(\lVert \vec{x}^{(l + 1)} \rVert, 0, \dotsc, 0)^T$, so its first entry is positive and the square of the first entry is in $\Gamma[\alpha_k + 1, 1]$ distribution, while all other entries are zero. Hence each  $X^{(l + 1)}$ satisfies the required properties,  so finishing the proof by induction.
\end{proof}

\begin{remark}
  Since we are interested in the spectral properties of $X^{\dagger}_{M \times n} X_{M \times n}$ in the multiple Laguerre ensemble
  of \cite{Adler-van_Moerbeke-Wang11}, and by Proposition \ref{P1} they are identical to $Y^{\dagger}_{N \times n} Y_{N \times n}$ with $\alpha_1, \dotsc, \alpha_N$ restricted to some subset to their domain, we can think 
   of the upper-triangular ensemble as a generalisation of the multiple Laguerre ensemble.
\end{remark}

\medskip

 The upper-triangular matrix $Y$ distributed as in \eqref{yk} with
\begin{equation}\label{XY1}
  \alpha_j = \theta (j-1) + c \qquad (j=1,\dots,N),
\end{equation}
has recently been shown by Cheliotis \cite{Cheliotis14} to have the PDF for its squared singular values given by the
Laguerre Muttalib--Borodin ensemble (\ref{1.1}) with weight (\ref{L}).  For this to relate to our
construction from complex Gaussian matrices according to Proposition \ref{P1} we must
have $\theta$ and $c$ non-negative integers. Thus in this circumstance we have identified a realisation of
 this ensemble as the eigenvalue PDF of a Wishart matrix.
 
 \begin{corollary} \label{cor:X_and_Y_L}
 Consider the matrix $X$ as defined below (\ref{IN}), and with $\theta, c \in \mathbb Z_{\ge 0}$, let $\alpha_j$ be specified as in (\ref{XY1}). We have that the eigenvalue PDF 
of the Wishart matrix  $X^\dagger X$ is given by the  Laguerre Muttalib--Borodin ensemble (\ref{1.1}) with weight (\ref{L}). 
\end{corollary}

Corollary \ref{cor:X_and_Y_L} is a special case of Corollary \ref{C1}\ref{enu:cor:C_1:a} below. Additional details of the spectral properties of $X^\dagger X$, or equivalently
according to Proposition \ref{P1}, of $Y^\dagger Y$ with the parameters $\alpha_j$ of integer values, beyond the joint distribution of the eigenvalues have been given in \cite{Adler-van_Moerbeke-Wang11}. In particular, we can read off from the multiple Laguerre part of \cite[Thm.~1]{Adler-van_Moerbeke-Wang11} the explicit functional form of the conditional distribution of the eigenvalues of $Y_{N \times n}^\dagger Y_{N \times n}$, given the eigenvalues of $Y_{N \times (n-1)}^\dagger Y_{N \times (n-1)}$, where $Y_{p \times q}$ denotes the top $p \times q$ sub-block of $Y$. Below we show that the results can be generalised to arbitrary upper-triangular random matrices $Y$ with real-valued $\alpha_j > -1$. The proofs are similar to those in \cite{Adler-van_Moerbeke-Wang11} and we mainly emphasis the differences.

\begin{prop} \label{P3}
  Let the $N \times N$ random matrix $Y$ be in the upper-triangular ensemble with diagonal entries specified by \eqref{yk}. Denote by $\{\lambda_1,\dots,\lambda_n\}$
  and $\{\mu_1,\dots,\mu_{n-1}\}$ the eigenvalues of $Y_{N \times n}^\dagger Y_{N \times n}$ and
  $Y_{N \times (n-1)}^\dagger Y_{N \times (n-1)}$ respectively in descending order. The conditional PDF of 
  $\{\lambda_1,\dots,\lambda_n\}$, with $\{\mu_1,\dots,\mu_{n-1}\}$ fixed and distinct,
  is equal to
  \begin{multline}\label{J1}
    p_{n,n-1}(\{\lambda_j\}_{j=1}^n, \{\mu_j \}_{j=1}^{n-1}) = \\
    {1 \over \Gamma(\alpha_n + 1)} \prod_{k=1}^n \lambda_k^{\alpha_n} e^{-\lambda_k} \prod_{l=1}^{n-1} \mu_l^{-\alpha_{n} - 1} e^{\mu_l}
    {\prod_{1 \le j < k \le n} ( \lambda_j - \lambda_k) \over \prod_{1 \le j < k \le n-1} ( \mu_j - \mu_k) },
  \end{multline}
  subject to the interlacing constraint
  \begin{equation}\label{J2}
    \lambda_1 \geq \mu_1 \geq \lambda_2 \geq \mu_2 \geq \cdots \geq \lambda_n \ge 0.
  \end{equation}
\end{prop}
Actually we can prove a slightly stronger result:
\begin{lemma} \label{lem:P3}
  Let $A = (a_{i, j})$ be an $(n - 1) \times (n - 1)$ matrix with the eigenvalues of $A^{\dagger}A$ being $\{\mu_1, \dotsc, \mu_{n - 1} \}$ in descending order. Define the $n \times n$ matrix $B = (b_{i, j})$ by letting: (1) the $(n - 1) \times (n - 1)$ upper-triangular block equal $A$; (2) the bottom row has all entries but the rightmost one equal $0$; (3) all entries of the rightmost row be independent random variables, such that the $(n, n)$-entry $b_{n, n}$ is real positive with $\lvert b_{n, n} \rvert^2 \mathop{=}\limits^{\rm d}  \Gamma[\alpha_n + 1, 1]$, and all other entries in the row are in standard complex normal distribution. Then the eigenvalues of $B^{\dagger}B$, denoted by $\{ \lambda_1, \dotsc, \lambda_n \}$ in descending order, satisfies the interlacing constraint \eqref{J2} and have the distribution given by $p_{n, n - 1}(\{\lambda_j\}_{j=1}^n, \{\mu_j \}_{j=1}^{n-1})$ in \eqref{J1}.
\end{lemma}

\begin{proof}[Proof of Lemma \ref{lem:P3}]
  A key point is that
  \begin{equation}\label{dd}
    B B^\dagger = B_{n \times (n-1)}  B^\dagger_{n \times (n-1)}  +
    \vec{y} \vec{y}^\dagger,
  \end{equation}
  where $\vec{y} = (b_{1, n}, b_{2, n}, \dotsc, b_{n, n})^T$. The distribution of $b_{k, n}$ implies that
  \begin{equation}\label{dd1}
    \lvert b_{k, n} \rvert^2  \mathop{=}\limits^{\rm d} \Gamma[1,1] \: (k=1,\dots,n-1), \quad
    \lvert b_{n, n} \rvert^2   \mathop{=}\limits^{\rm d}  \Gamma[\alpha_k+1,1].
  \end{equation}
  Moreover there exists an $(n - 1) \times (n - 1)$ unitary matrix $V$ depending on $A = B_{(n - 1) \times (n - 1)}$ such that
  \begin{equation}\label{dd2}
    (V \oplus \bbI_1) B_{n \times (n - 1)} B^\dagger_{n \times (n - 1)} (V \oplus \bbI_1)^{\dagger} = \diag[\mu_1,\dots,\mu_{n-1}, 0],
  \end{equation}
where
\begin{equation*}
  V \oplus \bbI_1 =
  \begin{bmatrix}
    V & \bbO_{(n-1) \times 1} \\
    \bbO_{1 \times (n-1)} & 1
  \end{bmatrix}.
\end{equation*}
 Using the property that the multiplication of a unitary matrix and a vector of independent standard complex Gaussians yields another vector of independent standard complex Gaussians, we have that the vector $\vec{z} = (z_1, \dotsc, z_n)^T$ defined as 
  \begin{equation*}
    \vec{z} = (V \oplus \bbI_1) \vec{y}
  \end{equation*}
 has the properties that all its components are independent, $z_1, \dotsc, z_{n - 1}$ are in standard complex normal distribution, and
  \begin{equation*}
    \lvert z_k \rvert^2  \mathop{=}\limits^{\rm d} \Gamma[1,1] \: (k=1,\dots,n-1), \quad \text{and} \quad z_n = b_{n, n}.
  \end{equation*}
  Conjugating both sides of \eqref{dd} as in \eqref{dd2}, we have that
  \begin{equation}\label{3.1}
    \ePDF B^{\dagger} B = \ePDF B B^{\dagger} = \ePDF \Big (\diag[\mu_1,\dots,\mu_{n-1},0] + \vec{z} \vec{z}^{\, \dagger} \Big )
  \end{equation}
  where ePDF denotes the eigenvalue PDF. 
  
  By a standard manipulation of the characteristic polynomial, one can show that the eigenvalue equation for the matrix on the RHS of (\ref{3.1}) is given by
  \begin{equation} \label{eq:char_equation_Laguerre}
    0 = 1 + \Big ( - {\lvert z_n \rvert^2 \over \lambda} + \sum_{k=1}^{n-1} {\lvert z_k \rvert^2 \over \mu_k - \lambda} \Big ).
  \end{equation}
  The distribution of the roots of this rational function, with residues distributed according to (\ref{dd1}), and thus
  the conditional PDF of $\{\lambda_1,\dots,\lambda_n\}$, can now be read off as a special case of
  \cite[Cor.~3]{Forrester-Rains05}, and (\ref{J1}) with interlacing constraint (\ref{J2}) follows. 
\end{proof}

\medskip
 
Knowledge of Proposition \ref{P3} allows us to rederive the result of Cheliotis \cite{Cheliotis14} noted below (\ref{XY1}).
We will require the use of a particular multiple integral evaluation.

\begin{lemma}\label{L1}
Let $R_{\preceq \lambda}$ denote the region \eqref{J2} for $(\mu_1, \dotsc, \mu_{n - 1})$, and suppose $\alpha_k  \ne \alpha_n$ $(k=1,\dots,n-1)$.
 We have
\begin{align}\label{A1}
&  \prod_{l=1}^n \lambda_l^{\alpha_n}
\int_{R_{\preceq \lambda}} \prod_{l=1}^{n-1} \mu_l^{- \alpha_n-1}\det[ \mu_j^{\alpha_k}]_{j,k=1,\dots,n-1}
\, d\mu_1 \cdots d \mu_{n-1}
\nonumber \\
&
\qquad = \prod_{k=1}^{n-1} {1 \over \alpha_k - \alpha_n}    \det[ \lambda_{j}^{\alpha_k}]_{j,k=1,\dots,n}. 
\end{align}
\end{lemma}

\begin{proof}
  We have
  \begin{equation}\label{S1}
    \prod_{l=1}^{n-1} \mu_l^{- \alpha_n-1}\det[ \mu_j^{\alpha_k } ]_{j,k=1,\dots,n-1}  =
    \det [ \mu_j^{\alpha_k - \alpha_{n} -1}]_{j,k=1,\dots,n-1}.
  \end{equation}
  Since the dependence on $\mu_j$ is entirely in row $j$, the integration over $\lambda_{j+1} >
  \mu_j > \lambda_j$ can be done row-by-row. Furthermore, by adding row $1,\dots,j-1$ to row $j$
  the integration can be taken to be over $\lambda_{j+1} >
  \mu_j > \lambda_1$. Applying this operation to (\ref{S1}) gives
  $$
  \prod_{k=1}^{n-1} {1 \over \alpha_k - \alpha_n} \det[ \lambda_{j+1}^{\alpha_k - \alpha_n} -   \lambda_{1}^{\alpha_k - \alpha_n}]_{j,k=1,\dots,n-1}.
  $$
  But
  $$
 \det[ \lambda_{j+1}^{\alpha_k - \alpha_n} -   \lambda_{1}^{\alpha_k - \alpha_n}]_{j,k=1,\dots,n-1} =
  \det[ \lambda_{j}^{\alpha_k-\alpha_n}]_{j,k=1,\dots,n},
  $$
  as can be seen by subtracting the first row from each of the next rows in the determinant on the RHS, 
  then expanding by the final column to obtain the LHS. Thus \eqref{A1} now follows. 
\end{proof}

\begin{remark}
  Let $\lambda = (\lambda_1,\lambda_2,\dots,\lambda_N)$ denote a partition \cite[Chap.~7]{Stanley99}.
  The Schur polynomial can be defined by
  \begin{equation}\label{SU}
    s_\lambda(x_1,\dots,x_N) = {\det [x_j^{\lambda_{N-k+1} + k - 1} ]_{j,k=1,\dots,N} \over
      \prod_{1 \le j < k \le N} (x_k - x_j) },
  \end{equation}
  which in fact is well defined for any $N$-array $\lambda$.
  With $N = n - 1$, set $\lambda_{N-k+1} + k = \alpha_{k} -\alpha_n $ and define 
  $\tilde{\alpha}^{(n-1)} = (\alpha_{n-1} - \alpha_n -n-1,
  \alpha_{n-2} - \alpha_n -n+ 2,\dots, \alpha_1 - \alpha_n  +1)$. 
  Substituting (\ref{SU}) in (\ref{A1}) gives
  \begin{multline}
    \int_{R_{\preceq \lambda}}  
    \prod_{1 \le j < k \le n - 1} (\mu_k - \mu_k)
    s_{\tilde{\alpha}^{(n-1)}}(\mu_1,\dots,\mu_{n-1}) \, d\mu_1 \cdots d\mu_{n-1}   \\
    = \prod_{k=1}^{n-1} {1 \over \alpha_k - \alpha_n}   
    \prod_{1 \le j < k \le n} (\lambda_k - \lambda_j) 
    s_{\tilde{\alpha}^{(n)}}( \lambda_1,\dots,\lambda_n),
  \end{multline}
  where $\tilde{\alpha}^{(n)}$ is the $n$-array formed from the $(n-1)$-array $\tilde{\alpha}^{(n-1)}$
  by appending 0 to the end.
  This is a special case of an integration formula from the theory of Jack polynomials
  \cite{Okounkov-Olshanski97, Kuznetsov-Mangazeev-Sklyanin03, Kohler11}; see also \cite[Eq.~(12.210)]{Forrester10}.
\end{remark}
 
\begin{corollary}\label{C1}
  Let the $N \times N$ random matrix $Y$ be in the upper-triangular ensemble with diagonal entries specified by \eqref{yk}.
  \begin{enumerate}[label=(\alph*)]
  \item \cite{Cheliotis14} \label{enu:cor:C_1:a}
    For all $n = 1, \dotsc, N$, the eigenvalue  PDF of $Y^\dagger_{N \times n} Y_{N \times n} = Y^\dagger_{n \times n} Y_{n \times n}$, denoted by $\lambda^{(n)} = \{ \lambda^{(n)}_1, \dotsc, \lambda^{(n)}_n \}$ in descendent order, is equal to
    \begin{multline}\label{Pd1x}
      p_n(\lambda^{(n)}_1, \dotsc, \lambda^{(n)}_n) = \\
      {C^{-1}_n \over \prod_{1 \le j < k \le n} (\alpha_j - \alpha_k)}  \left( \prod_{k=1}^n e^{-\lambda^{(n)}_k} \right)
      \prod_{1 \le j < k \le n} (\lambda^{(n)}_j - \lambda^{(n)}_k)    \det[ (\lambda^{(n)}_{j})^{\alpha_k}]_{j,k=1,\dots,n} ,
    \end{multline}
    where
    \begin{equation}\label{Pd2ax}
      C_n = \prod_{l=1}^n  \Gamma(\alpha_l + 1).
    \end{equation}
    In the special case the $\alpha_j$ are  given by (\ref{XY1}) this reduces to
    \begin{equation}\label{Pd2}
      p_n(\lambda^{(n)}_1, \dotsc, \lambda^{(n)}_n) = {1 \over C_{n,\theta,c}}   \prod_{k=1}^n  (\lambda^{(n)}_k)^c e^{-\lambda^{(n)}_k} 
      \prod_{1 \le j < k \le n} (\lambda^{(n)}_j - \lambda^{(n)}_k)   ((\lambda^{(n)}_j)^\theta - (\lambda^{(n)}_k)^\theta), 
    \end{equation}
    where
    \begin{equation}\label{Pd2b1}
      C_{n,\theta,c}  = \prod_{l=1}^n  \Gamma(\theta (l-1) + c + 1)  \, \,  \theta^{n(n-1)/2}\prod_{l=1}^{n-1} l!.
    \end{equation}
  \item \label{enu:cor:C_1:b}
    The joint probability density function of $\lambda^{(1)}, \lambda^{(2)}, \dotsc, \lambda^{(N)}$ is equal to
    \begin{multline} \label{eq:jpdf_upp_tri}
      p(\lambda^{(1)}, \dotsc, \lambda^{(n)}) = \frac{1}{C_N} \prod_{1 \leq j < k \leq N} (\lambda^{(N)}_j - \lambda^{(N)}_k) \prod^N_{n = 1} (\lambda^{(N)}_n)^{\alpha_N} e^{-\lambda^{(N)}_n} \\
      \times \prod^{N - 1}_{n = 1} \left( \prod^n_{i = 1} (\lambda^{(n)}_i)^{\alpha_n - \alpha_{n + 1} - 1} \right) 1_{\lambda^{(n)} \preceq \lambda^{(n + 1)}},
    \end{multline}
    where $C_n$ is defined in \eqref{Pd2ax}, and $1_{\mu \preceq \lambda}$ is the indicator function of the region that satisfies inequality \eqref{J2}.
  \end{enumerate}
  In case that some $\alpha_i$ are identical, we understand the formulas in the limiting sense with \lHopital's rule. 
\end{corollary}

\begin{proof}[Proof of Part \ref{enu:cor:C_1:a}]
  We prove the case that $\alpha_i$ are distinct, and the general result follows by analytical continuation.

   In the case $n=1$ we see from \eqref{yk} that the PDF of the unique eigenvalue $\lambda^{(1)}_1$ is equal to $\Gamma(\alpha_1 + 1)^{-1} (\lambda^{(1)}_1)^{\alpha_1} e^{-\lambda^{(1)}_1}$, which is the $n=1$ case of \eqref{Pd1x}. Let us now assume that \eqref{Pd1x} is valid in the case $n-1$. Next we prove the $n > 1$ cases by induction. For notational simplicity, we denote $\lambda^{(n)}_i$ by $\lambda_i$ and $\lambda^{(n - 1)}_i$ by $\mu_i$. Recall the conditional PDF $p_{n,n-1}(\{\lambda_j\}_{j=1}^n, \{\mu_j \}_{j=1}^{n-1})$ of $\{ \lambda_1, \dotsc, \lambda_n \}$ with fixed $\{ \mu_1, \dotsc, \mu_{n - 1} \}$, such that the interlacing condition \eqref{J2} is satisfied, defined in \eqref{J1}. Our task is to show that with $R_{\preceq \lambda}$  the domain of $\{ \mu_1, \dotsc, \mu_{n - 1} \}$ given by \eqref{J2},
  \begin{equation}\label{PE}
    \int_{R_{\preceq \lambda}} p_{n,n-1}(\{\lambda_j\}_{j=1}^n,
    \{\mu_j \}_{j=1}^{n-1}) p_{n-1}(\mu_1,\dots,\mu_{n-1}) \, d\mu_1 \cdots d \mu_{n-1} =
    p_n(\lambda_1,\dots,\lambda_n).
  \end{equation}
  Substituting for the integrand, then making use of Lemma \ref{L1} shows that the LHS is equal to
  \begin{equation*}
    {1 \over C_{n-1}} {1 \over \Gamma(\alpha_n + 1)} \prod_{1 \leq j < k \leq n} {1 \over \alpha_j - \alpha_k} \prod_{k=1}^n e^{-\lambda_k}
    \prod_{1 \le j < k \le n} (\lambda_j - \lambda_k)
    \det[ \lambda_{j}^{\alpha_k}]_{j,k=1,\dots,n}. 
  \end{equation*}
  Comparison with (\ref{Pd1x}) and recalling (\ref{Pd2ax}) shows that this is precisely the RHS.
  
  To deduce (\ref{Pd2}) from (\ref{Pd1x}) we make use of the Vandermonde determinant identity
  $$
  \prod_{1 \le j < k \le n}( x_k - x_j) = \det [ x_j^{k-1} ]_{j,k=1,\dots,n}.
  $$
\end{proof}

\begin{proof}[Proof of Part \ref{enu:cor:C_1:b}]
  By Lemma \ref{lem:P3}, we have that the eigenvalues $\lambda^{(1)}, \lambda^{(2)}, \dotsc, \lambda^{(N)}$ constitute an inhomogeneous Markov chain with the transition probability density function from time $n - 1$ to time $n$ given by \eqref{J1}. Thus the joint distribution function of $\lambda^{(n)}$ ($n = 1, \dotsc, N$) is obtained by multiplying \eqref{J1} repeatedly. The argument is the same as the proof of \cite[Cor. 1]{Adler-van_Moerbeke-Wang11} and we omit the details.
\end{proof}

\begin{remark}
In the special case $\theta = 0$ of (\ref{XY1}) we have from  (\ref{yk}) that  $2 \lvert y_{k,k} \rvert^2 
  \mathop{=}\limits^{\rm d}  \chi_{2(c+1)}^2$ independent of $k$. Taking the limit $\theta \to 0^+$ in \eqref{Pd2}, we have that the corresponding eigenvalue PDF is equal to \cite{Cheliotis14}
 \begin{equation}\label{Pd3A}
   {1 \over \Gamma(c + 1)^n \prod_{l=1}^{n-1} l!}   \prod_{k=1}^n  \lambda_k^c e^{-\lambda_k} 
   \prod_{1 \le j < k \le n} (\lambda_j - \lambda_k)   ( \log \lambda_j - \log \lambda_k).
 \end{equation}
\end{remark} 

From the joint probability distribution function \eqref{eq:jpdf_upp_tri}, we have, as a natural generalisation of \cite[Thm.~3(c)]{Adler-van_Moerbeke-Wang11}:
\begin{prop} \label{prop:corr_kernel_L}
  For matrices in the the upper-triangular ensemble, let the eigenvalues $\lambda^{(1)}$, $\dotsc, \lambda^{(N)}$ be defined as in Corollary \ref{C1}. The eigenvalues $\lambda^{(1)}, \dotsc, \lambda^{(N)}$ constitute a determinantal process, and the correlation kernel of $\lambda^{(n_1)}$ and $\lambda^{(n_2)}$ is given by
  \begin{multline} \label{eq:corr_kernel_L}
    K(n_1, x; n_2, y) = \frac{-1}{2\pi i} \oint_{\Gamma_{\alpha}} \frac{x^{-w - 1} y^w}{\prod^{n_2}_{l = n_1 + 1} (w - \alpha_l)} dw 1_{x < y} 1_{n_1 < n_2} \\
    + \frac{1}{(2\pi i)^2} \oint_{\Sigma} dz \oint_{\Gamma_{\alpha}} dw \frac{x^{-z - 1} y^w \Gamma(z + 1)}{(z - w) \Gamma(w + 1)} \frac{\prod^{n_1}_{k = 1} (z - \alpha_k)}{\prod^{n_2}_{l = 1} (w - \alpha_l)},
  \end{multline}
  where $\Gamma_{\alpha}$ is a contour enclosing $\alpha_1,\dots,\alpha_N$, while $\Sigma$ is a Hankel like contour, starting at $-\infty - i \epsilon$, running parallel to the negative real axis, looping around the point $z=-1$ and the contour $\Gamma_{\alpha}$, and finishing at $-\infty + i \epsilon$ after again running parallel to the negative real axis, see Figure \ref{fig:Gamma_alpha_Sigma}. 
\end{prop}
\begin{figure}[htb]
  \centering
  \includegraphics{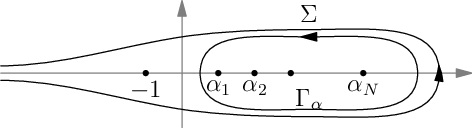}
  \caption{Schematic figures of $\Gamma_{\alpha}$ and $\Sigma$.}
  \label{fig:Gamma_alpha_Sigma}
\end{figure}
The proof of the proposition is by a standard argument for  determinantal processes based on the joint probability density function \eqref{eq:jpdf_upp_tri}. In \cite[Thm.~3(c)]{Adler-van_Moerbeke-Wang11}, the proposition for non-negative integer $\alpha_i$ under condition \eqref{IN} is proved. Since the proof does not use these additional conditions, it is also a complete proof to Proposition \ref{prop:corr_kernel_L}.

\subsection{The Jacobi upper-triangular ensemble}

Adler, van Moerbeke and Wang \cite{Adler-van_Moerbeke-Wang11} considered the joint eigenvalue PDF of a sequence of random matrices
\begin{equation}\label{YA1}
  \begin{split}
    \tilde{J}_n = {}& (X^\dagger_{M \times n} X_{M \times n})(A^\dagger_{M' \times n} A_{M' \times n} + X^\dagger_{M \times n} X_{M \times n})^{-1} \\
    = {}& (A_{M' \times n}^\dagger  A_{M' \times n} (X^\dagger_{M \times n} X_{M \times n})^{-1} + \mathbb I_n )^{-1}, \quad n = 1, \dotsc, N
  \end{split}
\end{equation}
where $X_{M \times n}$ is the left $M \times n$ sub-block of the $M \times N$ matrix $X$ specified by \eqref{IN1}, while $A_{M' \times n}$ is the top $M' \times n$ sub-block of the $M' \times N$ matrix
$A$ which has all elements independently distributed as standard complex Gaussians. Here it is required that
$M \ge n$ and $M' \ge n$. Due to its relationship to particular multiple orthogonal polynomials, this was referred to as the
Jacobi-\Pineiro\  ensemble.

We consider in this section a Jacobi-type counterpart of the random matrix ensembles $Y^{\dagger}_{N \times n} Y_{N \times n}$ in Section \ref{subsec:multiple_Laguerre_ensemble}. To this end, we use the $N \times N$ random matrix $Y$ specified in Section \ref{subsec:multiple_Laguerre_ensemble} above Proposition \ref{P1}, and denote the $N \times N$ random matrix $Z$, which is also in the same upper-triangular ensemble as $Y$. Let $Z = [z_{j, k}]_{j, k = 1, \dotsc, N}$ be an upper-triangular $N \times N$ random matrix with all upper-triangular entries independent, the diagonal entries be real positive with distributions depending on parameters $\beta_k > -1$
\begin{equation} \label{eq:diagonal_Z}
  \lvert z_{k,k} \rvert^2  \mathop{=}\limits^{\rm d}  \Gamma[\beta_k + 1, 1], \quad \text{or equivalently} \quad 2\lvert z_{k,k} \rvert^2  \mathop{=}\limits^{\rm d} \chi_{2({\beta_k+1)}}^2, \qquad (k=1,\dots,N),
\end{equation}
and all entries strictly above the diagonal be complex and in standard complex Gaussian distribution. Then define for all $n = 1, \dotsc, N$
\begin{equation} \label{eq:defn_J_n}
  \begin{split}
    J_n = {}& (Y^{\dagger}_{N \times n} Y_{N \times n})(Z^{\dagger}_{N \times n} Z_{N \times n} + Y^{\dagger}_{N \times n} Y_{N \times n})^{-1} \\
    = {}& (Z^{\dagger}_{N \times n} Z_{N \times n} (Y^{\dagger}_{N \times n} Y_{N \times n})^{-1} + \mathbb I_n )^{-1},
  \end{split}
\end{equation}
where $Y_{N \times n}$ (\resp\ $Z_{N \times n}$) is the left $N \times n$ sub-block of the $N \times N$ matrix $Y$ (\resp\ $Z$). Note that both $\tilde{J}_n$ and $J_n$ depend on parameters $\alpha_1, \dotsc, \alpha_n > -1$, $J_n$ also depends on $\beta_1, \dotsc, \beta_n > -1$ and for $\tilde{J}_n$ it is further assumed that $\alpha_1, \dotsc, \alpha_n$ are non-negative integers satisfying inequality \eqref{IN}. We say that the matrices $J_n$ form the Jacobi upper-triangular ensemble.

Parallel to Proposition \ref{P1}, we have
\begin{prop}
  Fix nonnegative integers $\alpha_1, \dotsc, \alpha_N$ such that \eqref{IN} is satisfied and fix $\beta_i = M' - i$ ($i = 1, \dotsc, N$), and consider $J_n$ and $\tilde{J}_n$ defined by \eqref{eq:defn_J_n} and \eqref{YA1} respectively with the same parameters $\alpha_1, \dotsc \alpha_n$. Then the joint distribution of $\tilde{J}_1, \tilde{J}_2, \dotsc, \tilde{J}_N$ is the same as the joint distribution of $J_1, J_2, \dotsc, J_N$. In particular, the joint distribution of the eigenvalues of $\tilde{J}_1, \dotsc, \tilde{J}_n$ is the same as the joint distribution of the eigenvalues of $J_1, \dotsc, J_N$.
\end{prop}
\begin{proof}
  The proof is analogous to that of Proposition \ref{P1}, and we divide it into two steps. As a bridge between $\tilde{J}_n$ and $J_n$, for all $n = 1, \dotsc, N$ we define
  \begin{equation} \label{eq:defn_J_hat_n}
    \hat{J}_n = (Y^{\dagger}_{N \times n} Y_{N \times n})(A^{\dagger}_{M' \times n} A^{\dagger}_{M' \times n} + Y^{\dagger}_{N \times n} Y_{N \times n})^{-1}.
  \end{equation}

  Recall the sequence of random matrices $X^{(0)} = X, X^{(1)}, \dotsc, X^{(N)}$ and random Householder matrices $U^{(1)}, \dotsc, U^{(N)}$ constructed in the proof of Proposition \ref{P1}. We have that for all $n = 1, \dotsc, N$, $(X^{(N)}_{M \times n})^{\dagger} X^{(N)}_{M \times n}$ are independent of $A$, and they have the same joint distribution as $Y^{\dagger}_{N \times n} Y_{N \times n}$. So the joint distribution of $\hat{J}_n$ is the same as that of 
  \begin{equation*}
    \tilde{J}^{(N)}_n = ((X^{(N)}_{M \times n})^\dagger X^{(N)}_{M \times n})(A^\dagger_{M' \times n} A_{M' \times n} + (X^{(N)}_{M \times n})^\dagger X^{(N)}_{M \times n})^{-1}.
  \end{equation*}
  On the other hand, $X^{(N)}$ is a function of $X$ and $(X^{(N)}_{M \times n})^{\dagger} X^{(N)}_{M \times n} = X^{\dagger}_{M \times n} X_{M \times n}$, so $\tilde{J}^{(N)}_n$ is identical to $\tilde{J}_n$ given that $X$ and $A$ are the same. Thus we have showed that the joint distribution of $\tilde{J}_n$ is the same as that of $\hat{J}_n$.

Next, since $A$ is also in the upper-triangular ensemble, we use the same algorithm in the proof of Proposition \ref{P1} to construct a sequence of random matrices $A^{(0)} = A, A^{(1)}, \dotsc, A^{(N)}$ and random Householder matrices $V^{(1)}, \dotsc, V^{(N)}$, such that $A^{(n + 1)} = V^{(n + 1)} A^{(n)}$ and $A^{(n)}$ satisfy the same conditions as $X^{(n)}$ but with the parameters $M, \alpha_1, \dotsc, \alpha_N$ replaced by $M', \beta_1 = M' -1, \dotsc, \beta_N = M' -N$. We have that for all $n = 1, \dotsc, N$, each $(A^{(N)}_{M' \times n})^{\dagger} A^{(N)}_{M' \times n}$ is independent of $Y$, and they have the same joint distribution as $Z^{\dagger}_{N \times n} Z_{N \times n}$. So the joint distribution of $J_n$ is the same as that of
  \begin{equation*}
    \hat{J}^{(N)}_n = (Y^{\dagger}_{N \times n} Y_{N \times n})((A^{(N)}_{N \times n})^{\dagger} A^{(N)}_{N \times n} + Y^{\dagger}_{N \times n} Y_{N \times n})^{-1}.
  \end{equation*}
  On the other hand, $A^{(N)}$ is a function of $A$ and $(A^{(N)}_{M' \times n})^{\dagger} A^{(N)}_{M' \times n} = A^{\dagger}_{M' \times n} A_{M' \times n}$. So $\hat{J}^{(N)}_n$ is identical to $\hat{J}_n$ given that $Y$ and $A$ are identical. Thus we show that the joint distribution of $\hat{J}_n$ is the same as that of $J_n$. 
\end{proof}

The following proposition is the Jacobi counterpart of Proposition \ref{P3}, and its proof is analogous to that of the Jacobi-\Pineiro\ part of \cite[Thm.~1]{Adler-van_Moerbeke-Wang11}.
\begin{prop}\label{P7}
  Let the $N \times N$ random matrices $Y$ and $Z$ be in the upper-triangular ensemble with diagonal entries specified by \eqref{yk} and \eqref{eq:diagonal_Z}, and thus with parameters $\alpha_1, \dotsc, \alpha_N$ and $\beta_1, \dotsc, \beta_N$ respectively. Denote by $\{\lambda_1,\dots,\lambda_n\}$
  and $\{\mu_1,\dots,\mu_{n-1}\}$ the eigenvalues of $J_n$ and
  $J_{n - 1}$ respectively in descending order, where $J_n$ and $J_{n - 1}$ are defined in \eqref{eq:defn_J_n}. The conditional PDF of 
  $\{\lambda_1,\dots,\lambda_n\}$, with $\{\mu_1,\dots,\mu_{n-1}\}$ fixed and distinct,
  is equal to
  
\begin{multline} \label{Pd4}
    p_{n,n-1}(\{\lambda_j\}_{j=1}^n, \{\mu_j \}_{j=1}^{n-1}) = \frac{\Gamma(\alpha_n + \beta_n + n + 1)}{\Gamma(\alpha_n + 1) \Gamma(\beta_n + 1)} \\
    \times \prod_{k=1}^n \lambda_k^{\alpha_n} (1 - \lambda_k)^{\beta_n} \prod_{l=1}^{n-1} \mu^{-\alpha_{n} - 1}_l (1 - \mu_l)^{-\beta_n - 1}
    {\prod_{1 \le j < k \le n} ( \lambda_j - \lambda_k) \over \prod_{1 \le j < k \le n-1} ( \mu_j - \mu_k) },
  \end{multline}
  subject to the interlacing constraint
  \begin{equation}\label{J2_Jacobi}
    1 \geq \lambda_1 \geq \mu_1 \geq \lambda_2 \geq \mu_2 \geq \cdots \geq \lambda_n \ge 0.
  \end{equation}
\end{prop}

Analogous to Proposition \ref{P3}, we actually prove a slightly stronger result:
\begin{lemma} \label{lem:P7}
  Let $C_{n - 1} = (c_{i, j})$ and $B_{n - 1} = (b_{i, j})$ be $(n - 1) \times (n - 1)$ invertible matrices such that the eigenvalues of $(C^{\dagger}_{n - 1} C_{n - 1})(B^{\dagger}_{n - 1} B_{n - 1} + C^{\dagger}_{n - 1} C_{n - 1})^{-1}$ are $\{ \mu_1, \dotsc, \mu_{n - 1} \}$ in descending order. Define the $n \times n$ matrix $C_n$ (\resp\ $B_n$) by letting: (1) the $(n - 1) \times (n - 1)$ upper-triangular block equal $C_{n - 1}$ (\resp\ $B_{n - 1}$); (2) the bottom row has all entries but the rightmost one equal $0$; (3) all entries of the rightmost column be independent random variables, such that the $(n, n)$-entry $c_{n, n}$ (\resp\ $b_{n, n}$) is real positive with $\lvert c_{n, n} \rvert^2 \mathop{=}\limits^{\rm d}  \Gamma[\alpha_n + 1, 1]$ (\resp\ $\lvert b_{n, n} \rvert^2 \mathop{=}\limits^{\rm d}  \Gamma[\beta_n + 1, 1]$), and all other entries in the column are in standard complex normal distribution. Then the eigenvalues of $(C^{\dagger}_n C_n)(B^{\dagger}_n B_n + C^{\dagger}_n C_n)^{-1}$, denoted by $\{ \lambda_1, \dotsc, \lambda_n \}$ in descending order, satisfy the interlacing constraint \eqref{J2_Jacobi} and have distribution given by $p_{n,n-1}(\{\lambda_j\}_{j=1}^n, \{\mu_j \}_{j=1}^{n-1})$ in \eqref{Pd4}.
\end{lemma}
\begin{proof}[Proof of Lemma \ref{lem:P7}]
  It is more convenient to consider for $* = n$ or $n - 1$, 
  \begin{equation*}
    T_* = (B^{\dagger}_* B_*)^{-1} C^{\dagger}_* C_* \quad \text{and} \quad R_* = (C_* B^{-1}_*)^{\dagger} (C_* B^{-1}_*), 
  \end{equation*}
  and their eigenvalues $\{ \tilde{\lambda}_1, \dotsc, \tilde{\lambda}_n \}$ for
  $T_n$ and $R_n$, and  
  $\{ \tilde{\mu}_1, \dotsc, \tilde{\mu}_{n - 1} \}$ for
  $T_{n - 1}$ and $R_{n - 1}$, assumed in ascending order, noting that $T_*$ and $R_*$ have the same eigenvalues. The relations between $\tilde{\lambda}_i, \tilde{\mu}_i$ and $\lambda_i, \mu_i$ are
  \begin{equation*}
    \tilde{\lambda}_i = \frac{\lambda_i}{1 - \lambda_i}, \quad \tilde{\mu}_i = \frac{\mu_i}{1 - \mu_i}, \quad \text{if $\lambda_i \neq 1$ and $\mu_i \neq 1$}.
  \end{equation*}
  It is straightforward to check that the proposition is equivalent to the statement that the conditional PDF of $\{ \tilde{\lambda}_1, \dotsc, \tilde{\lambda}_n \}$ is equal to
  \begin{multline} \label{eq:alternative_Jacobi}
    \tilde{p}_{n,n-1}(\{\tilde{\lambda}_j\}_{j=1}^n, \{\tilde{\mu}_j \}_{j=1}^{n-1}) = \frac{\Gamma(\alpha_n + \beta_n + n + 1)}{\Gamma(\alpha_n + 1) \Gamma(\beta_n + 1)} \\
    \times \prod_{k=1}^n \tilde{\lambda}_k^{\alpha_n} (1 + \tilde{\lambda}_k)^{-(\alpha_n + \beta_n + n + 1)} \prod_{l=1}^{n-1} \tilde{\mu}^{-\alpha_{n} - 1}_l (1 + \tilde{\mu}_l)^{\alpha_n + \beta_n + n}
    {\prod_{1 \le j < k \le n} ( \lambda_j - \lambda_k) \over \prod_{1 \le j < k \le n-1} ( \mu_j - \mu_k) },
  \end{multline}
  subject to the interlacing constraint
  \begin{equation} \label{eq:inter_J}
    0 < \tilde{\lambda}_1 \leq \tilde{\mu}_1 \leq \tilde{\lambda}_2 \leq \tilde{\mu}_2 \leq \dotsb \leq \tilde{\lambda}_n.
  \end{equation}

  For notational simplicity, we denote $\vec{y} = (c_{1, n}, c_{2, n}, \dotsc, c_{n - 1, n})^T$, $\vec{z} = (b_{1, n}, b_{2, n}$, $\dotsc, b_{n - 1, n})^T$, $\eta = c_{n, n}$ and $\zeta = b_{n, n}$. Then
  \begin{equation*}
    R_n = Q^{\dagger}Q, \quad \text{where} \quad Q =
    \begin{bmatrix}
      C_{n - 1} B^{-1}_{n - 1} & \zeta^{-1}(\vec{y} - C_{n - 1}B^{-1}_{n - 1} \vec{z}) \\
      \bbO_{1 \times (n - 1)} & \eta \zeta^{-1}
    \end{bmatrix}.
  \end{equation*}
  Taking the singular value decomposition to $C_{n - 1}B^{-1}_{n - 1}$, we have $(n - 1)$-dimensional unitary matrices $U, V$ such that
  \begin{equation*}
    C_{n - 1}B^{-1}_{n - 1} = UDV^{-1}, \quad \text{where $D = \diag[\sqrt{\tilde{\mu}_1}, \dotsc, \sqrt{\tilde{\mu}_{n - 1}}]$}.
  \end{equation*}
  Introducing
  \begin{equation*}
    U \oplus \bbI_1 =
    \begin{bmatrix}
      U & \bbO_{(n - 1) \times 1} \\
      \bbO_{1 \times (n - 1)} & 1
    \end{bmatrix},
    \quad 
    V \oplus \bbI_1 =
    \begin{bmatrix}
      V & \bbO_{(n - 1) \times 1} \\
      \bbO_{1 \times (n - 1)} & 1
    \end{bmatrix},
  \end{equation*}
  we have
  \begin{equation*}
    Q = (U \oplus \bbI_1)
    \begin{bmatrix}
      D & \zeta^{-1} \vec{w} \\
      \bbO_{1 \times (n - 1)} & \eta \zeta^{-1}
    \end{bmatrix}
    (V \oplus \bbI_1)^{-1}, \quad \text{where $\vec{w} = U^{-1} \vec{y} - DV^{-1} \vec{z}$}.
  \end{equation*}
  Analogous to \eqref{eq:char_equation_Laguerre}, the eigenvalue equation for the matrix $R_n = Q^{\dagger}Q$ is given by
  \begin{equation} \label{eq:char_equation_Jacobi}
    0 = \lambda + \left( -\frac{\eta^2}{\zeta^2} + \sum^{n - 1}_{k = 1} \frac{\lvert w_k \rvert^2}{\zeta^2} \frac{\lambda}{\tilde{\mu}_k - \lambda} \right),
  \end{equation}
  where the $w_k$ are components of $\vec{w}$. Note that $\lvert w_1 \rvert^2, \dotsc, \lvert w_{n - 1} \rvert^2, \eta^2, \zeta^2$ are independent, and their distribution functions are positive with densities
  \begin{equation*}
    \lvert w_k \rvert^2 \mathop{=}\limits^{\rm d} \frac{1}{1 + \tilde{\mu}_k} e^{-\frac{t}{1 + \tilde{\mu}_k}}, \quad \eta^2 \mathop{=}\limits^{\rm d} \frac{1}{\Gamma(\alpha_n + 1)} t^{\alpha_n} e^{-t}, \quad \zeta^2 \mathop{=}\limits^{\rm d} \frac{1}{\Gamma(\beta_n + 1)} t^{\beta_n} e^{-t}.
  \end{equation*}
  Comparing \eqref{eq:char_equation_Jacobi} with \cite[Eq.~(131)]{Adler-van_Moerbeke-Wang11}, and using the calculations in \cite[Eq.~(139)--(141)]{Adler-van_Moerbeke-Wang11}, we prove \eqref{eq:alternative_Jacobi}. (In \cite[Eqs.~(139)--(141)]{Adler-van_Moerbeke-Wang11} the calculations are done for integer valued $\alpha_n$ and $\beta_n$ which are denoted as $M' - n$. But the method works also for real-valued $\alpha_n$ and $\beta_n$.)
\end{proof}

This result can be used to derive the analogue of Corollary \ref{C1} in the Jacobi case.

\begin{corollary}\label{C2}
Let $\alpha_1, \dotsc, \alpha_N$ be arbitrary real numbers greater than $-1$ and
\begin{equation} \label{eq:beta_arith_progr}
  \beta_k = \beta + N - k \quad \text{for $k = 1, \dotsc, N$ where $\beta > -1$}.
\end{equation}
Let $J_n$ be defined in \eqref{eq:defn_J_n} for all $n = 1, \dotsc, N$.
\begin{enumerate}[label=(\alph*)]
\item \label{enu:cor:C_2:a}
  Then the eigenvalues of the random matrix $J_n$, denoted by $\lambda^{(n)} = \{ \lambda^{(n)}_1, \dotsc, \lambda^{(n)}_n \}$ in descending order, is equal to
  \begin{multline}\label{Pd3}
    p_n(\lambda^{(n)}_1, \dotsc, \lambda^{(n)}_n) = {C^{-1}_n \over \prod_{1 \le j < k \le n} (\alpha_j - \alpha_k)} \\
    \times \prod_{k=1}^n (1 - \lambda^{(n)}_k)^{\beta_n}
    \prod_{1 \le j < k \le n} (\lambda^{(n)}_j - \lambda^{(n)}_k)    \det[ (\lambda^{(n)}_j)^{\alpha_k}]_{j,k=1,\dots,n,}
  \end{multline}
  where
  \begin{equation}\label{Pd2a}
    C_n = \prod_{l=1}^n  { \Gamma(\alpha_l + 1) \Gamma(\beta_l + 1) \over  \Gamma(\alpha_l + \beta + N + 1) }.
  \end{equation}
  In the special case that  $\alpha_j$ is  given by (\ref{XY1}) this reduces to 
  \begin{multline}\label{Pd21}
    p_n(\lambda^{(n)}_1, \dotsc, \lambda^{(n)}_n) = \\
    {1 \over C_{n, \theta, c} }   \prod_{k=1}^n  (\lambda^{(n)}_k)^{c} (1 - \lambda^{(n)}_k)^{\beta_n}
    \prod_{1 \le j < k \le n} (\lambda^{(n)}_j - \lambda^{(n)}_k)   ((\lambda^{(n)}_j)^\theta - (\lambda^{(n)}_k)^\theta), 
  \end{multline}
  where
  \begin{equation}\label{Pd2b}
    C_{n, \theta, c} = \prod_{l=1}^n { \Gamma(\theta (l-1) + c + 1)  \Gamma(\beta_l + 1) \over
      \Gamma(\theta (l-1) + c + \beta + N + 1)}
    \, \,  \theta^{n(n-1)/2}\prod_{l=1}^{n-1} l!.
  \end{equation}
\item \label{enu:cor:C_2:b}
  The joint probability distribution function of $\lambda^{(1)}, \lambda^{(2)}, \dotsc, \lambda^{(N)}$ is equal to
  \begin{multline} \label{eq:jpdf_Jac}
    p(\lambda^{(1)}, \dotsc, \lambda^{(N)}) = \frac{1}{C_N} \prod_{1 \leq j < k \leq N} (\lambda^{(N)}_j - \lambda^{(N)}_k) \prod^N_{n = 1} (\lambda^{(N)}_n)^{\alpha_N} (1 - \lambda^{(N)}_n)^{\beta} \\
  \times  \prod^N_{n = 1} \left( \prod^n_{i = 1} (\lambda^{(n)}_i)^{\alpha_n - \alpha_{n + 1} - 1} \right) 1_{\lambda^{(n)} \preceq \lambda^{(n + 1)}},
  \end{multline}
  where $C_n$ is defined in \eqref{Pd2a}, and $1_{\mu \preceq \lambda}$ is the indicator function of the region that satisfies inequality \eqref{eq:inter_J}.
\end{enumerate}
In case that some $\alpha_i$ are identical, we understand the formulas in the limiting sense with \lHopital's rule. 
\end{corollary}
\begin{proof}[Proof of Part \ref{enu:cor:C_2:a}]
  We prove the case that the $\alpha_i$ are distinct, and the general result follows by analytical continuation.

  In the case $n=1$, $J_1$ defined in \eqref{eq:defn_J_n} is equal in distribution to $\chi_{2(\alpha_1 + 1)}^2/( \chi_{2(\beta_1 + 1)}^2  +
  \chi_{2(\alpha_1 + 1)}^2)$. It is a classical result \cite[Sec.~25.2]{Balakrishnan-Johnson-Kotz95} that this combination of random variables is distributed according to the beta distribution ${\rm B}[\alpha_1 + 1, \beta_1 + 1]$ and thus has for its PDF
  \begin{equation*}
    {\Gamma(\alpha_1 + \beta_1 + 2) \over \Gamma(\alpha_1+1) \Gamma(\beta_1 + 1)}
    x^{\alpha_1} (1 - x)^{\beta_1},
  \end{equation*}
  in agreement with (\ref{Pd3}) with $n=1$. Proceeding by induction, let us now assume that
  (\ref{Pd3}) is valid in the case $n-1$. For notational simplicity, we denote $\lambda^{(n)}_i$ by $\lambda_i$ and $\lambda^{(n - 1)}_i$ by $\mu_i$. Our task is to check the validity of 
  \begin{equation*}
    \int_{R^J_{\preceq \lambda}} p_{n,n-1}(\{\lambda_j\}_{j=1}^n,
    \{\mu_j \}_{j=1}^{n-1}) p_{n-1}(\mu_1,\dots,\mu_{n-1}) \, d\mu_1 \cdots d \mu_{n-1} =
    p_n(\lambda_1,\dots,\lambda_n),
  \end{equation*}
  which is analogous to \eqref{PE}, but the integral domain $R^J_{\preceq \lambda}$ for $\{ \mu_1, \dotsc, \mu_{n - 1} \}$ is defined by \eqref{J2_Jacobi}, and the $p_n$, $p_{n - 1}$ and $p_{n, n - 1}$ are defined differently. Recall the conditional PDF $p_{n,n-1}(\{\lambda_j\}_{j=1}^n, \{\mu_j \}_{j=1}^{n-1})$ of $\{ \lambda_1, \dotsc, \lambda_n \}$ with fixed $\{ \mu_1, \dotsc, \mu_{n - 1} \}$, such that the interlacing condition \eqref{J2_Jacobi} is satisfied, defined in \eqref{Pd4}.  Substituting for the integrand, then making use of Lemma \ref{L1} shows that the LHS is equal to
\begin{multline*}
  {1 \over C_{n-1}} 
  {\Gamma(\alpha_n + \beta_n + n + 1) \over \Gamma(\alpha_n) \Gamma(\beta_n + 1)} \\
  \times \prod_{1 \leq j < k \leq n} \frac{1}{\alpha_j - \alpha_k} \prod_{k=1}^n  (1 - \lambda_k)^{\beta_n}
  \prod_{1 \le j < k \le n} (\lambda_j - \lambda_k)
  \det[ \lambda_{j}^{\alpha_k}]_{j,k=1,\dots,n}. 
\end{multline*}
Comparison with (\ref{Pd3}) and (\ref{Pd2a}) shows that this is precisely the RHS.

We deduce (\ref{Pd21}) from (\ref{Pd3}) in the same way as we deduced (\ref{Pd2}) from (\ref{Pd1x}).
\end{proof}

\begin{proof}[Proof of Part \ref{enu:cor:C_2:b}]
  By Lemma \ref{lem:P7}, we have that the eigenvalues $\lambda^{(1)}, \lambda^{(2)}, \dotsc, \lambda^{(N)}$ constitute an inhomogeneous Markov chain with the transition probability density function from time $n - 1$ to time $n$ given by \eqref{Pd4}. Thus the joint distribution function of $\lambda^{(n)}$ ($n = 1, \dotsc, N$) is obtained by multiplying \eqref{Pd4} repeatedly. The argument is the same as the proof of \cite[Cor. 1]{Adler-van_Moerbeke-Wang11} and we omit the details.
\end{proof}
\begin{remark}
  The assumption \eqref{eq:beta_arith_progr} that $\beta_k$ are in arithmetic progression with common difference $-1$ is crucial in the application of Lemma \ref{L1}. By the symmetry of the model, it is also possible to let $\beta_k$ be arbitrary and $\alpha_k$ in arithmetic progression with common difference $-1$.
\end{remark}

From the joint probability distribution function \eqref{eq:jpdf_Jac}, we have, as a natural generalisation of \cite[Thm.~3(d)]{Adler-van_Moerbeke-Wang11}:
\begin{prop} \label{prop:corr_kernel_J}
  Let the matrices $J_n$ and the eigenvalues $\lambda^{(1)}, \dotsc, \lambda^{(N)}$ be defined as in Corollary \ref{C2}. The eigenvalues $\lambda^{(1)}, \dotsc, \lambda^{(N)}$ constitute a determinantal process, and the correlation kernel of $\lambda^{(n_1)}$ and $\lambda^{(n_2)}$ is given by
  \begin{multline} \label{eq:corr_kernel_J}
    K(n_1, x; n_2, y) = \frac{-1}{2\pi i} \oint_{\Gamma_{\alpha}} \frac{x^{-w - 1} y^w}{\prod^{n_2}_{l = n_1 + 1} (w - \alpha_l)} dw 1_{x < y} 1_{n_1 < n_2} \\
    + \frac{1}{(2\pi i)^2} \oint_{\Sigma} dz \oint_{\Gamma_{\alpha}} dw \frac{x^{-z - 1} y^w}{(z - w)} \frac{\Gamma(w + \beta + N + 1) \Gamma(z + 1)}{\Gamma(z + \beta + N + 1)\Gamma(w + 1)} \frac{\prod^{n_1}_{k = 1} (z - \alpha_k)}{\prod^{n_2}_{l = 1} (w - \alpha_l)},
  \end{multline}
  where
  \begin{enumerate}
  \item \label{enu:prop_corr_kernel_J:1}
    if $\beta \in \intZ$, $\Gamma_{\sigma}$ is a positively oriented contour enclosing $\alpha_1,\dots,\alpha_N$, while $\Sigma$ is a contour going counterclockwise enclosing $-1,-2,\dots, -(\beta + N)$ and the contour $\Gamma_\alpha$, and
  \item 
    if $\beta \notin \intZ$, $\Gamma_{\alpha}$ is a Hankel like contour, starting at $-\infty - i\epsilon$, running parallel to the negative real axis, enclosing the poles $w = -(\beta + N + k)$ with $k \in \intZ_+$ and $w = \alpha_1, \dotsc, \alpha_N$, and finishing at $-\infty + i\epsilon$ after again running parallel to the negative real axis, while $\Sigma$ is a Hankel like contour that loops around $\Gamma_{\alpha}$.
  \end{enumerate}
\end{prop}
The proof of the proposition is by a standard argument of determinantal process based on the joint probability density function \eqref{eq:jpdf_Jac}. In \cite[Thm.~3(d)]{Adler-van_Moerbeke-Wang11}, the proposition for non-negative integer $\beta$ and non-negative integer $\alpha_i$ under condition \eqref{IN} is proved. Since the proof does not use these additional conditions, it is also a complete proof to Proposition \ref{prop:corr_kernel_J}. Note that in \cite{Adler-van_Moerbeke-Wang11}, only case (\ref{enu:prop_corr_kernel_J:1}) of the contours occurs.

\medskip

It is possible to use Corollary \ref{C1}\ref{enu:cor:C_1:a} to give an alternative derivation of Corollary \ref{C2}\ref{enu:cor:C_2:a}, in the case that if $\beta_1, \dotsc, \beta_N$ satisfies \eqref{eq:beta_arith_progr} with $\beta  \in \intZ^+$. In this case we note that the eigenvalue PDF of $J_n$ is the same as the eigenvalue PDF of $\hat{J}_n$ defined in \eqref{eq:defn_J_hat_n}, where the height of the random matrix $A$ is $M' = \beta + N$. We also need a recent result due to Kuijlaars and Stivigny \cite{Kuijlaars-Stivigny14}. Below we give the derivation without the tedious calculation of the normalisation constant of the PDF.

\begin{prop}[Special case of {\cite[Thm.~2.1]{Kuijlaars-Stivigny14}}] \label{P5}
Let the matrix $W$ be an $n \times n$ random matrix such that $W^\dagger W$  has an eigenvalue PDF
proportional to
\begin{equation}\label{2.3w}
\prod_{1 \le j < k \le n} (x_k - x_j) \det [ f_{k-1}(x_j) ]_{j,k=1}^n
\end{equation}
for some $\{f_{k-1}(x) \}_{k=1,\dots,n}$.  For $\nu \ge 0$, let $G$ be an $(n + \nu) \times n$ random matrix whose entries are in independent standard complex Gaussian distribution. The squared singular values of $GW$, or equivalently the eigenvalues of $(GW)^{\dagger} GW$, have PDF proportional to
\begin{equation}\label{2.3wA}
\prod_{1 \le j < k \le n} (y_k - y_j)  \det [ g_{k-1}(y_j) ]_{j,k=1}^n,
\end{equation}
where
\begin{equation}\label{2.3wB}
g_k(y) = \int_0^\infty x^\nu e^{-x} f_k \Big ( {y \over x} \Big ) \, {dx \over x}, \qquad
(k=0,\dots,n-1).
\end{equation}
\end{prop}

\medskip
In the application of Proposition \ref{P5}, we let $W = Y^{-1}_{n \times n}$ and $G = A_{M' \times n}$, such that $\nu = M' - n = \beta_n = \beta + N - n$. In 
Corollary \ref{C1} we obtained the eigenvalue PDF of $Y^\dagger_{n \times n} Y_{n \times n}$. From this, by the change of variables $\lambda_j \mapsto 1/x_j$, we have that the eigenvalue PDF of $W^{\dagger} W =  (Y_{n \times n} Y^{\dagger}_{n \times n})^{-1}$ is proportional to
\begin{equation*}
  \prod_{k=1}^n x_k^{-(n+1)} e^{-1/ x_k}
  \prod_{1 \le j < k \le n} (x_k - x_j)    \det[ x_{j}^{-\alpha_k}]_{j,k=1,\dots,n},
\end{equation*}
and so we can take
\begin{equation*}
  f_k(x) = x^{-(n+1)} e^{-1/x} x^{-\alpha_k}.
\end{equation*}
Then substituting this in (\ref{2.3wB}) gives that
\begin{equation*}
  g_k(y) \propto y^{\nu} (1 + y)^{-(\nu + \alpha_k + n + 1)}.
\end{equation*}
Hence we deduce that the eigenvalue PDF of $(GW)^{\dagger} GW = (A_{M' \times n}Y^{-1}_{n \times n})^{\dagger} A_{M' \times n}Y^{-1}_{n \times n}$, which is the same as the eigenvalue PDF of
\begin{equation*}
  \hat{S}_n = (Y^\dagger_{n \times n} Y_{n \times n})^{-1} A_{M' \times n}^\dagger  A_{M' \times n},
\end{equation*}
has its eigenvalue PDF $p^{\hat{S}}_n(\sigma_1, \dotsc, \sigma_n)$ proportional to
\begin{equation*}
  \prod_{l=1}^n \sigma_l^\nu (1 + \sigma_l)^{-(\nu + n + 1)} \prod_{1 \le j < k \le n} (\sigma_k - \sigma_j)
  \det [ (1 + \sigma_j)^{-\alpha_k} ]_{j,k=1}^n.
\end{equation*}
The eigenvalues $\{\lambda_j\}$ of $\hat{J}_n$ and the eigenvalues $\{ \sigma_j \}$ of $\hat{S}_n$ are related by
\begin{equation}\label{B1}
  \lambda_j = {1 \over \sigma_j + 1}.
\end{equation}
Changing variables to $\{\lambda_j\}$ according to \eqref{B1} gives \eqref{Pd3} up to the normalisation constant $C_n$.

\section{The global density --- characteristic polynomial approach} \label{sec:characteristic_poly}

\subsection{The Laguerre Muttalib--Borodin ensemble}\label{S3.1}

Recent results \cite{Kuijlaars-Stivigny14,Forrester-Liu14} have revealed an intimate relationship between random
matrix products and the Laguerre Muttalib--Borodin ensemble. To explain this requires
the introduction of a family of integer sequences --- the Fuss--Catalan numbers --- parametrised by
$s \in \mathbb R^+$  and specified by
\begin{equation} 
 C_{s}(k)=\frac{1}{sk+1}\binom{sk+k}{k}, \qquad k=0,1,2,\dots \label{FCnumber}
 \end{equation}
For general $s > 0$ these are known to be moments of a PDF --- the Fuss--Catalan density --- with compact support $[0,L]$, $L > 0$ \cite{Banica-Belinschi-Capitaine-Collins11, Mlotkowski10}, and they uniquely define the PDF.

Consider first a product of $s$ $N\times N$ matrices $X_1,\dots,X_s$ with each containing independent, identically distributed zero mean, unit standard deviation random variables. Alternatively, for one such matrix $X_1$ say, consider the power $X_1^s$. In either case, ask for the limiting spectral density of the squared singular values after dividing by $N^s$ --- what results is precisely the Fuss--Catalan density with parameter $s$  \cite{Alexeev-Gotze-Tikhomirov10,Banica-Belinschi-Capitaine-Collins11,Nica-Speicher06}. 

Consider now the Laguerre Muttalib--Borodin ensemble defined by \eqref{1.1} and \eqref{L}. Let $\lambda^{(N)}_1, \dotsb, \lambda^{(N)}_N$ be the eigenvalues in the $N$-dimensional ensemble. As $N \to \infty$, by standard techniques we have that the empirical distribution of the scaled eigenvalues $N^{-1}\lambda^{(N)}_1, \dotsc, N^{-1}\lambda^{(N)}_N$ converges in distribution to a limiting probability distribution, also known as the equilibrium measure of the model; see \cite[Sec.~6.4]{Deift99}. The equilibrium measure is characterised as the minimum of a variation problem; see Claeys and Romano \cite[Eq.~(1.22)]{Claeys-Romano14}. By interpreting the recent results of \cite{Claeys-Romano14}, Forrester and Liu \cite{Forrester-Liu14} have identified the global density (i.e.~density scaled by an appropriate power of $N$ to have compact support) for the Laguerre Muttalib--Borodin ensemble in terms of the  Fuss--Catalan  density for general $s \ge 1$.
 
 \begin{prop}\label{Pfl}
Suppose $\theta \ge 1$ (this is for technical reasons in the working of \cite{Claeys-Romano14}; in \cite{Forrester-Liu14} it is commented that the same result is expected to hold for all $\theta > 0$ and in fact this has recently been established in \cite{Forrester-Liu-Zinn_Justin14}). After changing variables $x^{(N)}_k = (\lambda^{(N)}_k/(N\theta))^{\theta}$ where $\lambda^{(N)}_1, \dotsc, \lambda^{(N)}_N$ are the eigenvalues in the Laguerre Muttalib--Borodin ensemble defined by \eqref{1.1} and \eqref{L}, the empirical distribution of $x^{(N)}_1, \dotsc, x^{(N)}_N$ converges to the Fuss--Catalan density with parameter $\theta$ as $N \to \infty$.
 \end{prop}
 
In particular, this shows a relationship between the product of $s$ random matrices and the Laguerre Muttalib--Borodin ensemble with $\theta = s$ (see also \cite{Kuijlaars-Stivigny14}
and the appendix in \cite{Forrester-Liu-Zinn_Justin14}). Here we will demonstrate the relationship  in a different way, by considering the characteristic polynomial of the latter after the change of variables \eqref{Cv}. The proof of Proposition \ref{Pfl} via Claeys and Romano's approach is valid for all $\theta \ge  1$ at least, but depends on the construction of a mapping $J(s)$ \cite[Eq.~(1.25)]{Claeys-Romano14}, so it is unclear if it can be applied to the Jacobi case. On the other hand, the approach presented below is valid only for $\theta \in \intZ^+$, but it does generalise to the Jacobi case.

Crucial for the alternative proof is knowledge of certain biorthogonal polynomials associated with \eqref{1.1} where $V(x)$ is given in \eqref{L}. Thus for given $j=0,1,2,\dots$ let $p_j(x)$, $q_j(x)$ be monic polynomials of degree $j$,
and suppose these polynomials have the biorthogonal property 
\begin{equation}\label{PL} 
\int_0^\infty e^{-V(x)} p_j(x) q_k(x^\theta) \, dx = h_j \delta_{j,k}, \qquad h_j > 0,
 \end{equation}
 where the positivity of $h_j$ follows from the positivity of the integral over (\ref{1.1}), because $n! h_0 h_1 \dotsb h_{n - 1}$ is equal to the integral over \eqref{1.1}.
 Let $\mathcal P_{N,\theta}$ denote the PDF specified by (\ref{1.1}). Straightforward working 
 (see e.g.~\cite[Prop.~5.1.3]{Forrester10}) shows that
\begin{equation} \label{eq:p_q_as_average}
  \Big \langle \prod_{l=1}^N ( x - \lambda^{(N)}_l) \Big \rangle_{\mathcal P_{N,\theta}} = p_N(x), \qquad
  \Big \langle \prod_{l=1}^N ( x - (\lambda^{(N)}_l)^\theta) \Big \rangle_{\mathcal P_{N,\theta} }= q_N(x).
 \end{equation}
 
 Since Proposition \ref{Pfl} requires the change of variables (\ref{Cv}), we see that $q_N(x)$ is equal to the corresponding averaged characteristic polynomial. For the Laguerre weight (\ref{L}), the
 explicit form of $q_N(x)$ is known from a result of Konhauser \cite{Konhauser67}. A relation between $q_N(x)$ and generalised hypergeometric functions is observed in \cite{Srivastava82}. We will use the standard notation ${}_pF_q$ to denote the generalized hypergeometric function defined by a series as presented in e.g.~\cite[Sec.~3]{Luke69}. 
 
 \begin{prop} \label{prop:biorthogonal_poly_algebraic_Laguerre}
 For the Laguerre weight (\ref{L}), the biorthogonal polynomials $\{q_j(x)\}$ in (\ref{PL}) are
 given by
\begin{equation}
q_j(x) = (-1)^j \Gamma(\theta j + c + 1)
\sum_{l=0}^j (-1)^l \binom{j}{l} {x^l \over \Gamma(\theta l + c + 1)}.
\end{equation}
In the case that $\theta \in \mathbb Z^+$, use of the duplication formula for the gamma function
shows that
\begin{equation}\label{qF}
q_j(x) = C \cdot {}_1 F_\theta \bigg ( {\displaystyle -j \atop  \displaystyle (c+1)/\theta, (c+2)/\theta, \dots,
(c + \theta)/\theta} \bigg | {x \over \theta^\theta} \bigg ),
\end{equation}
 where $C$ is independent of $x$ and chosen so that $q_j(x)$ is monic.
\end{prop}

Below we will make use of the standard fact (see e.g.~\cite[Sec.~5.1]{Luke69}
) that the generalized hypergeometric function ${}_p F_q({a_1,\dots,a_p \atop b_1,\dots, b_q} | x)$ satisfies the differential equation
\begin{equation} \label{DpFq}
  x \prod_{n=1}^p \Big ( x {d \over dx} + a_n \Big ) f =  x {d \over dx} \prod_{n=1}^q \Big ( x {d \over dx} + b_n - 1 \Big ) f .
\end{equation}
Also, we require a technical result relating to the convergence of the Stieltjes transforms of the empirical distributions of $\lambda^{(N)}_1, \dotsc, \lambda^{(N)}_N$.
\begin{lemma} \label{lem:equilibrium}
  For all $z$ in a compact subset of $\compC \setminus [0, \infty)$, as $N \to \infty$, uniformly
  \begin{equation} \label{3.4}
    \begin{split}
      \lim_{N \to \infty} \frac{1}{N} \log \left\langle \prod^N_{n = 1} \left( z - \left( \frac{\lambda^{(N)}_n}{N\theta} \right)^{\theta} \right) \right\rangle = {}& \int_J \log \left( z - \left( \frac{x}{\theta} \right)^{\theta} \right) d\mu^{\mathrm{L}}(x) \\
      = {}& \int_{\tilde{J}} \log(z - \tilde{x}) \ d\tilde{\mu}^{\mathrm{L}}(\tilde{x}),
    \end{split}
  \end{equation}
  where $d\mu^{\mathrm{L}}$ denotes the equilibrium measure, $J$ is the corresponding support, $d\tilde{\mu}^{\mathrm{L}}$ is the measure transformed from $d\mu^{\mathrm{L}}$ by the change of variable $\tilde{x} = (x/\theta)^{\theta}$, and $\tilde{J}$ is the support of $d\tilde{\mu}^{\mathrm{L}}$.
\end{lemma}
The proof of Lemma \ref{lem:equilibrium} is similar to \cite[Lem.~6.77]{Deift99}. Note that in \cite{Deift99} it is required that the function $\phi$, corresponding to the function $\log(z - (x/\theta)^{\theta})$, is a bounded function in $x$, while our $\log(z - (x/\theta)^{\theta})$ is not. But since the growth of the function as $x \to \infty$ is mild, the argument there can be applied. Note that for $z$ in a compact subset of $\compC \setminus [0, \infty)$, the convergence in \eqref{3.4} is uniform, since the functions are uniformly bounded and equi-continuous. So if we take derivative on both sides of \eqref{3.4}, the convergence still holds. Comparing the left-hand side of \eqref{3.4} with the formula \eqref{eq:p_q_as_average} for $q_N$, we have that 
\begin{equation} \label{3.4a}
  \lim_{N \to \infty} \frac{1}{N} \frac{\frac{d}{dz} q_N((N\theta)^{\theta} z)}{q_N((N\theta)^{\theta} z)} = \lim_{N \to \infty} \frac{(N\theta)^{\theta}}{N} {q_N'(x) \over q_N(x) } \Big |_{x = (N \theta)^\theta z} = \tilde{G}^{\mathrm{L}}(z), \qquad z \in \compC \setminus [0, \infty),
\end{equation}
where
\begin{equation} \label{eq:defn_tilde_G^L}
  \tilde{G}^{\mathrm{L}}(z) = \int_{\tilde{J}} \frac{d\tilde{\mu}^{\mathrm{L}}(\tilde{x})}{z - \tilde{x}}
\end{equation}
is the limiting resolvent (or equivalently Stieltjes transform) of the measure $d\tilde{\mu}$. Below we show that $\tilde{G}^{\mathrm{L}}(z)$ satisfies a polynomial equation 
which uniquely characterises the Fuss--Catalan distribution.

 \begin{prop}\label{PC1}
   Transform the Laguerre Muttalib--Borodin ensemble according to \eqref{Cv}. Define $\tilde{G}^{\mathrm{L}}(z)$ by \eqref{eq:defn_tilde_G^L} so that it is equal to the resolvent corresponding to the global scaled density, scaling $x = (N \theta)^{\theta} z$. For $\theta \in \mathbb Z^+$ we have that
 \begin{equation} \label{3.4b}
z ( z \tilde{G}^{\rm L}(z) - 1) = (z \tilde{G}^{\rm L}(z))^{\theta + 1}.
\end{equation}
\end{prop}

\begin{proof}
  First we note that the convergence \eqref{3.4a} and \eqref{eq:defn_tilde_G^L} yields that as $N \to \infty$, for $z \in \compC \setminus [0, \infty)$,
  \begin{equation*}
    q_N((N\theta)^{\theta} z) = \exp \left( N \left( \int \tilde{G}^{\mathrm{L}}(z) dz + C + o(1) \right) \right),
  \end{equation*}
  where $o(1)$ is an analytic function in $z$ that vanishes uniformly for $z$ in any compact set of $\compC \setminus [0, \infty)$, and $C$ is a constant. The asymptotic formulas below are thus justified.
  
  Suppose $x = \bigO(N^{\theta})$ and $\arg x \in (0, 2\pi)$, then to leading order in $N$, after substituting for $\{a_i\}$, $\{b_j\}$ as
  implied by (\ref{qF}), we see that (\ref{DpFq})  reads
  \begin{equation} \label{3.4c}
    {x \over \theta^\theta }\Big ( x {d \over dx} - N \Big ) q_N =  \Big (x {d \over dx} \Big )^{\theta + 1} q_N + o(N^{\theta + 1}).
  \end{equation}
  Now change variables  $x= (N\theta)^\theta z$, and let
  $q_N^{(k)}(x)$ denote the $k$-th derivative of $q_N(x)$.
  Working contained in \cite[above Prop.~5.1]{Forrester-Liu14}, establishes that under the validity of
  (\ref{3.4a}),
  \begin{equation}\label{u}
    \frac{\frac{d^k}{dz^k} q_N((N\theta)^{\theta} z)}{q_N((N\theta)^{\theta} z)} = (N\theta)^{k\theta} {q_N^{(k)}(x) \over q_N(x) } \Big |_{x = (N\theta)^\theta z} \sim  N^k (\tilde{G}^{\rm L}(z))^k
  \end{equation}
  (note that (\ref{3.4a}) itself is the case $k=1$; the general $k$ case follows by expressing higher derivatives in terms of the logarithmic derivatives).  Using this in (\ref{3.4c}) gives (\ref{3.4b}).
\end{proof}

\begin{remark} \label{rmk:Fuss--Catalan_distr} 
  \begin{enumerate}[label=(\roman*)]
  \item 
    The large $z$ expansion
    \begin{equation} \label{eq:expansion_G^L}
      z\tilde{G}^{\rm L}(z) = \sum_{k=0}^\infty \frac{\tilde{m}^{\mathrm{L}}_k}{z^k},
    \end{equation}
    where $\tilde{m}^{\mathrm{L}}_k$
    denotes the $k$-th moment of the global scaled spectral measure, substituted in (\ref{3.4b})
    shows that $\tilde{m}^{\mathrm{L}}_k$ is equal to the $k$-th Fuss--Catalan number (\ref{FCnumber}) with
    $s = \theta$. The details of the required calculation can be found in e.g.~\cite[paragraph beginning with eq.~(2.4)]{Forrester-Liu14}.
  \item 
    The spectral density is the global scaled limit of the one-point correlation function. Let $\lambda^{(N)}_1 = \lambda, \lambda^{(N)}_2, \dotsc, \lambda^{(N)}_N$ be the eigenvalues in the 
    $N$-dimensional Laguerre Muttalib--Borodin ensemble. Then the one-point correlation function is (we suppress the $N$-dependence in the notation $\rho_{(1)}$)
    \begin{equation}\label{3.14a}
      \rho_{(1)}(\lambda) = N \int_0^\infty d\lambda^{(N)}_2 \dotsm \int_0^\infty d\lambda^{(N)}_N \, p_N(\lambda, \lambda^{(N)}_2, \dots, \lambda^{(N)}_N),
    \end{equation}
    where $p_N$ is the PDF implied by (\ref{1.1}) with the Laguerre weight (\ref{L}). The scaled spectral density, as $N \to \infty$, converges to the equilibrium measure in distribution:
    \begin{equation} \label{eq:limit_counting_measure}
      d\mu^{\mathrm{L}}(x) = \lim_{N \to \infty} \rho_{(1)}(Nx) dx,
    \end{equation}
    where $d\mu^{\mathrm{L}}$ is the same as in \eqref{3.4}.
  \item 
    The resolvent $\tilde{G}^{\mathrm{L}}(z)$ defined in \eqref{eq:defn_tilde_G^L} can be expressed by the limiting spectral density/equilibrium measure of the Laguerre Muttalib--Borodin ensemble as
    \begin{equation}\label{3.14aS}
      \tilde{G}^{\mathrm{L}}(z) = \int_J {d\mu^{\mathrm{L}}(x) \over z - (x/\theta)^{\theta}}.
    \end{equation}
  \end{enumerate}
\end{remark}

\bigskip
There is another viewpoint on deducing the global spectral density from knowledge of (\ref{qF}). This makes use of
a recent result of Hardy \cite{Hardy15}, which subject to a mild technical condition \cite[Eq.~(1.12)]{Hardy15} states that
characteristic polynomials coming from a class of determinantal point processes including the multiple orthogonal polynomial ensemble under present discussion (see \cite[Sec.~1.4]{Hardy15}), the limiting density of zeros equals the limiting global spectral density. On the other hand, the limiting density of zeros for the generalised hypergeometric function ${}_1 F_\theta$ in (\ref{qF}) has been shown by Neuschel \cite{Neuschel14} to be given by the Fuss--Catalan density with parameter $s = \theta$. (Strictly speaking the result of \cite{Neuschel14} assumes the bottom line of parameters in (\ref{qF}) to be positive integers. Since the leading asymptotics are independent of these parameters, it is expected that this assumption in \cite{Neuschel14} is not necessary.)

\subsection{The Jacobi Muttalib--Borodin ensemble}\label{S3.2}

As with the Laguerre weight (\ref{L}), for the Muttalib--Borodin ensemble with Jacobi weight (\ref{J}), we also let $\lambda^{(N)}_1, \dotsc, \lambda^{(N)}_N$ be the eigenvalues in the 
$N$-dimensional ensemble. As $N \to \infty$, by standard techniques, we have that the empirical distribution of the eigenvalues converges in distribution to a limiting probability distribution, also known as the equilibrium measure of this model. The equilibrium measure can be characterized as the minimum 
of a variation problem  analogous to that for the Laguerre case, and from this, in
 the recent work \cite{Forrester-Liu-Zinn_Justin14} the moments of the corresponding density have been given
in terms of certain binomial coefficients (see Proposition \ref{p3.10} below).
Analogous to the Laguerre ensemble, there is a corresponding system of biorthogonal polynomials. For $j = 0, 1, 2, \dotsc$, let $p_j(x), q_j(x)$ be monic polynomials of degree $j$, and with weight $V(x)$  defined in \eqref{J} suppose these polynomials have the biorthogonal property analogous to \eqref{PL}, 
\begin{equation}\label{eq:PL_Jacobi} 
\int^1_0 e^{-V(x)} p_j(x) q_k(x^\theta) \, dx = h_j \delta_{j,k}, \qquad h_j > 0.
 \end{equation}
Let $\mathcal P_{N,\theta}$ denote the PDF specified by \eqref{1.1} with $V$ specified by \eqref{J} and $0 < \lambda_l < 1$. Then \eqref{eq:p_q_as_average} also holds in the Jacobi ensemble with corresponding different meanings of $p_N, q_N$ and $\mathcal{P}_{N, \theta}$.
Below we state algebraic results on the biorthogonal polynomials analogous to Proposition \ref{prop:biorthogonal_poly_algebraic_Laguerre}. Again we will focus on the polynomials $q_j(x)$.

\begin{prop} \cite{Chat72, Carlitz73, Madhekar-Thakare82}
  For the Jacobi weight (\ref{J}), the biorthogonal polynomials $\{q_j(x)\}$ in (\ref{PL}) are
  given by
  \begin{equation}
    q_j(x) = (-1)^j {(1 + c_1)_{\theta j} \over (1 + c_1 + c_2 + j)_{\theta j}}
    \sum_{l=0}^j (-1)^l \binom{j}{l} {(1 + c_1 + c_2 + j)_{\theta l} \over (1 + c_1)_{\theta l}} x^l, \:\: 
  \end{equation}
  where $(a)_p := \Gamma(a+p)/\Gamma(a)$.
  In the case that $\theta \in \mathbb Z^+$, use of the duplication formula for the gamma function
  shows that
  \begin{equation}\label{qF6}
    q_j(x) \propto {}_{\theta + 1}  F_\theta \bigg ( {\displaystyle -j , \delta/\theta, (\delta + 1)/\theta, \dots,
      (\delta + \theta - 1)/\theta \atop  \displaystyle (c_1+1)/\theta, (c_1+2)/\theta, \dots,
      (c_1 + \theta)/\theta} \bigg | x \bigg ),
  \end{equation}
  where $\delta : = 1 + c_1 + c_2 + j$.
\end{prop}

Also we have the analogue of Lemma \ref{lem:equilibrium}.
\begin{lemma} \label{lem:equilibrium_Jacobi}
  For all $z \in \compC \setminus [0, 1]$, as $N \to \infty$,
  \begin{equation}
    \begin{split}
      \lim_{N \to \infty} \frac{1}{N} \log \left\langle \prod^N_{n = 1} \left( z - (\lambda^{(N)}_n)^{\theta} \right) \right\rangle = {}& \int_{[0, 1]} \log (z - x^{\theta}) d\mu^{\mathrm{J}}(x) \\
      = {}& \int_{[0, 1]} \log (z - \tilde{x}) d\tilde{\mu}^{\mathrm{J}}(\tilde{x}),
    \end{split}
  \end{equation}
  where $d\mu^{\mathrm{J}}$ denote the equilibrium measure of the Jacobi ensemble and $d\tilde{\mu}^{\mathrm{J}}$ denote the measure transformed from $d\mu^{\mathrm{J}}$ by the change of variable $\tilde{x} = x^{\theta}$. Note that both $d\mu^{\mathrm{J}}$ and $d\tilde{\mu}^{\mathrm{J}}$ have support $[0, 1]$.
\end{lemma}
The proof of Lemma \ref{lem:equilibrium_Jacobi} is similar to \cite[Lem.~6.77]{Deift99}, and we omit the details. From Lemma \ref{lem:equilibrium_Jacobi}, we derive the counterpart of \eqref{3.4a}, that
\begin{equation}
  \lim_{N \to \infty} \frac{1}{N} \frac{q'_N(z)}{q_N(z)} = \tilde{G}^{\mathrm{J}}(z), \quad \lambda \in \compC \setminus [0, 1],
\end{equation}
where
\begin{equation} \label{eq:defn_tilde_G^J}
  \tilde{G}^{\mathrm{J}}(z) = \int_{[0, 1]} \frac{d\tilde{\mu}^{\mathrm{J}}(\tilde{x})}{z - \tilde{x}}
\end{equation}
is the limiting resolvent of the measure $d\tilde{\mu}^{\mathrm{J}}$. Below we show that $\tilde{G}^{\mathrm{J}}(z)$ satisfies a polynomial equation that is similar to \eqref{3.4b} characterising the Fuss--Catalan distribution.

\begin{prop}\label{P3.8}
  Transform the Jacobi Muttalib--Borodin ensemble according to \eqref{Cv}. Define $\tilde{G}^{\mathrm{J}}(z)$ by \eqref{eq:defn_tilde_G^J} so that it is equal to the resolvent corresponding to the global density. For $\theta \in \intZ^+$ we have that
   \begin{equation}\label{G8}
    z ( z \tilde{G}^{\rm J}(z) - 1)( z\tilde{G}^{\rm J}(z) + 1/\theta)^\theta = (z \tilde{G}^{\rm J}(z))^{\theta + 1}.
  \end{equation}
\end{prop}

\begin{proof}
  Applying the identity \eqref{DpFq} to \eqref{qF6} with $j = N$, we have, to leading order in $N$,
  \begin{equation}\label{u1}
    z \Big ( z {d \over dz} - N \Big )  \Big ( z {d \over dz} + {N \over \theta} \Big )^\theta 
    q_N =  \Big ( z {d \over dz} \Big )^{\theta + 1} q_N,
  \end{equation}
  while the analogue of (\ref{u}) is 
  \begin{equation}\label{u2}
    {q_N^{(k)}(z) \over q_N(z) } \sim  N^k (\tilde{G}^{\rm J}(z))^k.
  \end{equation}
  Use of (\ref{u2}) in (\ref{u1}) gives (\ref{G8}).
\end{proof}

\begin{remark}\label{R3.9}
  \begin{enumerate}[label=(\roman*)]
  \item 
    Taking the limit $\theta \to \infty$ in (\ref{G8}) gives the nonlinear equation 
\begin{equation}\label{G9}
\Big ( {1 \over z \tilde{G}^{\rm J}(z)} - 1 \Big ) \exp \Big ( {1 \over z \tilde{G}^{\rm J}(z)} - 1 \Big ) = - {1 \over z e}.
\end{equation}
By definition, the Lambert $W$-function $W(z)$ is the principal branch of the functional equation
$z = W(z) e^{W(z)}$ defined in $\mathbb C\backslash (-\infty, - e^{-1})$, and thus we have in this case
\begin{equation}\label{G10}
  \lim_{\theta \to \infty} {1 \over z \tilde{G}^{\rm J}(z)} - 1 = W \Big ( - {1 \over z e} \Big ).
 \end{equation}
\item 
  Generally we expect the global density associated with the weight (\ref{J}) to be independent of
$c_1$ and $c_2$ provided those parameters are themselves independent of $N$. On the other hand, the
change of variables (\ref{Cv}) shows that for finite $N$ the density $\rho_{(1)}(x)$, defined by (\ref{3.14a})
with $p_N$ the PDF implied by (\ref{1.1}) with the Jacobi weight (\ref{J}), has the functional property
\begin{equation}\label{3.14b}
  \rho_{(1)}(x)  \Big |_{c_2 = 0} = \theta x^{\theta - 1}  \rho_{(1)}(x^\theta)  \Big |_{\substack{c_1 \mapsto (c_1+1)/\theta - 1, \\ c_2 = 0, \theta \mapsto 1/\theta}}.
\end{equation}
\item 
For $\theta = 1$ (\ref{G8}) reduces to the simple quadratic equation $(z-1) (z \tilde{G}^{\rm J}(z))^2 = z$ and thus
$z \tilde{G}^{\rm J}(z) = 1/ (1 - 1/z)^{1/2}$, where we have imposed the requirement that $z G(z) \to 1$ as $z \to \infty$ in choosing the root of the
quadratic. Recalling (\ref{3.14aS}), this implies the well known functional form for the density in the classical Jacobi ensemble
(see e.g.~\cite[Prop.~3.6.3]{Forrester10})
\begin{equation}\label{Aj}
d\tilde{\mu}^{\rm J}(y)  = d\mu^{\rm J}(y)  = {1 \over \pi} {1 \over \sqrt{y (1 - y)}} dy, \qquad 0 < y < 1.
\end{equation}
  \end{enumerate}
\end{remark}

\medskip
A feature of (\ref{Aj}) is that the moments are given in terms of a binomial coefficient,
\begin{equation}\label{Aj1}
{1 \over \pi} \int_0^1 {y^p \over \sqrt{y(1-y)}} \, dy = {1 \over 2^{2p}} \binom{2p}{p}.
\end{equation}
In fact it is possible to show that the moments of the density implied by the appropriate solution of (\ref{G8}) are given in terms of binomial
coefficients for general $\theta > 0$. 

By the definition \eqref{eq:defn_tilde_G^J} of $\tilde{G}^{\mathrm{J}}(z)$, for $|z|>1$ we have the expansion
   \begin{equation}  \label{zz0}
z \tilde{G}^{\rm J}(z) =  \sum_{p=0}^\infty {\tilde{m}_p^{\rm J} \over z^p},
\end{equation}
where $\tilde{m}_p^{\rm J}$ denotes the $p$-th moment of $\tilde{\mu}^{\rm J}(x)$, analogous to \eqref{eq:expansion_G^L}. To compute $\{ \tilde{m}_p^{\rm J} \}$ we are guided by the knowledge that the moments (\ref{FCnumber}) corresponding to the density for the Laguerre Muttalib--Borodin ensemble can be computed from the functional equation (\ref{3.4b}) by using the Lagrange inversion formula (see e.g.~\cite[Sec.~2]{Forrester-Liu14}). The setting of the latter requires two analytic functions  $f(z)$ and $\phi(z)$ in a neighbourhood $\Omega$ of a point $a$, and $t$ to be small enough so that $|t \phi(z)| < |z - a|$, $z \in \Omega$. It tells us that the equation in $\zeta$
 \begin{equation} \label{zz}
 \zeta = a + t \phi(\zeta)
 \end{equation}
 has one solution in $\Omega$ and furthermore
   \begin{equation}  \label{zz1}
  f(\zeta) = f(a) + \sum_{n=1}^\infty {t^n \over n!} {d^{n-1} \over d a^{n-1}} (f'(a) (\phi(a))^n ).
  \end{equation}

\begin{prop}\label{p3.10} 
Let $\theta > 0$ and $L := (1 + \theta)^{1 + \theta} \theta^{-\theta}$, and define $\tilde{G}^{\rm J}(z)$ as the solution of  (\ref{G8})  with the expansion (\ref{zz0}) for
$\theta > 0$.
We have
\begin{equation}\label{mL}
 \tilde{m}_p^{\rm J} = L^{-p}  \binom{(1 + \theta)p }{ p}.
\end{equation}
\end{prop}

\begin{proof}
  Let $1/z = Lt$ and $X = z \tilde{G}^{\rm J}(z)$. Then by \eqref{zz0}, $X$ is a power series in $1/z$ and also a power series in $t$. Simple manipulation of (\ref{G8}) shows 
  \begin{equation}\label{Xt}
    X = t \phi(X), \quad \text{where} \quad \phi(z) = (1 + \theta) {(1 + z)^{\theta + 1} \over (1 + \theta z/ (1 + \theta))^\theta},
  \end{equation}
  which is of the form (\ref{zz}) with $a=0$. Applying (\ref{zz1}) with $f(\zeta) = \zeta$ and recalling (\ref{zz0})
  shows that
  \begin{equation}\label{Xt1}
    \tilde{m}_p^{\rm J} = \left. \frac{L^{-p}}{p!} \frac{d^p X}{dt^p} \right\rvert_{t = 0} = L^{-p} {(1 + \theta)^p \over p!} \left. {d^{p-1} \over d z^{p-1}} ( {(1 + z)^{\theta + 1} \over (1 + \theta z/(1 + \theta))^\theta} \Big )^p \right\rvert_{z = 0}.
  \end{equation}
  
 Using the binomial theorem to expand the two main factors on the right-hand side in (\ref{Xt1}) into power series in $z$, then combining the coefficients appropriately to form a single power series shows
  \begin{equation} 
    \tilde{m}_p^{\rm J} = L^{-p} {(1 + \theta)^p \over p}
    \sum_{q=0}^{p-1} \binom{(\theta + 1)p}{q}   \binom{- \theta p}{p-1-q}
    \Big ( {\theta \over 1 + \theta} \Big )^{p-1-q}.
  \end{equation}
  The sum can be recognised as a polynomial example of a particular ${}_2 F_1$ Gaussian
  hypergeometric function, allowing us to write
  \begin{equation}\label{Xt2}
    \tilde{m}_p^{\rm J} = L^{-p} { \theta^p \over p}
    {(1 + \theta) \over \theta}
    \binom{- \theta p}{p-1} \,
    {}_2 F_1 \Big ( {1 - p, -p(1 + \theta) \atop 2 - p (1 + \theta)} \Big | {1 + \theta \over \theta} \Big ).
  \end{equation}
  
  The functional equation for the gamma function shows
  \begin{equation}\label{Xt3}
    \binom{- \theta p}{p-1} = { (-1)^{p-1} p \over p (\theta + 1) - 1}
    \binom{p(1 + \theta) -1 }{p}.
  \end{equation}
  Also, in general, it is a simple exercise to verify from the series definition of $ {}_2 F_1$ that
  \begin{multline*}
    {}_2 F_1(a,b;c;z) = {2b - c + 2 + (a - b - 1)z \over b - c + 1} \, {}_2 F_1(a,b+1;c;z)  \\
    + {(b+1)(z-1) \over b - c + 1} \, {}_2 F_1(a,b+2;c;z).
  \end{multline*}
  In the case of the ${}_2 F_1$ in (\ref{Xt2}) we have
  $
  2b - c + 2 + (a - b - 1)z =0$, $b+2 = c$. Thus on the RHS, only the second term contributes, and
  furthermore it can be simplified from the fact that ${}_2 F_1(a,b+2;b+2;z) = (1 - z)^{-a}$, implying the
  result
  \begin{equation}\label{Xt4}
    {}_2 F_1 \Big ( {1 - p, -p(1 + \theta) \atop 2 - p (1 + \theta)} \Big | {1 + \theta \over \theta} \Big ) =
    (-1)^{p-1}  {p(1 + \theta) - 1 \over \theta^p}.
  \end{equation}
  Substituting (\ref{Xt3}) and    (\ref{Xt4}) into (\ref{Xt2}) we obtain (\ref{mL}). 
\end{proof}

 \begin{remark}
   \begin{enumerate}[label=(\roman*)]
   \item 
     With $1/z = Lt$, the large $z$ expansion of $z\tilde{G}^{\rm J}(z)$ is thus seen to be a special case of
     the function
     \begin{equation}\label{Xt5}
       F(t) = \sum_{n=0}^\infty \binom{\alpha + \beta n}{n} t^n.
     \end{equation}
     This is intimately related to Lambert's solution of the trinomial equation $x = q + x^m$ in a power series in $q$
     \cite{Graham-Knuth-Patashnik94}. In mathematical physics, there are applications of (\ref{Xt5}), and its multivariable analogue, in the
     theory of anyons \cite{Aomoto-Iguchi99, Aomoto-Iguchi01}.
   \item 
     As $z \to 0$, we read off from (\ref{G8}) that $(z \tilde{G}^{\rm J}(z))^{\theta + 1} \sim - z (1/\theta)^\theta$. Since the corresponding global density is given in terms of the resolvent by $d\tilde{\mu}^{{\rm J}}(x) =
     {1 \over \pi} {\rm Im} \, \tilde{G}^{\rm J}(x - i 0) dx$, it follows that
     \begin{equation}\label{Ap}
       d\tilde{\mu}^{\mathrm{J}}(x) \mathop{\sim}_{x \to 0^+} \frac{\sin \left( \frac{\pi}{\theta + 1} \right)}{\pi \theta^{\frac{\theta}{\theta + 1}}} x^{-\frac{\theta}{\theta + 1}} dx. 
     \end{equation}
     Up to the scale factor $(1/\theta)^{\theta/(\theta + 1)}$, this is identical to the known $x \to 0^+$ behaviour of the global density in the Laguerre case \cite[Cor.~2.4 with $p=\theta + 1$, $r=1$]{Forrester-Liu14}. 
   \item 
     For $\theta = 2$ we have the explicit functional form \cite{Mlotkowski-Penson14}
     $$
     \left. d\tilde{\mu}^{{\rm J}}(x) \right\rvert_{\theta = 2} = \left( {3 (1 + \sqrt{1 - x})^{1/3} \over 4 \pi \sqrt{3(1-x)} }x^{-2/3} +
     {3 (1 + \sqrt{1 - x})^{-1/3} \over 4 \pi \sqrt{3(1-x)}} x^{-1/3} \right) dx.
     $$
     The corresponding leading term behaviour for $x \to 0^+$ agrees with that implied by (\ref{Ap}). 
   \end{enumerate}
 \end{remark}
 
 In addition to the above remarks, we draw attention to a relationship between the Jacobi Muttalib--Borodin ensemble and products of truncations of
 Haar distributed unitary matrices. Let $U_j$ be a Haar distributed unitary random matrix of size $m_j \times m_j$ and let $T_j$ be the corresponding
 $(n + \nu_j) \times (n + \nu_{j-1})$, with $\nu_j  \ge 0$ and $\nu_0 = 0$, and $m_j \ge 2n$ upper
 left block such that $\ell \ge 2n$. Let $G_j$ denote a standard complex Gaussian matrix of size $(n + \nu_j) \times (n + \nu_{j-1})$.
 According to the recent work \cite{Gawronski-Neuschel-Stivigny14}, in the limit $n \to \infty$ the singular values squared of the random matrix product
 \begin{equation}\label{GTs}
 G_{r} \cdots G_{s+1} T_s \cdots T_1
 \end{equation}
  has a density such that its moments $J_{r,s,1}(n)$ are given by
 \begin{equation}\label{Ja}
 J_{r,s,a}(n) = {a \over n} \Big ( {a^r \over (1 + a)^s } \Big )^n P_{n-1}^{(\alpha_{n-1}, \beta_{n-1})} \Big ( {1 - a \over 1 + a} \Big ),
 \end{equation}
 where
 $$
 \alpha_n = r n + r + 1, \qquad \beta_n = -(r+1-s)n - (r+2-s)
 $$
 and $P_k^{(\alpha,\beta)}(x)$ denotes the Jacobi polynomial. It is also required that $r > s$ and thus at least one Gaussian matrix in
 the product. The immediate relevance of this work is due to the fact that the moments (\ref{Ja}) are shown in \cite{Gawronski-Neuschel-Stivigny14} to be such
 that $w = azG(z)$, where $G(z)$ denotes the corresponding resolvent, satisfy
 \begin{equation}\label{GSa}
 w^{r+1} - z (w-a) (w+1)^s = 0.
 \end{equation}
 This is the same as our equation (\ref{G8}) with $w/a = z \tilde{G}^{\rm J}$, $r=s=a=\theta$. Indeed for these parameters we can check
 that (\ref{Ja}) reduces to (\ref{mL}), up to the form of the scale factor $L$.
 
 Subsequent to  \cite{Gawronski-Neuschel-Stivigny14}, the work \cite{Kieburg-Kuijlaars-Stivigny15} has considered integrability and exactly solvable features of the random matrix product 
 (\ref{GTs}) with $r=s$. In particular, it has been shown that the corresponding characteristic polynomial for the squared singular values is given by \cite[Eq.~(2.31) with $r \mapsto s$]{Kieburg-Kuijlaars-Stivigny15}
  \begin{equation}\label{GSb}
  P_n(x) = G_{s+1,s+1}^{0,s+1} \Big ( {n+1,n-m_1,\dots, n - m_s \atop
  0, -\nu_1,\dots,-\nu_s} \Big | x \Big ),
  \end{equation}
  where $G_{s+1,s+1}^{0,s+1}$ is a particular Meijer $G$-function and instead of the constraint $m_j \ge 2n$ as required in \cite{Gawronski-Neuschel-Stivigny14}, it is required $m_1 \ge 2n _ 1 + \mu_1$ and
  $m_j \ge  n + \nu_j + 1$ ($j=2,\dots,n$). According to the strategy introduced in \cite{Forrester-Liu14}, and applied in the derivation of
  Propositions \ref{PC1} and \ref{P3.8} above, the significance of this result is that the Meijer $G$-function satisfies a linear differential equation
  (see e.g.~\cite[Sec.~5.8]{Luke69}
  ) allowing us to deduce a polynomial equation for the corresponding resolvent $G(z)$. Specifically, this procedure gives
  \begin{equation}\label{3.47}
  z( zG(z) - 1) \prod_{j=1}^s \Big ( zG(z) -1 + \alpha_j) = z G(z) \prod_{j=1}^s ( zG(z) + \beta_j),
  \end{equation}
where $\alpha_j = \lim_{n \to \infty} m_j/n$, $\beta_j = \lim_{n \to \infty} \nu_j/n$. We see that (\ref{3.47})  reduces to (\ref{G8}) in the case $s = \theta$, $\beta_j=0$ $(j=1,\dots,s)$,
$\alpha_j = 1 + 1/\theta$ $(j=1,\dots,s)$. Thus we learn that the global spectral density for the Jacobi Muttalib--Borodin ensemble in the case $\theta = s \in \mathbb Z^+$ is the same as the global spectral density for the
squared singular values of the random matrix product $T_s \cdots T_1$ where each $T_j$ is an $n \times n$ sub-block of a Haar distributed unitary matrix of size $(n + p) \times (n+p)$ where
$\lim_{n \to \infty} p/n = 1/s$.

Choosing instead $\beta_j = 0$ $(j=1,\dots,s)$ and $\alpha_j = 1 + 1/a$ we see that (\ref{3.47}) is the same equation as
(\ref{GSa}) with $w/a = zG$ in the latter. Thus $J_{s,s,a}(n)$ has an interpretation as the moments of the spectral density of the squared singular values of the random matrix product
$T_s \cdots T_1$ where each $T_j$ is an $n \times n$ sub-block of a Haar distributed unitary matrix of size $(n + p) \times (n+p)$ where
$\lim_{n \to \infty} p/n = 1/a$ for general $a > 0$. The case $a \to 0$ exhibits further structure. Then the underlying unitary matrices are of infinite size, and after rescaling the
sub-blocks have a Gaussian distribution. More specifically, an $n \times n$ sub-block $U_n$ of an $(n+p) \times (n+p)$ Haar distributed unitary matrix has a distribution proportional to
$(\det (\mathbb I_n - U_n^\dagger U_n))^{p-n}$ \cite{Zyczkowski-Sommers00}. It follows that for $p \to \infty$, and with $n$ fixed, the distribution of $G_n = \sqrt{p} U_n$ is proportional to
$e^{-G_n^\dagger G_n}$. Hence $G_n$ is a standard complex Gaussian matrix. For $a \to 0$ we see from (\ref{Ja}) that
\begin{equation}
a^{-1 - sn} J_{s,s,a}(n) \to {1 \over n} P_{n-1}^{(sn+1,-n)}(1) =
{1 \over sn+1} \Big ( {sn + n \atop n} \Big ).
\end{equation}
These are Fuss--Catalan numbers (\ref{FCnumber}), which we know are moments of global spectral limit for the squared singular values of a product of $s$ standard complex Gaussian matrices.

\section{The global density --- saddle point method} \label{sec:saddle_pt}

It has already been remarked that the Borodin-Muttalib ensemble \eqref{1.1} is intimately related to 
biorthogonal polynomials of one variable as specified by (\ref{PL}). Moreover, as noted in \cite{Muttalib95}
(\ref{1.1}) is an example of a determinantal point process. This means that the general $k$-point correlation
function $\rho_{(k)}(x_1,\dots,x_k)$ is fully determined by a function $K(x,y) = K_N(x, y)$, independent of $k$ and referred
to as the correlation kernel, according to the formula
\begin{equation}\label{KK}
\rho_{(k)}(x_1,\dots,x_k) = \det [ K(x_i, x_j) ]_{i,j=1,\dots,k}.
\end{equation}
Moreover, the correlation kernel is expressed in terms of the biorthogonal polynomials $p_l, q_l$ defined and corresponding normalisation $h_l$ in \eqref{PL} according to
\begin{equation}\label{KK1}
K(x,y) = e^{-(V(x) + V(y))/2} \sum_{l=0}^{N-1}  {1 \over h_l} p_l(x) q_l(y^\theta).
\end{equation} 
Since we focus on the classical Laguerre and Jacobi cases, we denote the correlation kernel by $K^{\mathrm{L}}(x, y)$ for $V$ defined in \eqref{L} and $K^{\mathrm{J}}(x, y)$ for $V$ defined in \eqref{J}.

In \cite{Borodin99}, an indirect way to transform the summation in (\ref{KK1}) was devised for the classical Laguerre and Jacobi weights. This transformed summation enabled the computation of the so called hard edge scaled limit. This refers to the limit $N \to \infty$ with $x,y$ scaled so that in the neighbourhood of the origin the spacing between eigenvalues is of order unity. Our present interest is in the global limit of the one-point function. The transformed summation of \cite{Borodin99} is not suited for that purpose. Fortunately Propositions \ref{prop:corr_kernel_L} and \ref{prop:corr_kernel_J} provide the double contour integral formulas for both the Laguerre and Jacobi cases of the Borodin-Muttalib ensemble, which is suited. In Section \ref{sec:realisations_and_extensions} we showed that they are spectrally equivalent to the upper-triangular ensemble (defined in \eqref{yk}) and its Jacobi counterpart (defined in \eqref{eq:diagonal_Z}) respectively, but with restricted values of $\alpha_i$ and $\beta_i$. 

\subsection{The Laguerre Muttalib--Borodin ensemble} \label{subsec_multiple_Laguerre_ensemble}

In the Laguerre case, comparing the correlation kernel formula \eqref{KK1} with the correlation formula \eqref{eq:corr_kernel_L} with $n_1 = n_2 = N$, we derive the kernel formula \eqref{KK2} in Proposition \ref{prop:kernel_formula_L}, modulo the factor $h(x)/h(y)$. A multiplicative factor in the form of $h(x)/h(y)$ in the correlation kernel does not change the joint probability density function expressed in \eqref{KK}, so the kernel formula in \eqref{KK2} is valid. This factor account for the different meaning of $K(n_1, x; n_2, y)$ adopted in Proposition \ref{prop:corr_kernel_L} in consistency with \cite[Eq.~(21)]{Adler-van_Moerbeke-Wang11}, which has the factor $e^{-(V(x) + V(y))/2}$ in \eqref{KK1} replaced by $e^{-x} y^c$, see also \cite[Eq.~(108)]{Adler-van_Moerbeke-Wang11} and derivations above it. Thus the function $h$ can be determined, up to a multiplicative constant, by the requirement
\begin{equation}\label{hh}
{ h(x) \over h(y) } = \Big ({x \over y} \Big )^{c/2} e^{(x-y)/2}.
\end{equation}

We recall from \eqref{Pd3A} that the upper-triangular ensemble is well defined for $\theta = 0$ and $c > -1$ and it is the $\theta \to 0^+$ limit of the Laguerre Muttalib--Borodin ensemble. We call it the $\theta = 0$ case of the ensemble. The double contour integral formula \eqref{KK2} is still valid in this case, where $\alpha_j = c$ for all $j = 1, \dotsc, N$. Note that this case is degenerate in the sense that the biorthogonal polynomials $p_l$ and $q_l$ are not well defined by \eqref{PL}. 

Below we consider the limiting global density in two cases, first for $\theta > 0$ and next for $\theta = 0$. We give details of the derivation in the former case, while point out the differences of the argument in the latter.

\subsubsection{The $\theta > 0$ case} \label{subsubsec:theta>0_L}

We seek to use (\ref{KK2}) to compute the density function of the the transformed limiting counting measure $d\tilde{\mu}$ in \eqref{3.4}. By definition
the density function is given by \eqref{eq:limit_counting_measure} and the required change of variables is $\tilde{x} = (x/\theta)^\theta$ (recall the text below (\ref{3.4})),
so (\ref{KK1}) gives
\begin{equation}\label{hh1}
  \tilde{\rho}^{{\rm L}}_{(1)}(x) =  \lim_{N \to \infty} \tilde{K}^{\mathrm{L}}(x) := \lim_{N \to \infty}  x^{1/\theta - 1} K^{\rm L}(N \theta x^{1/\theta},
N \theta x^{1/\theta}),
\end{equation}
and is the global density appearing in (\ref{eq:defn_tilde_G^L}). The scale factor $\theta$ is chosen with the benefit of hindsight as it allows (\ref{3.4b}) to be reclaimed without the need for rescaling.

Under the assumption $\theta \in \intZ^+$, we showed in Section \ref{S3.1} that the resolvent $\tilde{G}^{\mathrm{L}}(z)$ defined in \eqref{eq:defn_tilde_G^L} corresponding to $d\tilde{\mu}^{\mathrm{L}}$, satisfies the identity \eqref{3.4b} characterizing the Fuss--Catalan density. Consequently
\begin{equation} 
  \tilde{\rho}^{{\rm L}}_{(1)}(x) = \lim_{\epsilon \to 0^+} {1 \over 2 \pi i}
  \left( \tilde{G}^{\rm L}(x - i \epsilon) -  \tilde{G}^{\rm L}(x + i \epsilon) \right),
\end{equation}
is the Fuss--Catalan distribution supported on $I = (0, (1 + \theta)^{1 + \theta} \theta^{-\theta})$. The parametrization by Biane and independently Neuschel \cite{Bercovici-Pata99, Neuschel14} gives that for $x \in I$, there is a unique $\varphi \in (0, \pi/(\theta + 1))$ such that
\begin{equation}\label{4.15}
  x = \frac{(\sin((\theta + 1)\varphi))^{\theta + 1}}{\sin\varphi(\sin(\theta\varphi))^\theta}, \quad \text{where} \quad 0 < \varphi < \frac{\pi}{\theta + 1},
\end{equation}
and which allows for the simple functional form
\begin{equation} \label{eq:Fuss-Catalan_density}
  \rho^{\FC}(x) = \frac{1}{\pi x} \frac{\sin((\theta + 1)\varphi)\sin\varphi}{\sin(\theta\varphi)} = \frac{\sin(\theta\varphi)^{\theta - 1}(\sin\varphi)^2}{\pi \sin((\theta + 1)\varphi)^{\theta}}.
\end{equation}
With the meaning of $\rho^{\FC}(x)$ in (\ref{1.8}) so established, we now turn to our main task of this subsection.
 
 \medskip
 \begin{proof}[Proof of Proposition \ref{prop:global_density}]
Our analysis,
which is based on the method of steepest descents, is guided by a very similar calculation carried out recently by
Liu, Wang and Zhang \cite{Liu-Wang-Zhang14} in relation to the global density for the squared singular values of a product $M$
standard complex Gaussian matrices. That these two computations should be closely related is not surprising upon
recalling from Section \ref{S3.1} that the global density in the case $\theta = M \in \mathbb Z^+$ coincides with the 
global density for the squared singular values of a product $M$
standard complex Gaussian matrices. It turns out that if $\theta \geq 1$, our proof follows that in \cite{Liu-Wang-Zhang14} closely, and if $\theta \in (0, 1)$, we need to construct the contour $\Sigma$ in an alternative way, which we explain in the proof. As well as guiding our overall strategy, \cite{Liu-Wang-Zhang14} will be referred to for the
proof of some technical bounds, which we omit below.

Substituting $N\theta x^{1/\theta}$ for $x$ and $N\theta x^{1/\theta}$ for $y$ in \eqref{KK2} gives
\begin{equation} \label{eq:double_contour_density}
  \tilde{K}^{\rm L}(x) := x^{\frac{1}{\theta} - 1} K^{\rm L}(N\theta x^{\frac{1}{\theta}}, N\theta x^{\frac{1}{\theta}}) = \frac{1}{N\theta x} \frac{1}{(2\pi i)^2} \oint_{\Sigma} dz  \oint_{\Gamma_\alpha} dw \frac{e^{F(z; x)}}{e^{F(w; x)}} \frac{1}{z - w},
\end{equation}
where
\begin{equation}
  F(z; x) = \log \left( \frac{\Gamma(z + 1)}{(N\theta x^{\frac{1}{\theta}})^{z}} \prod^N_{k = 1} (z - c - (k - 1)\theta) \right).
\end{equation}
The logarithm function takes the principal branch and we assume that the value of $\log z$ for $z \in (-\infty, 0)$ is continued from above, to remove ambiguity. To derive the asymptotics of $F(z; x)$, we denote
\begin{equation}\label{zuv}
  z = N\theta u, \quad w = N\theta v,
\end{equation}
and let $\epsilon$ be a positive constant. We begin by noting that uniformly in $z$ such that for all $z$ satisfying
\begin{equation} \label{eq:z_away_from_interval}
  \dist(z, [0, N\theta]) > \epsilon N, \quad \text{or equivalently} \quad \dist(u, [0, 1]) > \epsilon/\theta,
\end{equation}
we have
\begin{equation}\label{4.10}
  \begin{split}
    & \log \left( \prod^N_{k = 1} (z - c - (k - 1)\theta) \right) \\
    = {}& N (\log N + \log \theta) + N \log u + N \int^1_0 \log \left( 1 - \frac{c}{N\theta u} - \frac{x}{u} \right) dx \\
    & + \sum^N_{k = 1} \left[ \log \left( 1 - \frac{\frac{c}{\theta} + (k - 1)}{Nu} \right) - N \int^{\frac{k}{N}}_{\frac{k - 1}{N}} \log \left( 1 - \frac{c}{N\theta u} - \frac{x}{u} \right) dx \right] \\
    = {}& N \left( (1 - u) \log \left( 1 - \frac{1}{u} \right) + \log u \right) \\
    & + \left( \frac{c}{\theta} - \frac{1}{2} \right) \log \left( 1 - \frac{1}{u} \right) + N (\log N + \log \theta - 1) + \bigO(N^{-1}).
  \end{split}
\end{equation}
Next using Stirling's formula \cite[6.1.37]{Abramowitz-Stegun64}, we have that for $z$ satisfying \eqref{eq:z_away_from_interval} and $\arg z = \arg u \in (-\pi + \epsilon, \pi - \epsilon)$,
\begin{equation}
  \log \Gamma(z + 1) = N\theta u(\log N + \log \theta + \log u - 1) + \frac{1}{2} \log u + \log \sqrt{2\pi N\theta} + \bigO(N^{-1}). \nonumber
\end{equation}
This allows us to write
\begin{equation} \label{eq:F_in_F_hat}
  F(z; x) = N \hat{F}_1(u; x) + \hat{F}_0(u)  + N \log N + \log \sqrt{2\pi N\theta} + \bigO(N^{-1}),
\end{equation}
where 
\begin{align}
  \hat{F}_1(u; x) & = (\theta + 1)u(\log u - 1) - (u - 1)(\log(u - 1) - 1) - u\log x,  \label{4.13} \\
  \hat{F}_0(u) & = \Big({c \over \theta}  - {1 \over 2} \Big )  \log \left( 1 - \frac{1}{u} \right) + \frac{1}{2} \log u,
\end{align}
and the error bound  $ \bigO(N^{-1})$ holds uniformly for all $z$ satisfying \eqref{eq:z_away_from_interval} and $\arg z = \arg u \in (-\pi + \epsilon, \pi - \epsilon)$.
Substituting in (\ref{eq:double_contour_density}) shows that if we can deform $\Sigma$ and $\Gamma_{\alpha}$ such that $\arg z, \arg w \in (-\pi + \epsilon, \pi - \epsilon)$ and $\lvert z \rvert, \lvert w \rvert > \epsilon N$, then
\begin{equation} \label{eq:double_contour_density1}
  \tilde{K}^{\rm L}(x) = \frac{1}{N\theta x} \frac{1}{(2\pi i)^2} \oint_{\Sigma} dz  \oint_{\Gamma_\alpha} dw \frac{e^{N\hat{F}_1(u; x)
  + \hat{F}_0(u; x)}}{e^{N\hat{F}_1(v; x)
  + \hat{F}_0(v; x)}} \frac{1}{z - w} \Big (1 + \bigO(N^{-1}) \Big ).
\end{equation}
In the case
$\theta = M \in \mathbb Z^+$ this contour integral is comparable to \cite[Eq.~(2.7)]{Liu-Wang-Zhang14}, and our $\hat{F}_1(u; x)$ is defined the same as the $\hat{F}(z; a)$ occurring in \cite[Eq.~(2.6)]{Liu-Wang-Zhang14} with $u = z$ and $x = a$.

It follows from (\ref{4.13}) that
$$
{d \over du} \hat{F}_1(u;x) = \log {u^{\theta + 1} \over (u-1) x},
$$
and thus stationary points occur for $u$ such that
\begin{equation}\label{ux}
u^{\theta + 1} = (u - 1) x.
\end{equation}
Note that this is precisely the equation (\ref{3.4b}) after the identification $z \mapsto x$, $z \tilde{G}^{\rm L}(z) \mapsto u$ in the latter. It is known that (\ref{ux}) permits a pair of complex conjugate solutions for $x \in (0, \theta(1 + 1/\theta)^{\theta + 1})$, $u_+^{\rm L}, u_-^{\rm L}$ say, which merge to real solutions for $x=0$ and $x= \theta (1 + 1/\theta)^{\theta + 1}$.
With $x$ parametrized by $\varphi$ as in \eqref{4.15}, we can check that
\begin{equation} \label{eq:U^L_pm}
  u_+^{\rm L} = \frac{\sin((\theta + 1)\varphi)}{\sin(\theta\varphi)}e^{i\varphi}, \quad u_-^{\rm L} = \frac{\sin((\theta + 1)\varphi)}{\sin(\theta\varphi)}e^{-i\varphi}.
\end{equation}
Note that our $u^{\mathrm{L}}_{\pm}$ are the same as the $w_{\pm}$ defined in \cite[Eq.~(2.10)]{Liu-Wang-Zhang14}. In terms of this parametrisation we have (comparing with \cite[Eq.~(2.11)]{Liu-Wang-Zhang14})
\begin{equation} \label{eq:defn_C_pm}
  C_{\pm} := \frac{d^2}{du^2} \hat{F}_1(u_{\pm}^{\rm L}; x) = \frac{1}{u_{\pm}^{\rm L}} \left( \theta + 1 - \frac{\sin((\theta + 1)\varphi)}{\sin\varphi} e^{\mp i\theta \varphi} \right),
\end{equation}
and in particular $C_{\pm} \neq 0$ for $\varphi$ in the range given in (\ref{4.15}).

Next we deform the contours $\Sigma$ and $\Gamma_{\alpha}$ for the steepest-descent analysis. We consider the cases $\theta \geq 1$ and $\theta \in (0, 1)$ separately. 

\paragraph{Construction of the contours: $\theta \geq 1$ case}

We assume that $\Re N\theta u_{\pm}^{\rm L}$ is different from all $\alpha_k$ ($k = 1, \dotsc, N$). If this assumption is not satisfied, see Remark \ref{rmk:identical_to_a_k} below.

It is clear that the Hankel like contour $\Sigma$ can be deformed to an infinite contour that is from $-i \cdot \infty$ to $i \cdot \infty$, as long as it keeps the poles $-1, -2, \dotsc$ of $z$ to its left and $\Gamma_{\alpha}$ to its right. Actually this deformation has more freedom. By the residue theorem, if $\Sigma$ is deformed into the infinite vertical contour to the right of $\Gamma_{\alpha}$, the double contour integral in \eqref{eq:double_contour_density} remains the same. Moreover, we can split $\Gamma_{\alpha}$ into two positively oriented contours that jointly enclose all the poles $\alpha_1, \dotsc, \alpha_N$, and let $\Sigma$ be an infinite vertical contour passing between them. See \cite[Eq.~(2.8)]{Liu-Wang-Zhang14} for the explicit computation in a similar case. Specifically, deform $\Sigma$ from the Hankel contour into the upward vertical contour 
\begin{equation} \label{eq:Simaga_-1_vertical}
  \Sigma = \{ \Re N\theta u_{\pm}^{\rm L} + it \mid t \in \realR \}.
\end{equation}
To express the shape of the deformed contour $\Gamma_{\alpha}$, we define first the contours $\tilde{\Gamma}$ and $\tilde{\Gamma}^{\epsilon}$ where $\epsilon > 0$. Thus define
\begin{equation} \label{eq:Gamma_constr_Lag}
  \tilde{\Gamma} = \left\{ \frac{\sin((\theta + 1)\phi)}{\sin(\theta\phi)}e^{i\phi} \middle| \phi \in \left[ -\frac{\pi}{\theta + 1}, \frac{\pi}{\theta + 1} \right] \right\},
\end{equation}
and for a small enough $\epsilon > 0$, (see Figure \ref{fig:tilde_Gamma})
\begin{multline} \label{eq:defn_tilde_Gamma^epsilon}
  \tilde{\Gamma}^{\epsilon} = \{ z \in \tilde{\Gamma} \mid \lvert z \rvert \geq \epsilon \} \\
  \cup \text{the arc of $\{ \lvert z  \rvert = \epsilon \}$ connecting $\tilde{\Gamma} \cap \{ \lvert z \rvert = \epsilon \}$ and through $-\epsilon$}.
\end{multline}
Then we define $\Gamma_{\alpha}$, depending on two small positive constants $\epsilon$ and $\epsilon'$. Here we take the notational convention that if $C$ is a contour and $r > 0$, then $rC$ is the contour consisting of $\{ z \mid z/r \in C \}$ and with the same orientation. In terms of this notation
\begin{equation} \label{eq:defn_contour_Gamma}
  \Gamma_{\alpha} = \Gamma_{\curved} \cup \Gamma_{\vertical}, \quad \text{where} \quad \Gamma_{\curved} = \Gamma_1 \cup \Gamma_2, \quad \Gamma_{\vertical} = \Gamma_3 \cup \Gamma_4,
\end{equation}
and
\begin{align}
  \Gamma_1 = {}& N \theta \tilde{\Gamma}^r \cap \{ z \mid \Re z \leq \Re N \theta u^{\mathrm{L}}_{\pm} - \epsilon \}, \label{eqGamma_alpha_split:1} \\
  \Gamma_2 = {}& N \theta \tilde{\Gamma}^r \cap \{ z \mid \Re z \geq \Re N \theta u^{\mathrm{L}}_{\pm} + \epsilon \}, \qquad \text{with $r = \frac{[\epsilon' N] + \frac{1}{2}}{N}$}, \label{eqGamma_alpha_split:2} \\
  \Gamma_3 = {}& \text{vertical bar connecting the two ending points of $\Gamma_1$}, & \label{eqGamma_alpha_split:3} \\
  \Gamma_4 = {}& \text{vertical bar connecting the two ending points of $\Gamma_2$}. \label{eqGamma_alpha_split:4}
\end{align}
\begin{figure}[htb]
  \begin{minipage}[t]{0.45\linewidth}
    \centering
    \includegraphics{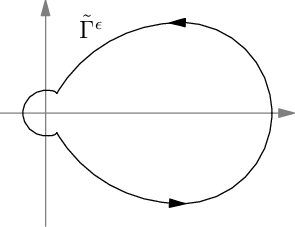}
    \caption{The schematic shape of $\tilde{\Gamma}^{\epsilon}$ for the Laguerre Muttalib--Borodin ensemble with $\theta \geq 1$.}
    \label{fig:tilde_Gamma}
  \end{minipage}
  \hspace{\stretch{1}}
  \begin{minipage}[t]{0.45\linewidth}
    \centering
    \includegraphics{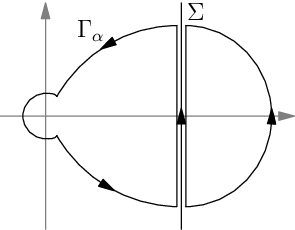}
    \caption{The schematic shapes of $\Sigma$ and $\Gamma_{\alpha}$ for the Laguerre Muttalib Borodin ensemble with $\theta \geq 1$.}
    \label{fig:contours}
  \end{minipage}
\end{figure}

\begin{figure}[htb]
  \begin{minipage}[t]{0.45\linewidth}
    \centering
    \includegraphics{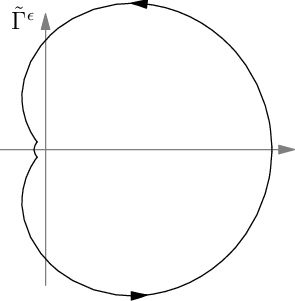}
    \caption{The schematic shape of $\tilde{\Gamma}^{\epsilon}$ for the Laguerre Muttalib--Borodin ensemble with $\theta \in (0, 1)$.}
    \label{fig:small_theta}
  \end{minipage}
  \hspace{\stretch{1}}
  \begin{minipage}[t]{0.45\linewidth}
    \centering
    \includegraphics{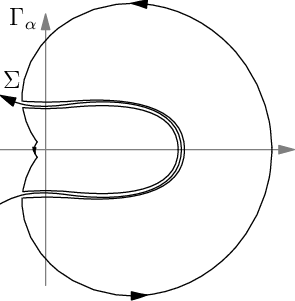}
    \caption{The schematic shapes of $\Sigma$ and $\Gamma_{\alpha}$ for the Laguerre Muttalib Borodin ensemble with $\theta \in (0, 1)$.}
    \label{fig:small_theta_all}
  \end{minipage}
\end{figure}
Note that $\Gamma_1 \cup \Gamma_3$ and $\Gamma_2 \cup \Gamma_4$ are disjoint closed contours, and we assume that they are both oriented counterclockwise. Here we choose $\epsilon$ small enough so that $\Re N \theta u_{\pm}^{\rm L} \pm \epsilon$ lie between two poles of the integral $\alpha_j$ and $\alpha_{j + 1}$, and then $\Gamma_1 \cup \Gamma_3$ and $\Gamma_2 \cup \Gamma_4$ jointly cover all the poles $\alpha_1, \dotsc, \alpha_N$. Our contours $\Sigma$ and $\Gamma_{\alpha}$ are identical to the contours $\mathcal{C}$ and $\Sigma$ respectively defined in \cite[Sec.~2.2]{Liu-Wang-Zhang14} up to the factor $\theta$ and our $\tilde{\Gamma}$ and $\tilde{\Gamma}^{\epsilon}$ are identical to $\tilde{\Sigma}$ and $\tilde{\Sigma}^{\epsilon}$ in \cite[Sec.~3.1]{Liu-Wang-Zhang14} respectively, if $\theta = M \in \intZ^+$. See Figure \ref{fig:contours} for the shapes of $\Sigma$ and $\Gamma_{\alpha}$.

\paragraph{Construction of contours: $\theta \in (0, 1)$ case}

The construction of the contours in the $\theta \geq 1$ case obviously is not valid if $\theta < 1$ and $x$ is close to $0$, since if the vertical line through $u^{\mathrm{L}}_{\pm}$ will intersect $\tilde{\Gamma}^{\epsilon}$ defined in \eqref{eq:defn_tilde_Gamma^epsilon} at other points, see Figure \ref{fig:small_theta}. So $\Sigma$ needs to be deformed into a more complicated shape, and the construction becomes less straightforward. We use method of elementary dynamical systems to construct the contour $\Sigma$.

We construct $\Sigma = N \theta \tilde{\Sigma}$, where $\tilde{\Sigma}$ is a contour passing through $u^{\mathrm{L}}_{\pm}$. Consider the gradient field generated by the function $\Re \hat{F}_1(z; x)$, denoted by $\nabla \Re \hat{F}_1(z; x)$. Through a regular point, there is a unique flow line associated to $\nabla \Re \hat{F}_1(z; x)$, and the value of $\Re \hat{F}_1(z; x)$ increases as $z$ moves along the flow line. Since $\Re \hat{F}_1(z; x)$ is a harmonic function, a flow line will not stop at a local maximum or start from a local minimum. At a saddle point, we can concatenate a flow line into it and a flow line out of it, and $\Re \hat{F}_1(z; x)$ increases along the concatenated flow line. Hence we have that any flow line can be prolonged so that it starts from either infinity or a singular point of $\nabla \Re \hat{F}_1(z; x)$, (that is,  $0$ or $1$,) although the prolonged flow line may not be unique if it goes through a saddle point. 

By the definition \eqref{eq:U^L_pm} of $u^{\mathrm{L}}_+$, we have that $u^{\mathrm{L}}_+$ is a critical point of $\nabla \Re \hat{F}_1(z; x)$ of second order, so from this point there are two flow lines going into this point. Note that the curve $\tilde{\Gamma} \cap \compC_+$ divides $\compC_+$ into two disconnected parts. Since $\Re \hat{F}_1(z; x)$ attains its minimum over $\tilde{\Gamma} \cap \compC_+$ at $u^{\mathrm{L}}_+$, the two flow lines into $u^{\mathrm{L}}_+$ cannot intersect $\tilde{\Gamma} \cap \compC_+$. We denote $\gamma_1$ the flow line coming to $u^{\mathrm{L}}_+$ from below $\tilde{\Gamma} \cap \compC_+$, and $\gamma_2$ the flow line coming to $u^{\mathrm{L}}_+$ from above it. 

Since the gradient field $\nabla \Re \hat{F}_1(z; x)$ does not have a singularity within the region enclosed by $\tilde{\Gamma} \cap \compC_+$ and the inteval $[0, 1 + \theta^{-1}]$, flow line $\gamma_1$ does not start within the region, and has to enter the region before reaching $u^{\mathrm{L}}_+$. Since the value of $\Re \hat{F}_1(z; x)$ on $\tilde{\Gamma} \cap \compC_+$ attains its minimum at $u^{\mathrm{L}}_+$, $\gamma_1$ has to enter the region from the $[0, 1 + \theta^{-1}]$ side. Hence we let $\sigma_1$ be the part of $\gamma_1$ between its (last) entrance to the region and $u^{\mathrm{L}}_+$. It is clear that $\sigma_1$ lies in the region and connects $u^{\mathrm{L}}_+$ and a point on $(0, 1 + \theta^{-1})$. (We can further show that this point is $1$, but it is not relevant to our construction.)

On the other hand, we want to show that the flow line $\gamma_2$ does not intersect the real axis. If this holds, it has to be from infinity to $u^{\mathrm{L}}_+$, and from the limiting behaviour of $\hat{F}_1(z; x)$ we know that it comes from the direction $e^{i\pi} \cdot \infty$. We construct $\tilde{\Sigma} \cap \compC_+$ as $\gamma_2 \cup \sigma_1$, and then have $\tilde{\Sigma}$ by taking the reflection about the real axis. In the end, we take $\Sigma = N \theta \tilde{\Sigma}$ and orient $\Sigma$ as from $N \theta u^{\mathrm{L}}_-$ to $N \theta u^{\mathrm{L}}_+$, and finish the construction of $\Sigma$. 

To show that $\gamma_2$ does not intersect the real axis, or more specifically, $\realR \setminus [0, 1 + \theta^{-1}]$, we note that the ray $(1 + \theta^{-1}, +\infty)$ is a flow line, so $\gamma_2$ does not intersect with it since $\gamma_2$ does not overlap with it. Let $x^* \in (-\infty, 0)$ be the unique critical point of $\nabla \Re \hat{F}_1(z; x)$ on $(-\infty, 0)$, then $(-\infty, x^*)$ and $(x^*, 0)$ are flow lines, so $\gamma_2$ can intersect with $(-\infty, 0)$ only at $x^*$. However, we can check that $\Re \hat{F}_1(x^*; x) > \Re \hat{F}_1(0; x) > \Re \hat{F}_1(u^{\mathrm{L}}_+; x)$, so $\gamma_2$ cannot pass through $x^*$ before reaching $u^{\mathrm{L}}_+$.

Now we briefly describe the construction of $\Gamma_{\alpha}$ in the $\theta \in (0, 1)$ case. We still define $\tilde{\Gamma}^{\epsilon}$ by \eqref{eq:defn_tilde_Gamma^epsilon}, and then define $\Gamma_1$ and $\Gamma_2$ by \eqref{eqGamma_alpha_split:1} and \eqref{eqGamma_alpha_split:2}, and also define $\Gamma_{\curved}$ as the union of $\Gamma_1$ and $\Gamma_2$, as in \eqref{eq:defn_contour_Gamma}. Next, analogous to $\Gamma_3$ and $\Gamma_4$ defined in \eqref{eqGamma_alpha_split:3} and \eqref{eqGamma_alpha_split:4}, we define, in our case, $\Gamma_3$ as a contour connecting the two ending points of $\Gamma_1$ such that all points of $\Gamma_3$ are within distance $\epsilon$ to $\Sigma$, and $\Gamma_4$ as a contour connecting the two ending points of $\Gamma_2$ such that all points of $\Gamma_4$ are within distance $\epsilon$ to $\Sigma$. Finally we denote the union of $\Gamma_3$ and $\Gamma_4$ as $\Gamma_{\vertical}$, although $\Gamma_3$ and $\Gamma_4$ are no longer vertical, and let $\Gamma = \Gamma_{\curved} \cup \Gamma_{\vertical}$ as in \eqref{eq:defn_contour_Gamma}. Hence we finish the construction of the contours in the $\theta \in (0, 1)$ case, and see Figure \ref{fig:small_theta_all} for the shape of the contours.

\begin{remark}
  The construction of contours $\Sigma$ and $\Gamma_{\alpha}$ above also applies for the $\theta \geq 1$ case. We still keep the construction with vertical $\Sigma$, because it is conceptually simpler and more analogous to the construction in \cite{Liu-Wang-Zhang14}.
\end{remark}

Our choice of the contours $\Sigma$ and $\tilde{\Gamma}_{\alpha}$, in either the $\theta \geq 1$ case or the $\theta \in (0, 1)$ case, allows $N \theta u_\pm^{\rm L}$ to be identified as maximum points of $\Re F(z; x)$ for $z \in \Sigma$ and minimum points of $\Re F(z; x)$ for $z \in \Gamma_{\curved}$, which is key to the subsequent asymptotic analysis. To be precise, we state the result as follows, where we denote $D_r(z) = \{ w \in \compC \mid \lvert w - z \rvert < r \}$.
\begin{lemma} \label{lem:global_steepest_descent}
  There exists $\delta > 0$ such that for all $N$ large enough,
  \begin{align}
    \Re F(z; x) \geq {}& \Re F(N \theta u^{\mathrm{L}}_{\pm}) + \delta N \left\lvert \frac{z}{N \theta} - u^{\mathrm{L}}_{\pm} \right\rvert^2, && z \in \Gamma_{\curved} \cap D_{N^{\frac{3}{5}}} (N \theta u^{\mathrm{L}}_{\pm}), \label{eq:global_steepest_descent_Gamma_1} \\
    \Re F(z; x) > {}& \Re F(N \theta u^{\mathrm{L}}_{\pm}) + \delta N^{\frac{1}{5}}, && z \in \Gamma_{\curved} \setminus \left( D_{N^{\frac{3}{5}}} (N \theta u^{\mathrm{L}}_+) \cup D_{N^{\frac{3}{5}}} (N \theta u^{\mathrm{L}}_+) \right), \label{eq:global_steepest_descent_Gamma_2} \\
    \Re F(z; x) \leq {}& \Re F(N \theta u^{\mathrm{L}}_{\pm}) - \delta N \left\lvert \frac{z}{N \theta} - u^{\mathrm{L}}_{\pm} \right\rvert^2, && z \in \Sigma \cap D_{N^{\frac{3}{5}}} (N \theta u^{\mathrm{L}}_{\pm}), \label{eq:global_steepest_descent_Sigma_1} \\
    \Re F(z; x) < {}& \Re F(N \theta u^{\mathrm{L}}_{\pm}) - \delta N^{\frac{1}{5}}, && z \in \Sigma \setminus \left( D_{N^{\frac{3}{5}}} (N \theta u^{\mathrm{L}}_+) \cup D_{N^{\frac{3}{5}}} (N \theta u^{\mathrm{L}}_+) \right), \label{eq:global_steepest_descent_Sigma_2} \\
    \Re F(z; x) < {}& \Re F(N \theta u^{\mathrm{L}}_{\pm}) - \delta \lvert z \rvert && z \in \Sigma \cap \{ \lvert z \rvert > \delta^{-1} N \}. \label{eq:global_steepest_descent_Sigma_3} 
  \end{align}
\end{lemma}
The most important ingredient of the proof of Lemma \ref{lem:global_steepest_descent} is the estimate of the leading term $\hat{F}_1(z; x)$, which is identical to $\hat{F}(z; x)$ in \cite{Liu-Wang-Zhang14} if $\theta = M$. In the $\theta \geq 1$ case, the estimate of $\hat{F}_1(z; x)$ is stated in \cite[Lem.~3.1 and 3.2]{Liu-Wang-Zhang14}, where the results and the proofs hold for general $\theta > 1$. Then Lemma \ref{lem:global_steepest_descent} is proved analogously to the proof of \cite[Lem.~2.1]{Liu-Wang-Zhang14} in \cite[Sec.~3.2]{Liu-Wang-Zhang14}. We omit the detail. In the $\theta \in (0, 1)$ case, \eqref{eq:global_steepest_descent_Gamma_1} and \eqref{eq:global_steepest_descent_Gamma_2} are the same as the $\theta \geq 1$ case, and \eqref{eq:global_steepest_descent_Sigma_1}, \eqref{eq:global_steepest_descent_Sigma_2}, \eqref{eq:global_steepest_descent_Sigma_3} are due to the flow line definition of $\Sigma$. We also omit the detail.

\begin{remark} \label{rmk:identical_to_a_k}
  In the $\theta \geq 1$ case, If $\Re Nu_{\pm}^{\rm L}$ happens to be identical to a pole $\alpha_j$, we replace $\Re Nu_{\pm}^{\rm L}$ into $\Re Nu_{\pm}^{\rm L} + 1/2$ in formulas \eqref{eq:Simaga_-1_vertical}, \eqref{eqGamma_alpha_split:1} and \eqref{eqGamma_alpha_split:2}. All later arguments are valid with notational changes. In the $\theta \in (0, 1)$ case, if $\Sigma$ hits a pole $\alpha_j$, we deform $\Sigma$ and $\Gamma_{\vertical}$ slightly in an analogous way to avoid the pole.
\end{remark}

  Taking the limit $\epsilon \to 0$, we have
  \begin{equation} \label{eq:K=I_1+I_2}
    \tilde{K}^{\mathrm{L}}(x) = I_1 + I_2,
  \end{equation}
  where, with p.v.~denoted the principal value integral,
  \begin{equation} \label{eq:I_1_integral}
    \begin{split}
     I_1 = {}& \lim_{\epsilon \to 0} \frac{(N\theta x)^{-1}}{(2\pi i)^2} \int_{\Sigma} dz \int_{\Gamma_1 \cup \Gamma_2} dw \frac{e^{F(z; x)}}{e^{F(w; x)}} \frac{1}{z - w} \\
     = {}& \frac{(N\theta x)^{-1}}{(2\pi i)^2} \pv \int_{N \theta \tilde{\Gamma}^r} \left( \int_{\Sigma} dz \frac{e^{F(z; x)}}{e^{F(w; x)}} \frac{1}{z - w} \right) dw,
    \end{split}
  \end{equation}
  and
  \begin{equation} \label{eq:exact_value_of_I_2}
    \begin{split}
      I_2 = {}& \lim_{\epsilon \to 0} \frac{(N\theta x)^{-1}}{(2\pi i)^2} \int_{\Sigma} dz \int_{\Gamma_3 \cup \Gamma_4} dw \frac{e^{F(z; x)}}{e^{F(w; x)}} \frac{1}{z - w} \\
      = {}& \frac{(N\theta x)^{-1}}{2\pi i} \int^{N \theta u_+^{\rm L}}_{N \theta u_-^{\rm L}} \frac{e^{F(z; x)}}{e^{F(z; x)}} dz \\
      = {}& \frac{\Im u^{\mathrm{L}}_+}{\pi x} = \frac{(\sin\varphi)^2 (\sin(\theta\varphi))^{\theta - 1}}{\pi (\sin((\theta + 1)\varphi))^{\theta}}.
    \end{split}
  \end{equation}
  
  To evaluate $I_1$, we define
  \begin{equation}
    \Sigma^{\pm}_{\local} = \Sigma_{} \cap D_{N^{\frac{3}{5}}}(N\theta u_{\pm}^{\rm L}), \quad \Gamma^{\pm}_{\local} = \tilde{\Gamma}_{\alpha} \cap D_{N^{\frac{3}{5}}}(N\theta u_{\pm}^{\rm L}).
  \end{equation}
  Our strategy is to consider the Cauchy principal integral \eqref{eq:I_1_integral} on $\Sigma^+_{\local} \times \Gamma^+_{\local}$ and $\Sigma^-_{\local} \times \Gamma^-_{\local}$ first, and then to show that the remaining part of the integral is negligible in the asymptotic analysis.
  
  Make the change of variables
  \begin{equation}
    z = N\theta u_+^{\rm L} + N^{\frac{1}{2}}\theta s, \quad w = N\theta u_+^{\rm L} + N^{\frac{1}{2}}\theta t.
  \end{equation}
  Then by \eqref{eq:F_in_F_hat}
  \begin{equation} \label{eq:F_near_saddle_pt}
    \begin{split}
      F(z; x) = {}& N\hat{F}_1(u_+^{\rm L} + N^{-\frac{1}{2}}s; x)  + \hat{F}_0(u_+^{\rm L} + N^{-\frac{1}{2}}s) + N \log N + \log \sqrt{2\pi N\theta} + \bigO(N^{-1}) \\
      = {}& \frac{C_+}{2} s^2 + N\hat{F}_1(u_+^{\rm L}; x) + N \log N + \log \sqrt{2\pi N\theta} + \hat{F}_0(u_+^{\rm L}) + \bigO(N^{-\frac{1}{5}}),
    \end{split}
  \end{equation}
where $C_+$ is defined in \eqref{eq:defn_C_pm}. This gives
\begin{multline} \label{eq:crossing_estimate}
  \pv \int_{\Gamma^+_{\local}} dw \int_{\Sigma^+_{\local}} dz \frac{e^{F(z; x)}}{e^{F(w; x)}} \frac{1}{z - w}\\
  = \frac{1}{\sqrt{N}\theta} \pv \int_{\Gamma^+_{\local}} dw \int_{\Sigma^+_{\local}} dz \frac{e^{\frac{C_+}{2} s^2}}{e^{\frac{C_+}{2} t^2}} \frac{1 + \bigO(N^{-\frac{1}{5}})}{s - t}.
\end{multline}
Note that the $\bigO(N^{-1/5})$ term in the integrand on the right-hand side of \eqref{eq:crossing_estimate} is uniform and analytic in $N_{N^{3/5}}(N\theta u_+^{\rm L})$. Using \eqref{eq:global_steepest_descent_Gamma_1} of \eqref{eq:global_steepest_descent_Gamma_2} in Lemma \ref{lem:global_steepest_descent}, we find that there exists $c_1, c_2 > 0$ such that for $z \in \Sigma^+_{\local}$ and $w \in \Gamma^+_{\local}$,
\begin{equation} \label{eq:estimate_near_saddle_pt}
  \lvert e^{\frac{C_+}{2} s^2} \rvert \leq e^{-c_1 s^2}, \quad \lvert e^{\frac{C_+}{2} t^2} \rvert \geq e^{c_2 t^2}.
\end{equation}
Hence standard application of the saddle point method yields
\begin{equation} \label{eq:bulk_est_insenssial_1}
  \pv \int_{\Gamma^+_{\local}}  dw \int_{\Sigma^+_{\local}} dz \frac{e^{F(z; x)}}{e^{F(w; x)}} \frac{1}{z - w} = \bigO(N^{\frac{1}{2}}).
\end{equation}
Analogous reasoning gives
\begin{equation} \label{eq:bulk_est_insenssial_2}
  \pv \int_{\Gamma^-_{\local}}  dw \int_{\Sigma^-_{\local}} dz \frac{e^{F(z; x)}}{e^{F(w; x)}} \frac{1}{z - w} = \bigO(N^{\frac{1}{2}}).
\end{equation}

Finally, by \eqref{eq:global_steepest_descent_Sigma_1}, \eqref{eq:global_steepest_descent_Sigma_2} and \eqref{eq:global_steepest_descent_Sigma_3} in Lemma \ref{lem:global_steepest_descent}, there exists $c_3 > 0$ such that for large enough $N$
\begin{align}
  \lvert e^{-F(w; x)} \rvert < {}& \left\lvert e^{-F(Nu^{\mathrm{L}}_{\pm}; x)} \right\rvert e^{-c_3 N^{\frac{1}{5}}} \quad \text{if $w \in \Gamma_{\alpha} \setminus \Gamma^{\pm}_{\local}$}, \label{eq:global_I1_1} \\
  \lvert e^{F(z; x)} \rvert < {}&
  \begin{cases}
    \left\lvert e^{F(Nu^{\mathrm{L}}_{\pm}; x)} \right\rvert e^{-c_3 N^{\frac{1}{5}}} & \text{if $z \in \Sigma \setminus \Sigma^{\pm}_{\local}$}, \\
    \left\lvert e^{F(Nu^{\mathrm{L}}_{\pm}; x)} \right\rvert e^{-c_3 \lvert z \rvert} & \text{if $z \in \Sigma \setminus \{ \lvert z \rvert > \frac{N}{c_3} \}$}.
  \end{cases} \label{eq:global_I1_2}
\end{align}
With the help of estimates \eqref{eq:global_I1_1}, \eqref{eq:global_I1_2}, \eqref{eq:F_near_saddle_pt} and \eqref{eq:estimate_near_saddle_pt}, we obtain
\begin{multline} \label{eq:bulk_est_insenssial_3}
  \pv \int_{\tilde{\Gamma}_{\alpha}} dw \int_{\Sigma_{-1}}  dz \frac{e^{F(z; x)}}{e^{F(w; x)}} \frac{1}{z - w} - \pv \int_{\Gamma^+_{\local}}  dw \int_{\Sigma^+_{\local}} dz \frac{e^{F(z; x)}}{e^{F(w; x)}} \frac{1}{z - w} \\
  - \pv \int_{\Gamma^-_{\local}}  dw \int_{\Sigma^-_{\local}} dz \frac{e^{F(z; x)}}{e^{F(w; x)}} \frac{1}{z - w} = \bigO(e^{-\epsilon n^{\frac{1}{5}}}).
\end{multline}

Plugging \eqref{eq:bulk_est_insenssial_1}, \eqref{eq:bulk_est_insenssial_2} and \eqref{eq:bulk_est_insenssial_3} into \eqref{eq:I_1_integral}, we have that
\begin{equation} \label{eq:result_I_2}
  I_1 = \bigO(N^{- \frac{1}{2}}).
\end{equation}
Therefore we have proved Proposition \ref{prop:global_density} upon combining \eqref{eq:K=I_1+I_2}, \eqref{eq:exact_value_of_I_2} and \eqref{eq:result_I_2}.
\end{proof}

\begin{remark} \label{rmk:density_saddle_relation}
  The asymptotic analysis above confirms a simple relation that
  \begin{equation} \label{eq:density_saddle_relation}
    \tilde{\rho}^{\mathrm{L}}_{(1)}(x) = \lim_{N \to \infty} \tilde{K}^{\mathrm{L}}(x) = \frac{\Im u^{\mathrm{L}}_+}{\pi x}.
  \end{equation}
  This pattern persists in the computation of the global density for the Jacobi Muttalib--Borodin ensemble given in \S \ref{subsec:Jacobi_universality}.
\end{remark}

\begin{remark}
  The asymptotic analysis above can also prove the bulk local universality of the limiting distribution of the eigenvalues, as in \cite{Liu-Wang-Zhang14}. The result, which is unsurprisingly the sine universality, is also the same as in \cite{Liu-Wang-Zhang14}. Since the local universality is not our focus in this paper, we omit further discussion on it. In the case $\theta = 2$ this problem, along with Airy kernel universality at the soft edge, was solved some time ago \cite{Lueck-Sommers-Zirnbauer06}. Very recently, these universalities have been established for general $\theta > 0$ by Zhang \cite{Zhang15} according to the method of \cite{Liu-Wang-Zhang14}. In \cite{Zhang15}, an equivalent form of \eqref{KK2} is also derived.
\end{remark}

\subsubsection{The $\theta = 0$ case} \label{subsubsec:theta=0_L}

In Proposition \ref{PC1}, we do not use the most straightforward scaling transform
\begin{equation} \label{eq:scaling_theta=0}
  x \mapsto \frac{x}{N}
\end{equation}
to compute the limiting global density but rather we took the combined change of variables and scaling transform
\begin{equation} \label{eq:scaling_theta>0}
  x \mapsto \left( \frac{x}{N\theta} \right)^{\theta},
\end{equation}
like in Proposition \ref{prop:global_density}, to conform our result to the Fuss--Catalan distribution. It is clear that the combined change of variables and scaling transform \eqref{eq:scaling_theta>0} is not well defined for $\theta = 0$. So in the $\theta = 0$ case, which is the particularly interesting $\theta \to 0$ limit of the Laguerre Muttalib--Borodin ensemble \eqref{Pd3A}, we use the scaling transform \eqref{eq:scaling_theta=0} instead.

According to (\ref{KK2}), we then have
\begin{equation}\label{KK2p}
  K^{\mathrm{L}}(x) := K^{\rm L}(Nx,Nx) = {1 \over (2 \pi i)^2}  \oint_{\Sigma} dz  \oint_{\Gamma_\alpha} dw \,
  { x^{-z - 1} x^w \Gamma(z+1)(z - c)^N \over
    (z - w) \Gamma(w+1) (w - c)^N} N^{w - z - 1}.
\end{equation}
Changing variables $z=Nu$, $w = Nv$ and using Stirling's formula shows that for large $N$, if we can deform $\Sigma$ and $\Gamma_{\alpha}$ such that $\arg u, \arg v \in (-\pi + \epsilon, \pi - \epsilon)$ and $\lvert u \rvert, \lvert v \rvert > \epsilon > 0$,
$$
 K^{\rm L}(x) = {1 \over Nx (2 \pi i)^2}  \oint_{\Sigma} dz  \oint_{\Gamma_\alpha} dw \,
 {1 \over z - w} {e^{N H_1(u;x) + {1 \over 2} \log u - \frac{c}{u}} \over e^{N H_1(v;x) + {1 \over 2} \log v - \frac{c}{v}} }
 \Big ( 1 + \bigO \Big ( {1 \over N} \Big ) \Big ),
 $$
 where $H_1(u;x) := u (\log u - 1) - u \log x + \log u$. Noting that
 $$
 {d \over du} H_1(u;x) = \log \Big ( {u \over x} e^{1/u} \Big )
 $$
 we see that the stationary points of $H_1(u;x)$ occur when
 \begin{equation} \label{3.3}
 u e^{1/ u} = x, \quad \text{or equivalently} \quad \left( -\frac{1}{u} \right) e^{-\frac{1}{u}} = -\frac{1}{x}.
 \end{equation}
This equation is solved by the Lambert $W$ function \cite[Sec.~4.13]{Boisvert-Clark-Lozier-Olver10} --- recall too Remark \ref{rmk:Fuss--Catalan_distr}(i) --- and the two solutions are complex conjugates
\begin{equation}
  u^{\mathrm{L}}_+ = -W \left( -\frac{1}{x} \right), \quad u^{\mathrm{L}}_- = -\overline{W \left( -\frac{1}{x} \right)}.
\end{equation}
Then using steepest descent arguments like in the $\theta > 0$ case, we have
\begin{prop}
  As $\theta \to 0$, the limiting global density of the Muttalib--Borodin ensemble is
  \begin{equation}\label{4.55}
    \lim_{N \to \infty} K^{\mathrm{L}}(x) = - \Im \frac{1}{\pi x W(-x^{-1})}.
  \end{equation}
\end{prop}
The proof is omitted since it is similar to the $\theta > 0$ case. Note that the relation \eqref{eq:density_saddle_relation} in Remark \ref{rmk:density_saddle_relation} still holds.
The result (\ref{4.55}) was derived in \cite[Cor.~1]{Cheliotis14} by a determination of the moments, with the latter also known from \cite{Dykema-Haagerup04}.

\subsection{The Jacobi Muttalib--Borodin ensemble} \label{subsec:Jacobi_universality}

In the case of the Jacobi Muttalib--Borodin ensemble, comparing the correlation kernel formula \eqref{KK1} with the correlation kernel \eqref{eq:corr_kernel_J}, we have, analogous to \eqref{KK2},
\begin{multline}\label{KK2j}
  K^{\rm J}(x,y) = {1 \over (2 \pi i)^2} {\tilde{h}(x) \over \tilde{h}(y)} \\
  \times \oint_{\Sigma} dz  \oint_{\Gamma_\alpha} dw \,
  { x^{-z - 1} y^w \over (z - w) }
  { \Gamma(w + c_2 + N + 1) \Gamma(z+1) \prod_{k=1}^N ( z - \alpha_k) \over
   \Gamma(z + c_2 + N + 1)  \Gamma(w+1) \prod_{l=1}^N (w - \alpha_l)},
\end{multline}
where $\alpha_j = \theta(j - 1) + c_1$ as specified in \eqref{XY1} and $c = c_1$. Analogous to \eqref{KK2} we have that the factor $\tilde{h}(x)/\tilde{h}(y)$ in (\ref{KK2j}) cancels out of the determinant (\ref{KK}), and the reasoning leading to (\ref{hh}) tells us that we can take
\begin{equation}\label{hhj}
  \tilde{h}(x) = x^{c_1/2} (1 - x)^{-c_2/2}, \quad \text{because} \quad \frac{\tilde{h}(x)}{\tilde{h}(y)} = \Big ({x \over y} \Big )^{c_1/2} \Big ({1 - x \over 1 - y} \Big )^{-c_2/2}.
\end{equation}
Here we assume $\theta > 0$. The $\theta = 0$ case of the Jacobi upper-triangular ensemble, that is, $\alpha_1 = \dotsb = \alpha_N = c_1$, is well defined and can be thought as the $\theta \to 0^+$ limit of the Jacobi Muttalib--Borodin ensemble, but we omit it in this paper. Our aim is to use a saddle point analysis to compute
\begin{equation}\label{hhj1}
  \tilde{\rho}^{\rm J}_{(1)}(x) = \lim_{N \to \infty} \tilde{K}^{\mathrm{J}}(x) := \lim_{N \to \infty} \frac{1}{N\theta} x^{1/\theta - 1} K^{\rm J}( x^{1/\theta}, x^{1/\theta}),
\end{equation}
which is the limiting density of the eigenvalues under the change of variables $x \to x^{\theta}$. The required workings is structurally identical to that just given to derive Proposition \ref{prop:global_density};
only a brief sketch will be given below.

Recall that in the special case that $\theta \in \intZ^+$, the resolvent $\tilde{G}^{\mathrm{J}}(z)$ defined in \eqref{eq:defn_tilde_G^J}, satisfies the identity \eqref{G8}, and then the global density $\tilde{\rho}^{\mathrm{J}}(x)$ satisfies
\begin{equation} 
  \tilde{\rho}^{\rm J}_{(1)}(x)  = \lim_{\epsilon \to 0^+} {1 \over 2 \pi i}
  \Big ( \tilde{G}^{\rm J}(x - i \epsilon) - \tilde{G}^{\rm J}(x + i \epsilon) \Big ), \quad
  x \in (0,1).
\end{equation}
Below we show that the global density $\tilde{\rho}^{\mathrm{J}}(x)$ has an explicit formula given as follows. For all $\theta > 0$ and all $x \in (0, 1)$, there is a unique $\varphi \in (0, \pi/(\theta + 1))$ such that
\begin{equation} \label{eq:para_Jacobi}
  x = \frac{\theta^{\theta}}{(1 + \theta)^{1 + \theta}} \frac{\sin((\theta + 1)\varphi)^{\theta + 1}}{\sin \varphi \sin(\theta\varphi)^{\theta}}.
\end{equation}
We define
\begin{equation} \label{eq:v_for_Jacobi}
  v(\varphi) = \frac{\theta \sin((\theta + 1)\varphi)}{(1 + \theta)\sin(\theta \varphi)} e^{i\varphi},
\end{equation}
and then the density function 
\begin{equation*}
  \begin{split}
    \rho^{\JFC}(x) = {}& \frac{1}{\pi x} \Im \left( \frac{1}{\theta} \frac{v(\varphi)}{1 - v(\varphi)} \right) \\
    = {}& \frac{1}{(\pi x)} \frac{(1 + \theta)\sin(\theta \varphi)\sin\varphi \sin((\theta + 1)\varphi)}{\sin(\theta \varphi)^2 + \theta^2\sin(\varphi)^2 - 2\theta \cos((\theta + 1)\varphi) \sin(\theta\varphi) \sin(\varphi)}.
  \end{split}
\end{equation*}

Thus we have, parallel to Proposition \ref{prop:global_density}
\begin{prop} \label{prop:global_densityJ}
  For all $\theta > 0$, the limiting global density of the Jacobi Muttalib--Borodin ensemble with change of variable \eqref{Cv} is
  \begin{equation*}
    \tilde{\rho}^{{\rm J}}_{(1)}(x) = \rho^{\JFC}(x).
  \end{equation*}
\end{prop}

\begin{remark}
  The distribution $\rho^{\JFC}(x)$ is the counterpart of the Fuss--Catalan distribution $\rho^{\FC}(x)$ in the Laguerre case. 
  According to the results of \S \ref{S3.2}, for $\theta \in \mathbb Z^+$,
   the measure $d\mu(x) = \rho^{\JFC}(x) dx$ is the unique probability measure supported on $[0, 1]$ that makes the following identity hold:
  \begin{equation}\label{GGa}
    z(zG(z) - 1)(z G(z) + 1/\theta)^{\theta} = (zG(z))^{\theta + 1}, \quad \text{where} \quad G(z) = \int_{[0, 1]} \frac{\rho^{\JFC}(x) dx}{z - x}.
  \end{equation}
  This is analogous to the resolvent for the  Fuss--Catalan distribution satisfying the identity $z(zG(z) - 1) = (zG(z))^{\theta + 1}$ hold. 
 Of interest is to extend  the characterisation (\ref{GGa}) to general $\theta > 0$. Such a characterisation was established in \cite{Gawronski-Neuschel-Stivigny14} 
 for the measures corresponding to the moments (\ref{Ja}) with $r>s$, whereas as remarked below (\ref{GSa}) the moments (\ref{mL}) deduced
 from (\ref{GGa}) require $r=s$.
\end{remark}
The method of the proof to Proposition \ref{prop:global_densityJ} is the same as the proof of Proposition \ref{prop:global_density} for the Laguerre case. So we only give a sketch of the proof and point out the main differences. 

\begin{proof}[Sketch of the proof of Proposition \ref{prop:global_densityJ}]
  First we consider the case that $c_2 \in \intZ$. Making the replacements $x \mapsto x^{1/\theta}$, $y \mapsto y^{1/\theta}$ in \eqref{KK2j} gives
  \begin{equation} \label{eq:double_contour_densityJ}
    \tilde{K}^{\rm J}(x) := {1 \over N \theta}
    x^{\frac{1}{\theta} - 1} K^{\rm J}(x^{\frac{1}{\theta}}, x^{\frac{1}{\theta}}) = \frac{1}{\theta x} \frac{1}{(2\pi i)^2} \oint_{\Sigma} dz  \oint_{\Gamma_\alpha} dw \frac{e^{G(z; x)}}{e^{G(w; x)}} \frac{1}{z - w},
  \end{equation}
  where
  \begin{equation}
    G(z; x) = \log \left( \frac{\Gamma(z + 1)}{\Gamma(z+N+c_2 +1)} x^{\frac{z}{\theta}} \prod^N_{k = 1} (z - c_1 - (k - 1)\theta) \right).
  \end{equation}
  
  Changing variables in the integrand according to (\ref{zuv}), then use of Stirling's formula and (\ref{4.10}) shows that
  for large $N$
  \begin{equation} \label{eq:G_in_G_hat}
    G(z; x) = N \hat{G}_1(u; x) + \hat{G}_0(u)  + N \log \theta -c_2 \log N + \bigO(N^{-1}),
  \end{equation}
  where, after the change of variables of $z, w$ into $u, v$ as in \eqref{zuv},
  \begin{align}
    \hat{G}_1(u; x) & = \theta u \log {\theta u \over 1 + \theta u}  - \log(1 + \theta u)  - u\log x + (1 - u)
                      \log (1 - {1 \over u}) + \log u,  \label{4.13G} \\
    \hat{G}_0(u) & = {1 \over 2} \log  {\theta u \over 1 + \theta u} - c_2 \log (1 + \theta u) +
                      \Big ( {c \over \theta} - {1 \over 2} \Big ) \log (1 - {1 \over u} ),
  \end{align}
  and the validity of the error bound $ \bigO(N^{-1})$ is the same as in (\ref{eq:F_in_F_hat}), \ie, all $z$ satisfying \eqref{eq:z_away_from_interval} and $\arg z = \arg u \in (-\pi + \epsilon, \pi - \epsilon)$.
  Substituting in (\ref{eq:double_contour_densityJ}) gives
  \begin{equation} \label{eq:double_contour_densityJ1}
    \tilde{K}^{\rm J}(x) = \frac{1}{N \theta x} \frac{1}{(2\pi i)^2} \oint_{\Sigma} dz  \oint_{\Gamma_\alpha} dw \frac{e^{N\hat{G}_1(u; x)
        + \hat{G}_0(u)}}{e^{N\hat{G}_1(v; x)
        + \hat{G}_0(v)}} \frac{1}{z - w} \Big (1 + \bigO(N^{-1}) \Big ).
  \end{equation}
  This is structurally identical to \eqref{eq:double_contour_density1}, and is analysed accordingly. For the sake of steepest-descent analysis, we need to deform the contours $\Sigma$ and $\Gamma_{\alpha}$. Since we assume that $c_2 \in \intZ$, $\Gamma_{\alpha}$ is a finite contour similar to the $\Gamma_{\alpha}$ in \eqref{eq:double_contour_density}. It is clear that $\Sigma$ can be deformed to an infinite contour that is from $-i \cdot \infty$ to $i  \cdot \infty$,as long as it keeps the poles $-1, \dotsc, -(\beta + N)$ and $\Gamma_{\alpha}$ to the left. Then we see that we can split $\Gamma_{\alpha}$ into two positively oriented contours that jointly enclose all the poles $\alpha_1, \dotsc, \alpha_N$, and let $\Sigma$ be an infinite vertical contour passing between them. Note that the precise shapes of $\Sigma$ and $\Gamma_{\alpha}$ in the Laguerre case depend on the computation of the critical points of $\hat{F}_1(u; x)$ that is the counterpart of our $\hat{G}_1(u; x)$, see \eqref{eq:U^L_pm} and \eqref{eq:Gamma_constr_Lag}. Below we explain the analogous construction of $\Sigma$ and $\Gamma_{\alpha}$ in the Jacobi case.

  A difference in detail is now that
  $$
  {d \over du} \hat{G}_1(u;x) = \log \bigg ( \Big ( {\theta u \over 1 + \theta u} \Big )^\theta \Big ( {u \over u - 1} \Big ) {1 \over x}
  \bigg ),
  $$
  so the stationary points now occur for $u$ such that
  \begin{equation}\label{4.16J}
    u^{\theta + 1} = x(1/\theta + u)^\theta (u-1).
  \end{equation}
  Analogous to the situation with (\ref{ux}), we observe
  that this is precisely the equation 
  (\ref{G8}) after the identification $z \mapsto x$,
  $z \tilde{G^{\rm J}}(z) \mapsto u$ in the latter. It is straightforward to verify that if $x \in (0, 1)$ is parametrised by $\varphi \in (0, \pi/(\theta + 1))$ in \eqref{eq:para_Jacobi}, then with $v(\varphi)$ defined in \eqref{eq:v_for_Jacobi}
  \begin{equation*}
    u^{\mathrm{J}}_+ = \frac{1}{\theta} \frac{v(\varphi)}{1 - v(\varphi)}, \quad u^{\mathrm{J}}_- = \overline{u^{\mathrm{J}}_+} = \frac{1}{\theta} \frac{\overline{v(\varphi)}}{1 - \overline{v(\varphi)}}
  \end{equation*}
  are two solutions to \eqref{4.16J}. As $x$ runs over $(0, 1)$, the locus of $u^{\mathrm{J}}_+$ (\resp\ $u^{\mathrm{J}}_-$) is a curve in $\compC_+$ (\resp\ $\compC_-$) whose two ends are on the real line. Thus we let $\Sigma$ be the upward vertical contour through $N\theta u^{\mathrm{J}}_{\pm}$, analogous to \eqref{eq:Simaga_-1_vertical} in the Laguerre case. To define $\Gamma_{\alpha}$, we first define
  \begin{equation*}
    \tilde{\Gamma}^{\mathrm{J}} = \left\{ \frac{1}{\theta} \frac{v(\phi)}{1 - v(\phi)} \middle| \phi \in \left[ -\frac{\pi}{\theta + 1}, \frac{\pi}{\theta + 1} \right] \right\}
  \end{equation*}
  analogous to $\tilde{\Gamma}$ in \eqref{eq:Gamma_constr_Lag}, and  then deform it to the desired $\Gamma_{\alpha}$ parallel to the deformation carried out in the Laguerre case. We omit the detail, but only point out that the shapes of $\Sigma$ and $\Gamma_{\alpha}$ is like those in Figure \ref{fig:contours} if $\theta \geq 1$ or Figure \ref{fig:small_theta_all} if $\theta \in (0, 1)$.

  By saddle-point analysis analogous to that in the Laguerre case, we derive, as expected in Remark \ref{rmk:density_saddle_relation},
  \begin{equation*}
    \tilde{\rho}^{\mathrm{J}}_{(1)}(x) = \lim_{N \to \infty} \tilde{K}^{\mathrm{J}}(x) = \frac{\Im u^{\mathrm{J}}_+}{\pi x} = \rho^{\JFC}(x),
  \end{equation*}
  and hence finish the proof.

  In the case that $c_2 \notin \intZ$, the contour $\Gamma_{\alpha}$ in \eqref{KK2j} is infinite. We express it as the combination of $\Gamma'_{\alpha} \cup \Gamma''_{\alpha}$. Here $\Gamma'_{\alpha}$ is an infinite, Hankel like contour starting at $-\infty - i\epsilon$, running parallel to the negative real axis, looping around the poles $w = -(c_2 + N + k)$ with $k \in \intZ_+$, and finishing at $-\infty + i\epsilon$ after again running parallel to the negative real axis. $\Gamma''_{\alpha}$ is a finite positive oriented contour enclosing $\alpha_1, \dotsc, \alpha_N$. Then by argument of contour deformation used in the $c_2 \in \intZ$ case and the Laguerre case, we can deform $\Sigma$ into a contour going from $-i \cdot \infty$ to $i \cdot \infty$ between $\Gamma'_{\alpha}$ and $\Gamma''_{\alpha}$. We write 
  \begin{multline*}
    \tilde{K}^{\mathrm{J}}(x) = \tilde{K}^{\mathrm{J},}{}'(x) + \tilde{K}^{\mathrm{J},}{}''(x), \quad \text{where} \\
    \tilde{K}^{\mathrm{J}, *}(x) = \frac{1}{\theta x} \frac{1}{(2\pi i)^2} \oint_{\Sigma} dz  \oint_{\Gamma^*_\alpha} dw \frac{e^{G(z; x)}}{e^{G(w; x)}} \frac{1}{z - w}, \quad
    * = {}' \text{ or }''.
  \end{multline*}

  In the computation of $\tilde{K}^{\mathrm{J},}{}''(x)$, we further deform the contours $\Gamma''_{\alpha}$ into two parts as the deformation of $\Gamma_{\alpha}$ in the $c_2 \in \intZ$ case, and let $\Sigma$ go between the two parts, like the deformed shapes of $\Gamma''_{\alpha}$, $\Gamma'_{\alpha}$ and $\Sigma$ in the $c_2 \in \intZ$ case. See Figure \ref{fig:c2_not_int} if $\theta \geq 1$. By the argument in the $c_2 \in \intZ$ case we have that
  \begin{equation*}
    \lim_{N \to \infty} \tilde{K}^{\mathrm{J},}{}''(x) = \rho^{\JFC}(x).
  \end{equation*}
  On the other hand, by estimating the factors of the integrand
  \begin{equation*}
    x^{-z - 1} \Gamma(z+1) \prod_{k=1}^N ( z - \alpha_k) \Gamma(z + c_2 + N + 1)^{-1}
  \end{equation*}
  by Stirling's formula on $\Sigma_{-1/2}$ and
  \begin{equation*}
    y^w \Gamma(w + c_2 + N + 1) \left( \Gamma(w+1) \prod_{l=1}^N (w - \alpha_l) \right)^{-1}
  \end{equation*}
  on $\Gamma'_{\alpha}$ separately, we find that
  \begin{equation*}
    \lim_{N \to \infty} \tilde{K}^{\mathrm{J},}{}'(x) = 0.
  \end{equation*}
  Thus we prove the proposition by combining the two limits above.
\end{proof}

\begin{figure}[htb]
  \begin{minipage}[t]{0.45\linewidth}
    \centering
    \includegraphics{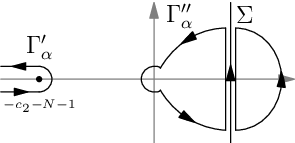}
    \caption{The deformed shape of contours $\Sigma$ and $\Gamma_{\alpha}$ for the computation of the global density with $\theta \geq 1$ and $c_2 \notin \intZ$.}
    \label{fig:c2_not_int}
  \end{minipage}
  \hspace{\stretch{1}}
  \begin{minipage}[t]{0.45\linewidth}
    \includegraphics{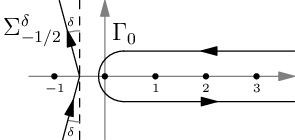}
    \caption{The shape of contours $\Sigma^{\delta}_{-1/2}$ and $\Gamma_0$ for the hard edge correlation kernel.}
    \label{fig:hard_edge}
  \end{minipage}
\end{figure}

  \section{Hard edge density}\label{sec:hed}

\subsection{The Laguerre Muttalib--Borodin ensemble}

\subsubsection{The $\theta > 0$ case} \label{subsubsec:hard_edge_L}

The meaning of the hard edge scaling has already been noted in Section \ref{subsubsec:theta>0_L}. It is the
$N \to \infty$ limit with $x,y$ scaled  in
the neighbourhood of the origin so that the spacing between eigenvalues is of order unity. The limiting kernel with proper scaling is stated in Proposition \ref{prop:Laguerre_hard_edge}, and below we give the proof.

\begin{proof}[Proof of Proposition \ref{prop:Laguerre_hard_edge}]
  First we use the argument in Section \ref{subsec_multiple_Laguerre_ensemble}, more specifically the paragraph above \eqref{eq:Simaga_-1_vertical}, to justify that the double contour integral formula \eqref{KK2} still holds with $\Sigma$ deformed into a contour from $e^{-(\pi/2 + \delta)i} \cdot \infty$ to $e^{(\pi/2 + \delta)i} \cdot \infty$ and going on the left of $\Gamma_{\alpha}$. Thus analogous to \eqref{eq:Simaga_-1_vertical} we can take $\Sigma$ to be $\Sigma^{\delta}_{-1/2}$ in \eqref{KK2}. Furthermore it is clear that the closed contour $\Gamma_{\alpha}$ can be deformed to the infinite Hankel loop contour $\Gamma_0$.

  Substituting (\ref{XY1}) and making use of the functional and reflection formula for the gamma function shows that
  $$
  {\prod_{k=1}^N (z - \alpha_k) \over \prod_{k=1}^N (w - \alpha_k)} =
  {\Gamma((z-c)/\theta + 1) \Gamma(N - (z-c)/\theta) \sin \pi (z - c)/\theta \over
    \Gamma((w-c)/\theta + 1) \Gamma(N - (w-c)/\theta) \sin \pi (w - c)/\theta }.
  $$
  It is standard result that for $z$ and $w$ fixed and $N \to \infty$,
  $$
  { \Gamma(N - (z-c)/\theta)  \over  \Gamma(N - (w-c)/\theta) } = N^{(w - z)/\theta} \Big ( 1 + O\Big ( {1 \over N} \Big )  \Big ).
  $$
  Taking the limit $N \to \infty$ inside the integral (a step which is justified by estimating the integrand and using
  dominated convergence; see \cite[Sec.~5.2]{Kuijlaars-Zhang14} for the details in a very similar setting) we thus see that
  \begin{multline*}
    \lim_{N \to \infty} N^{-1/\theta} K^{\rm L}(N^{-1/\theta} x, N^{-1/\theta} y) = \Big ( {x \over y } \Big )^{c/2} \\
    \times {1 \over (2 \pi i)^2} \oint_{\Sigma^{\delta}_{-1/2}} dz  \oint_{\Gamma_0} dw \,
    { x^{-z - 1} y^w \over z - w}
    {\Gamma(z+1) \Gamma((z-c)/\theta + 1) \sin \pi (z - c)/\theta \over
      \Gamma(w+1) \Gamma((w-c)/\theta + 1) \sin \pi (w - c)/\theta}.
  \end{multline*}
  The result (\ref{KK3}) now follows by the change of variables $z \mapsto \theta z + c$,  $w \mapsto \theta w + c$, and then deforming the contours back to
  $\Sigma^{\delta}_{-1/2}$ and $\Gamma_0$. (Since the deformation does not cross any pole, it does not change the value of integral). 
  
  In the special case $\theta \in \mathbb Z^+$ the duplication formula can be used to rewrite
  $\Gamma(\theta z + c + 1) / \Gamma(\theta w + c + 1)$ and (\ref{KK3a}) results.
\end{proof}

\medskip
Borodin has previously given a different formula for the hard edge scaled limit in 
(\ref{KK3}).

\begin{prop}\label{P5.2} \cite{Borodin99}
We have
\begin{equation}\label{KK9}
\lim_{N \to \infty} N^{-1/\theta} K^{\rm L}(N^{-1/\theta} x, N^{-1/\theta} y)  =
\Big ( { y \over x} \Big )^{c/2}  {K}^{(c,\theta)}(x,y)
\end{equation}
where \footnote{Our notation $K^{(c, \theta)}(x, y)$ agrees with the notation $K^{(\alpha, \theta)}(x, y)$ defined in \cite[Eq.~(5.3)]{Kuijlaars-Stivigny14}, but differs from the $\mathcal{K}^{(\alpha, \theta)}(x, y)$ in \cite[Eq.~(3.6)]{Borodin99} by a factor.}
\begin{equation}\label{BK}
{K}^{(c,\theta)}(x,y) = \theta x^c \int_0^1 J_{(c+1)/\theta,1/\theta}(xu) J_{c+1,\theta}(yu)^\theta) u^c \, du,
\end{equation}
with  $J_{a,b}$ denoting  Wright's generalisation of the Bessel function given by
\begin{equation}\label{KK9a}
J_{a,b}(x) = \sum_{j=0}^\infty {(-x)^j \over j! \Gamma (a + j b)}.
\end{equation}
\end{prop}

Equating (\ref{KK9}) and (\ref{KK3}) gives us a double contour integral form of the kernel
(\ref{BK}). Before doing this,
 we note that the 
 term in the second line on the RHS of (\ref{KK3a}) can be identified with the hard edge scaled kernel 
$K_{\nu_1,\dots,\nu_M}(x,y)$ of
Kuijlaars and Zhang \cite{Kuijlaars-Zhang14}, which came about from the hard edge scaled limit of the correlation kernel
for the product of complex standard Gaussian rectangular random matrices \cite{Akemann-Ipsen-Kieburg13}. We also note that the hard edge scaled kernel is a component of the kernel for the Meijer G random point field, see \cite{Bertola-Gekhtman-Szmigielski09}, \cite{Bertola-Gekhtman-Szmigielski14} and \cite{Bertola-Bothner14}. This kernel reads
\begin{equation}\label{KK3p}
  K_{\nu_1,\dots,\nu_M}(x,y) = {1 \over (2 \pi i)^2} \int_{\Sigma^{\delta}_{-1/2}} dz  \oint_{\Gamma_0} dw \,
{ x^{w} y^{-z-1} \over z - w} \prod_{j=0}^M {\Gamma(z+1+\nu_j) \over \Gamma(w+1+\nu_j)} 
{\sin \pi z \over \sin \pi w},
\end{equation}
where  $\nu_0 :=0$, $\nu_j > -1$ ($j=1,\dots,M$). Denoting the second line on the RHS of (\ref{KK3a}) by
$\tilde{K}^{(c,\theta)}(x,y)$, for $\theta \in \mathbb Z^+$ we see that
\begin{equation}\label{KK4}
\tilde{K}^{(c,\theta)}(x,y)  = K_{\nu_1,\dots,\nu_\theta}(y,x) , \qquad \nu_j = {c \over \theta} - 1 + {j \over \theta} \: \:\:
(j=1,\dots,\theta).
\end{equation}
Recalling (\ref{KK9}) and the LHS of (\ref{KK3a}) we thus have
\begin{equation}\label{KK4a}
x^{1/\theta - 1} {K}^{(c,\theta)}(\theta x^{1/\theta}, \theta y^{1/\theta})  = K_{\nu_1,\dots,\nu_\theta}(y,x) ,
\end{equation}
for $\{ \nu_j  \}$ as in (\ref{KK4}). This equation has been deduced using a rewrite of (\ref{BK}) in
\cite[Eq.~(5.7)]{Kuijlaars-Stivigny14}.

We now specify the double contour integral formula for Borodin's kernel (\ref{BK}).

\begin{corollary}
We have
\begin{equation}\label{BK2}
K^{(c,\theta)}(x, y)  = {\theta \over (2 \pi i)^2} \oint_{\Sigma^{\delta}_{-1/2}} dz  \oint_{\Gamma_0} dw \,
{ x^{-\theta z - 1} y^{\theta w} \over z - w}
{\Gamma(\theta z + c + 1) \over \Gamma(\theta w + c + 1) }
{\Gamma(z+1) \over \Gamma(w+1)} {\sin \pi z \over \sin \pi w}.
\end{equation}
\end{corollary}

\begin{remark}
  \begin{enumerate}[label=(\roman*)]
  \item
    Our proof of \eqref{BK2} is indirect. A direct proof can be obtained via the Meijer-G function representation of $J_{a, b}(x)$ in \eqref{KK9a}, analogous to the proof of \cite[Thm.~5.1]{Kuijlaars-Stivigny14}.
  \item
    Here we use a nearly vertical contour $\Sigma^{\delta}_{-1/2}$ to be in consistent with \cite{Kuijlaars-Zhang14}. It is also possible to deform $\Sigma$ into a Hankel-like form, which is symmetric to $\Gamma_0$. With the help of this symmetry, and under the change of variables $z \mapsto - (z + (c+1)/\theta)$, $w \mapsto - (w + (c+1)/\theta)$
    in (\ref{BK2}) the contours $\Sigma$ and
    $\Gamma_0$ interchange. Making use of the reflection formula for the gamma function then allows
    us to deduce that
    \begin{equation}\label{BK2b}
      {1 \over \theta} x^{1/\theta - 1} K^{(\alpha,\theta)}(x^{1/\theta}, y^{1/\theta}) =
      \Big ( {x \over y} \Big )^{\alpha'} K^{(\alpha',1/\theta)}(y,x),
    \end{equation}
    which as noted in \cite{Kuijlaars-Stivigny14} also follows from the original form (\ref{BK}).
    \item When our article was almost complete, we received a preprint from Zhang \cite{Zhang15}, containing amongst other things
    an independent analysis of the hard edge scaling of the Laguerre Muttalib--Borodin ensemble.
  \end{enumerate}
\end{remark}

\subsubsection{The $\theta = 0$ case} \label{subsubsec:Laguerre_hard_edge_L0}

We could take the limit of the double contour integral formula \eqref{KK2} with $\alpha_1 = \dotsb = \alpha_N = c$ to derive the limiting correlation kernel near $0$ for the $\theta = 0$ case of the Laguerre Muttalib--Borodin ensemble, but here we introduce an alternative approach. Note that the joint PDF of the $\theta = 0$ Laguerre Muttalib--Borodin ensemble is given in \eqref{Pd3A}. Changing variables $\log \lambda_j \mapsto \mu_j$, we have that the joint PDF for $\mu_1 \geq \mu_2 \geq \dotsb \geq \mu_N$ is proportional to
\begin{equation*}
  \prod^N_{k = 1} e^{-V(\mu_k)} \prod_{1 \leq j < k \leq N} (e^{\mu_j} - e^{\mu_k})(\mu_j - \mu_k), \quad \text{where} \quad V(x) = -(c + 1)x + e^x.
\end{equation*}
Then the results in \cite{Claeys-Wang11} indicate that the correlation functions for the smallest variables $\mu_N, \mu_{N - 1}, \dotsc$ converge to the correlation functions in the Tracy-Widom distribution, upon proper scaling. More specifically, the results in \cite{Claeys-Wang11} are only concerned with the asymptotics of biorthogonal polynomials. In principle, by summing up the products of the biorthogonal polynomials, the asymptotics of the correlation kernel results, but technically this is nontrivial. Note that in \cite{Claeys-Wang11} there is a technical assumption that $V(x) \to +\infty$ faster than any linear equation as $x \to \pm \infty$. But it is not hard to see that in the $-\infty$ direction this requirement can be relaxed to a linear growth to $+\infty$.

\subsection{The Jacobi Muttalib--Borodin ensemble}

The hard edge scaled limit of the Jacobi Muttalib--Borodin ensemble is the same as that for the
Laguerre Muttalib--Borodin ensemble, although the specific scale is different \cite{Borodin99}.

\begin{prop}
Let $c_1 = c$ in (\ref{J}), and specify the ratio $\tilde{h}(x)/ \tilde{h}(y)$ occurring in the definition (\ref{KK2j})
by (\ref{hhj}). We have
\begin{equation} \label{eq:Jacobi_Laguerre_same}
\lim_{N \to \infty} N^{-1 - 1/\theta} K^{\rm J}( N^{-1 - 1/\theta}x,  N^{-1 - 1/\theta}y) =
\lim_{N \to \infty} N^{-1/\theta} K^{\rm L}(N^{-1/\theta}x, N^{-1/\theta} y).
\end{equation}
\end{prop}

\begin{proof}
  First we consider the case that $c_2 \in \intZ$. Then the contour $\Gamma_{\alpha}$ in the integral \eqref{KK2j} is a finite contour enclosing $\alpha_1, \dotsc, \alpha_N$. By the argument similar to that in Section \ref{subsec:Jacobi_universality}, we can deform the contour $\Sigma$ in \eqref{KK2j} into a vertical contour going from $e^{-(\pi/2 + \delta)i} \cdot \infty$ to $e^{(\pi/2 + \delta)i} \cdot \infty$ on the left of $\Gamma_{\alpha}$. Furthermore we can deform the closed contour $\Gamma_{\alpha}$ into the infinite contour $\Gamma_0$. Thus the double integral formula \eqref{KK2j} still holds with $\Sigma$ and $\Gamma_{\alpha}$ replaced by $\Sigma^{\delta}_{-1/2}$ and $\Gamma_0$ respectively.

  Comparing the integrands of $K^{\rm L}$ (\ref{KK2}) and $K^{\rm J}$ (\ref{KK2j}) we see that the only difference is that in the latter there is an additional factor $\Gamma(w + c_2 + N + 1)/\Gamma(z + c_2 + N + 1)$. For $z$ and $w$ fixed and $N \to \infty$, this simplifies to
  $$
  {\Gamma(w + c_2 + N + 1) \over \Gamma(z + c_2 + N + 1)} = N^{(w - z)} \Big ( 1 + \bigO(1/N) \Big ),
  $$
  so the only difference in the present case is an additional factor of $N^{(w - z)}$. The additional factor is accounted for by the different scalings: $x \mapsto N^{-1/\theta} x$ for the Laguerre case, and $x \mapsto N^{-1-1/\theta} x$ for the Jacobi case.

  In the case that $c_2 \notin \intZ$, the contour $\Gamma_{\alpha}$ in \eqref{KK2j} is infinite. We express $\Gamma_{\alpha} = \Gamma'_{\alpha} \cup \Gamma''_{\alpha}$ and $\Sigma$ a contour from $e^{-(\pi/2 + \delta)i} \cdot \infty$ to $e^{(\pi/2 + \delta)i} \cdot \infty$ in between $\Gamma'_{\alpha}$ and $\Gamma''_{\alpha}$ as in the proof of Proposition \ref{prop:global_densityJ}. Furthermore, we deform $\Gamma''_{\alpha}$ into the infinite contour $\Gamma_0$, and then take $\Sigma$ as $\Sigma^{\delta}_{-1/2}$. We write
  \begin{multline*}
    K^{\mathrm{J}}(x, y) = K^{\mathrm{J},}{}'(x, y) + K^{\mathrm{J},}{}''(x, y), \quad \text{where} \quad K^{\mathrm{J}, *}(x, y) = {1 \over (2 \pi i)^2} {\tilde{h}(x) \over \tilde{h}(y)}\\
    \times \oint_{\Sigma} dz  \oint_{\Gamma^*_\alpha} dw \,
    { x^{-z - 1} y^w \over (z - w) }
    { \Gamma(w + c_2 + N + 1) \Gamma(z+1) \prod_{k=1}^N ( z - \alpha_k) \over
      \Gamma(z + c_2 + N + 1)  \Gamma(w+1) \prod_{l=1}^N (w - \alpha_l)}, \quad
    * = {}' \text{ or }''.
  \end{multline*}
  
  By the argument in the $c_2 \in \intZ$ case, we have that
  \begin{equation} \label{eq:Jacobi_Laguerre_same:1}
    \lim_{N \to \infty} N^{-1 - 1/\theta} K^{\rm J,}{}''( N^{-1 - 1/\theta}x,  N^{-1 - 1/\theta}y) = \lim_{N \to \infty} N^{-1/\theta} K^{\rm L}(N^{-1/\theta}x, N^{-1/\theta} y).
  \end{equation}
  On the other hand, by estimating the factors of the integrand 
  \begin{equation*}
    x^{-z - 1} \Gamma(z+1) \prod_{k=1}^N ( z - \alpha_k) \Gamma(z + c_2 + N + 1)^{-1}
  \end{equation*}
  by Stirling's formula on $\Sigma_{-1/2}$ and
  \begin{equation*}
    y^w \Gamma(w + c_2 + N + 1) \left( \Gamma(w+1) \prod_{l=1}^N (w - \alpha_l) \right)^{-1}
  \end{equation*}
  on $\Gamma'_{\alpha}$ separately, we find that
  \begin{equation}  \label{eq:Jacobi_Laguerre_same:2}
    \lim_{N \to \infty} N^{-1 - 1/\theta} K^{\rm J,}{}'( N^{-1 - 1/\theta}x,  N^{-1 - 1/\theta}y) = 0.
  \end{equation}
  Thus we prove \eqref{eq:Jacobi_Laguerre_same} by combining \eqref{eq:Jacobi_Laguerre_same:1} and \eqref{eq:Jacobi_Laguerre_same:2}.
\end{proof}

\section*{Acknowledgements}
The work of PJF was supported by the Australian Research Council. The work of DW was partially supported by the start-up grant R-146-000-164-133. Thanks are also given to the
Department of Mathematics, National University of Singapore, for hosting a visit during July 2014 when this
work was begun. We thank Arno Kuijlaars and Lun Zhang for spotting misprints in an earlier version. We thank Dongzhou Huang for noticing that the $\theta \in (0, 1)$ case of Proposition \ref{prop:global_density} requires a deformation of $\Sigma$ other than a vertical line.
  

\def\cydot{\leavevmode\raise.4ex\hbox{.}}
\providecommand{\bysame}{\leavevmode\hbox to3em{\hrulefill}\thinspace}
\providecommand{\MR}{\relax\ifhmode\unskip\space\fi MR }
\providecommand{\MRhref}[2]{%
  \href{http://www.ams.org/mathscinet-getitem?mr=#1}{#2}
}
\providecommand{\href}[2]{#2}

\end{document}